\documentclass[english]{article}
\usepackage[T1]{fontenc}
\usepackage[utf8]{inputenc}
\usepackage{geometry}
\usepackage{babel}
\usepackage{float}
\usepackage{textcomp}
\usepackage{amsmath}
\usepackage{amsthm}
\usepackage{graphicx}
\usepackage{setspace}
\usepackage[unicode=true,pdfusetitle,
 bookmarks=true,bookmarksnumbered=false,bookmarksopen=false,
 breaklinks=false,pdfborder={0 0 1},backref=false,colorlinks=false]
 {hyperref}

\makeatletter
\theoremstyle{plain}
\newtheorem{prop}{\protect\propositionname}
\theoremstyle{plain}
\newtheorem{lyxalgorithm}{\protect\algorithmname}

\makeatother

\providecommand{\algorithmname}{Algorithm}
\providecommand{\propositionname}{Proposition}

\begin{document}
\begin{center}
\textbf{\Large{}The Blockchain Risk Parity Line: Moving From The Efficient
Frontier To The Final Frontier Of Investments}{\Large\par}
\par\end{center}

\begin{center}
\textbf{\large{}Ravi Kashyap (ravi.kashyap@stern.nyu.edu)}\footnote{\begin{doublespace}
Numerous seminar participants, particularly at a few meetings of the
econometric society and various finance organizations, provided suggestions
to improve the paper. The following individuals have been a constant
source of inputs and encouragement: Dr. Yong Wang, Dr. Isabel Yan,
Dr. Vikas Kakkar, Dr. Fred Kwan, Dr. Costel Daniel Andonie, Dr. Guangwu
Liu, Dr. Jeff Hong, Dr. Humphrey Tung and Dr. Xu Han at the City University
of Hong Kong. The views and opinions expressed in this article, along
with any mistakes, are mine alone and do not necessarily reflect the
official policy or position of either of my affiliations or any other
agency.
\end{doublespace}
}
\par\end{center}

\begin{center}
\textbf{\large{}Estonian Business School / City University of Hong
Kong}{\large\par}
\par\end{center}

\begin{center}
\today
\par\end{center}

\begin{center}
\par\end{center}

\begin{center}
Journal of Economic Literature Codes: G11 Investment Decisions; G14
Information and Market Efficiency; D81 Criteria for Decision-Making
under Risk and Uncertainty; C32 Time-Series Models; B23 Econometrics,
Quantitative and Mathematical Studies; D8: Information, Knowledge,
and Uncertainty; I31: General Welfare, Well-Being; O3 Innovation •
Research and Development • Technological Change • Intellectual Property
Rights;
\par\end{center}

\begin{center}
under uncertainty in the context of artificial intelligence
\par\end{center}

\begin{doublespace}
\begin{center}
\pagebreak{}
\par\end{center}
\end{doublespace}

\begin{center}
\tableofcontents{}\pagebreak{}
\par\end{center}

\begin{doublespace}
\begin{center}
\listoffigures 
\par\end{center}

\begin{center}
\listoftables 
\par\end{center}
\end{doublespace}

\begin{singlespace}
\begin{center}
\pagebreak{}
\par\end{center}
\end{singlespace}
\begin{doublespace}

\section{Abstract}
\end{doublespace}
\begin{itemize}
\item We engineer blockchain based risk managed portfolios by creating three
funds with distinct risk and return profiles: 1) Alpha - high risk
portfolio ; 2) Beta - mimics the wider market; and 3) Gamma - represents
the risk free rate adjusted to beat inflation. Each of the sub-funds
(Alpha, Beta and Gamma) provides risk parity because the weight of
each asset in the corresponding portfolio is set to be inversely proportional
to the risk derived from investing in that asset. This can be equivalently
stated as equal risk contributions from each asset towards the overall
portfolio risk.
\item We provide detailed mechanics of combining assets - including mathematical
formulations - to obtain better risk managed portfolios. The descriptions
are intended to show how a risk parity based efficient frontier portfolio
management engine - that caters to different risk appetites of investors
by letting each individual investor select their preferred risk-return
combination - can be created seamlessly on blockchain.
\item Any Investor - using decentralized ledger technology - can select
their desired level of risk - or return - and allocate their wealth
accordingly among the sub funds, which balance one another under different
market conditions. This evolution of the risk parity principle - resulting
in a mechanism that is geared to do well under all market cycles -
brings more robust performance and can be termed as conceptual parity.
\item The mechanisms we have used to construct the efficient frontier on
blockchain will ensure that the efficient set keeps evolving to suit
the market environment while allowing all investors to alter their
risk tolerance or return preference. The inclusion of newer and more
diversified assets into the sub-funds - as the crypto landscape expands
- can be viewed as a natural progression from the conventional efficient
frontier to a progressive final frontier of investing, which will
continue to transcend itself. 
\item We have given several numerical examples that illustrate the various
scenarios that arise when combining Alpha, Beta and Gamma to obtain
Parity. 
\item The final investment frontier is now possible - a modification to
the efficient frontier, thus becoming more than a mere theoretical
construct - on blockchain since anyone from anywhere can participate
at anytime to obtain wealth appreciation based on their financial
goals. 
\end{itemize}

\section{\label{sec:Introduction:-The-Efficient-Frontier}Introduction: The
Efficient Frontier - Despite Its Might - Might Be Deficient}

A remarkable idea from the financial markets is that of the efficient
frontier. There are many ways to combine assets to create portfolios.
Among all the possible mixtures of assets that can be created, the
set of combinations that are superior to the rest - in terms of risk
and expected returns - form the efficient frontier (Markowitz 1952;
1959; Merton 1972; Rubinstein 2002; End-note \ref{enu:Efficient-Frontier}). 

Despite the efficient frontier being an intriguing idea - and the
many enhancements done to the basic techniques associated with it
and efforts to simplify its implementation - there are many practical
limitations to accomplish it (Michaud 1989; Green \& Hollifield 1992;
Zhang, Li \& Guo 2018; Kim et al., 2021). To ensure that we are not
constrained by the many reservations of traditional finance portfolio
selection techniques - and to be able to benefit from recent blockchain
technological developments (Nakamoto 2008; Narayanan \& Clark 2017;
Monrat et al., 2019; Schlecht et al., 2021) - our innovation has been
to come up with the idea of conceptual parity tailored for the crypto
environment. 

We engineer blockchain based risk managed portfolios by creating three
funds with distinct risk and return profiles: 1) Alpha - high risk
portfolio ; 2) Beta - mimics the wider market; and 3) Gamma - represents
the risk free rate adjusted to beat inflation. Each of the sub-funds
(Alpha, Beta and Gamma) provides risk parity because the weight of
each asset in the corresponding portfolio is set to be inversely proportional
to the risk derived from investing in that asset (Roncalli 2013; End-note
\ref{enu:Risk-parity-(or}). This can be equivalently stated as equal
risk contributions from each asset towards the overall portfolio risk.

In our modified approach to create risk parity portfolios: Alpha will
be a sub-fund composed of assets that provide higher returns and take
on higher risks. Beta will be representative of the larger market
behavior and provide more steady returns with a correspondingly lower
level of risks. Gamma will take on the role of acting as the risk
free rate, with decent returns but with very little to no risk. Gamma
will also be filled with assets that demonstrate negative correlation
to Alpha and Beta assets while having characteristics that can beat
inflation. Kashyap (2022) has a detailed discussion of the individual
funds and the asset selection procedures to create them.

Any Investor - using decentralized ledger technology - can select
their desired level of risk - or return - and allocate their wealth
accordingly among the sub funds, which balance one another under different
market conditions. This evolution of the risk parity principle - resulting
in a mechanism that is geared to do well under all market cycles -
brings more robust performance and can be termed as conceptual parity.

The mechanisms we have used to construct the efficient frontier on
blockchain will ensure that the efficient set keeps evolving to suit
the market environment while allowing all investors to alter their
risk tolerance or return preference. The inclusion of newer and more
diversified assets into the sub-funds - as the crypto landscape expands
- can be viewed as a natural progression from the conventional efficient
frontier to a progressive final frontier of investing, which will
continue to transcend itself (Krugman 1998; Baldry 1999; Pearson \&
Davies 2014; Brode \& Brode 2015; End-note \ref{enu:Star-Trek-is}). 

Over the past few years - till the end of 2021, at-least - plenty
of people have made money in decentralized finance (Zetzsche et al.,
2020; Harvey et al., 2021; Piñeiro-Chousa et al., 2022; DeFi - End-note
\ref{enu:Decentralized-finance}). The lending interest, yield farming,
liquidity mining, and staking have all generated incredible returns
in DeFi (Xu \& Feng 2022; End-note \ref{enu:Types-Yield-Enhancement-Services}).
Though there have been some wild swings - amidst raging turbulence
and steep market crashes - the blockchain financial innovation has
been real and the impact hard to ignore. That said, it takes a certain
type of person to jump into DeFi. DeFi investing has been embraced
by the intrepid adventurers - the risk-takers who believe in the whole
decentralized philosophy - who have paved the way forward for their
peers. Or, to put it another way, the first generation of DeFi hasn’t
exactly been for everyone.

The goal of this paper - and its sister papers - is intended to change
all that. All the components we have outlined in Kashyap (2022) seek
to build a simpler way to invest in DeFi that takes away much of the
risk while keeping the spectacular gains. We have given several numerical
examples that illustrate the various scenarios that arise when combining
Alpha, Beta and Gamma to obtain Parity. The Risk Parity movement has
been successfully doing just that in the traditional investing world
for many years. This is the first paper to lead the way in bringing
risk parity to multi-chain DeFi.

The final investment frontier is now possible - a modification to
the efficient frontier, thus becoming more than a mere theoretical
construct - on blockchain since anyone from anywhere can participate
at anytime to obtain wealth appreciation based on their financial
goals. 

\subsection{\label{subsec:Outline-of-the}Outline of the Sections Arranged Inline}

Section (\ref{sec:Introduction:-The-Efficient-Frontier}) which we
have already seen, provides a summary of the main problems we are
seeking to solve in the decentralized finance realm. Section (\ref{sec:Review-of-Related})
is a detailed review of the literature related to mean variance optimization
- the efficient frontier (Section \ref{subsec:The-Efficient-Frontier})
- and risk parity portfolio construction (Section \ref{sec:Risk-Parity-Portfolios}).
Section (\ref{sec:Risk-Parity:-The-Kryptonite}) provides intuitive
descriptions of how we are constructing risk parity using blockchain
technology. The motivations for bringing innovations from traditional
finance to the blockchain realm - such as risk parity and the efficient
frontier - are outlined so that this paper is accessible to a wide
audience including finance specialists, quantitative engineers and
technologists.

Section (\ref{sec:Risk-Parity:-Combining}) is a discussion of the
detailed mechanics of combining assets - including mathematical formulations
- to obtain better risk managed portfolios. The descriptions are intended
to show how a risk parity based efficient frontier portfolio management
engine - that caters to different risk appetites of investors by letting
each individual investor select their preferred risk-return combination
- can be created seamlessly on blockchain. The Sub-Sections (\ref{subsec:Parity-Line:-Moving};
\ref{subsec:Distancing-the-Distance}; \ref{subsec:Robust-and-Simple};
\ref{subsec:Parity-Rebalancing}; \ref{subsec:Efficient-Frontier-Parabolic})
in Section (\ref{sec:Risk-Parity:-Combining}) discuss various topics
related to the creation and maintenance of parity portfolios including
the fund flows that need to happen periodically to rebalance investor
allocations of the sub-funds - Alpha, Beta and Gamma - so that users
stay aligned with their risk-return preferences. 

A lot of material - given as Appendices in the Supplementary Material
portion for the electronic online component - can help readers obtain
a deeper understanding of these topics. Section (\ref{sec:Parity-Flow-Flow})
has the flow charts related to the material discussed in Section (\ref{sec:Risk-Parity:-Combining}).
The diagram in Section (\ref{sec:Parity-Flow-Flow}) is given for
completion and for helping readers obtain a better understanding of
the concepts involved. Section (\ref{sec:Numerical-Results}) explains
the numerical results we have obtained, which illustrate how our innovations
compare to existing wealth management techniques. Appendix (\ref{sec:Appendix-of-Illustrations})
has supplementary illustrations and other material that can be useful
to help readers obtain a better understanding of the concepts we have
discussed.

Sections (\ref{sec:Areas-for-Further}; \ref{sec:Conclusion}) suggest
further avenues for improvement and the conclusions respectively.
Sub-Sections (\ref{subsec:Beating-Benchmarks-with}; \ref{subsec:Crypto-Environmental-Nuances};
\ref{subsec:Fortune-Favors-The}; \ref{subsec:Risk-Management-Is};
\ref{subsec:Sharpening-the-Sharpe}; \ref{subsec:Investor-Experience-on};
\ref{subsec:Parity-Fees}) in Section (\ref{sec:Areas-for-Further})
have several topics related to risk management, incorporating better
statistical techniques over time, nuances of managing portfolios in
a decentralized environment, enhancing the investor experience on
blockchain and the parity fee structure. 

\section{\label{sec:Review-of-Related}Review of Related Literature}

\subsection{\label{subsec:The-Efficient-Frontier}The Efficient Frontier}

Several studies related to the efficient frontier - and optimal asset
selection in a portfolio - have been devoted to:
\begin{itemize}
\item finding simplified methods - several of which provide approximate
solutions - of solving for the efficient frontier (Calvo et al., 2012;
Elton et al., 1973; 1976; 1977a; 1977b; 1978; 1979). 
\item heuristic techniques using neural networks, machine learning and other
intelligent computational approaches (Fernández \& Gómez 2007; Gunjan
\& Bhattacharyya 2023);
\item statistical techniques to combat the sensitivity of the model parameters
to estimation errors - and also when return distributions are not
normal - which naturally result in better techniques to estimate expected
returns, variance and the dependence structure of asset prices (Kalymon
1971; Elton \& Gruber 1973; Britten‐Jones 1999; Leland 1999; Goldfarb
\& Iyengar 2003; Ledoit \& Wolf 2003; Zhang \& Nie 2004; Schuhmacher
et al., 2021);
\item establishing upper and lower bounds on the asset holdings and other
parameters (Elton et al., 1977a; Best \& Hlouskova 2000; Chen, He
\& Zhang 2011);
\item alternatives to variance for measuring risk such as Value at Risk
(VAR) or semi-variance in which only returns below expected value
are counted as risk (Nantell \& Price 1979; Duffie \& Pan 1997; Bond
\& Satchell 2002; Ballestero 2005; Jorion 1996; 2007; Yan, Miao \&
Li 2007; Tsao 2010; McNeil et al., 2015; End-notes \ref{enu:The-semivariance-is};
\ref{enu:Value-at-risk});
\item portfolio selection in a fuzzy environment in which risk and return
are characterized as fuzzy variables (Zadeh 1965; Kwakernaak 1978;
Nahmias 1978; Huang 2006; 2007a; 2007b; 2008a; 2008b; Liu \& Zhang
2013; End-note \ref{enu:Fuzzy-logic-is}). Calvo et al., (2016) is
an interesting use of fuzzy techniques to include non-financial goals
- such as social responsibility of the portfolio and being able to
include information pertaining to the costs that arise when deviating
from the financially efficient portfolio - in the portfolio selection
problem.
\item possibility distributions - which grew out of fuzzy approach to tackle
uncertainty - instead of probability distributions for variance and
returns (Zadeh 1978; Rosenthal 2006;Zhang et al., 2007; Dubois \&
Prade 2001; Dubois 2006; End-notes \ref{enu:Possibility-theory-is};
\ref{enu:Probability-theory-or});
\item translating subjective views and judgements on security selection
- using Bayesian frameworks - to the portfolio construction problem
(Mao \& Särndal 1966; Lindley 1972; Treynor \& Black 1973; Black \&
Litterman 1992; Drobetz 2001; Christodoulakis 2002; He \& Litterman
2002; Da Silva et al., 2009; Gelman \& Shalizi 2013; van de Schoot
et al., 2021; End-note \ref{enu:Bayesian-statistics-is}). The use
of Bayesian statistical techniques can be extended to translate prior
degrees of beliefs in asset pricing models - or risk factors that
explain the sources of risk and hence the returns - regarding expected
returns to the resulting portfolios (Pástor 2000);
\item adding more realistic constraints to relax some of the assumptions
- such as transaction costs, short sales, leverage policies and taxes
(Pogue 1970; Zhang \& Wang 2008); 
\item sequential selection of assets in a multi-period setting as opposed
to single-period batch selection (Li \& Hoi 2014; Li et al., 2015); 
\item using a fundamentally different approach - to mean variance optimization
- such as network theory in which securities correspond to the nodes
and the links relate to the correlation of returns (Garas et al.,
2008; Namaki et al., 2011; Kocheturov et al., 2014; Peralta \& Zareei
2016); 
\item incorporating user behavior - and related behavioral insights - into
the portfolio selection process and conducting experiments that investigate
the performance of individuals related to portfolio selection problems
(Hursh 1984; Kroll et al., 1988; Camerer 1999; Thaler 1999; Mullainathan
\& Thaler 2000; Tomer 2007; Hirshleifer 2015; Barberis \& Thaler 2003;
Ritter 2003; Thaler 1999; 2016; 2017; Bi et al., 2018; Momen et al.,
2019; 2020; Zhao et al., 2022; End-note \ref{enu:Behavioral-finance-is});
\item conducting empirical tests of the corresponding portfolio selection
theory and evaluations of the constructed portfolios (Cohen \& Pogue
1967; Fama \& MacBeth 1973; Frost \& Savarino 1986; Constantinides
\& Malliaris 1995; Polson \& Tew 2000; Fama \& French 2004; Palczewski
\& Palczewski 2014).
\end{itemize}

\subsection{\label{sec:Risk-Parity-Portfolios}Risk Parity Portfolios}

Risk Parity can be considered the holy grail - an extremely worthy
pursuit - of investing in the decentralized finance world. To obtain
parity, the amount of money allocated to the individual assets in
a portfolio has to be proportional to the extent of risk encountered
from investing in that specific asset, regardless of its expected
returns (Fabozzi, Simonian \& Fabozzi 2021; End-note \ref{enu:Risk-parity-(or}).
As the risk characteristics of an asset fluctuate, the weight assigned
to that asset has to be correspondingly modified. Risk Parity offers
an alternative to the mean–variance framework, which is based on an
optimization that targets a specific return with a minimal level of
risk or vice versa - a desired level or risk and the maximum return
corresponding to the chosen level or risk (Elton \& Gruber 1997; Fabozzi,
Gupta \& Markowitz 2002; Elton et al., 2009; Buttell 2010; End-note
\ref{enu:Modern-portfolio-theory}). 

Chaves et al., (2011) find that the Risk Parity technique significantly
outperforms optimized allocation strategies such as minimum-variance
and mean–variance efficient portfolios on a consistent basis, though
it does not consistently outperform the equal weighting or a model
pension fund portfolio anchored to the 60/40 equity/bond portfolio
structure (Ambachtsheer 1987; Bender et al., 2010; Harvey et al.,
2018; Konstantinov 2021). Many traditional portfolios - 60/40 asset
allocation or other traditional mean variance optimized portfolios
- are not truly diversified because of their higher allocation to
high-risk assets. As a result, their expected risk-adjusted returns
are low in comparison to risk parity portfolios - which follows the
principle of risk diversification more thoroughly than traditional
asset allocation approaches - achieving both higher risk-adjusted
returns and higher total returns than traditional techniques (Qian
2011). 

Individual asset weights depend on both systematic and idiosyncratic
risk in risk-based portfolios - for example: risk parity, maximum
diversification, and minimum variance portfolio techniques. Risk-parity
portfolios include all investable assets, and idiosyncratic risk has
little effect on weight magnitude, but systematic risk eliminates
many investable assets in long-only, constrained, maximum-diversification,
and minimum-variance portfolios (Clarke, De Silva \& Thorley 2013).
Fisher et al., (2015) find - using very general assumptions and also
by verifying the theoretical results empirically - that the probability
of risk parity beating any other portfolio is more than 50\%.

Hence, the risk parity approach ensures that the weights are relatively
stable - compared to other risk oriented portfolio construction mechanisms
- and also more assets are included in the overall portfolio without
a small proportion of the assets getting an overweight allocation.
Benefitting from the higher risk adjusted returns of safer assets
requires leverage which parity portfolios are also better poised to
take advantage of. Risk parity recommends the application of leverage
to the risk-balanced portfolio - or the market portfolio - to increase
both its expected return and its risk to desired levels (Pogue 1970;
De Souza \& Smirnov 2004; Asness, Frazzini \& Pedersen 2012; Jacobs
\& Levy 2012; Ammann et al., 2016; End-note \ref{enu:In-finance,-leverage}).

By comparing four investment strategies - value weighted, 60/40 fixed
mix, and un-levered and levered risk parity - Anderson et al., (2012)
suggest that risk parity may be a preferred strategy under certain
market conditions - or with respect to certain yardsticks - and they
also show that leverage exacerbates market frictions, which degrade
both return and risk-adjusted return. Bhansali (2011); Roncalli \&
Weisang (2016) analyze methods to achieve portfolio diversification
based on the decomposition of the portfolio’s risk into risk factor
contributions. Exploring the relationship between risk factors and
asset contributions can help to formulate the diversification problem
in terms of an optimization program based on risk factors (Cochrane
2009; Qian 2015; End-notes \ref{enu:In-finance,-risk-factors}; \ref{enu:In-financial-Asset-Pricing}).

Maillard et al., (2010) derive the theoretical properties of risk
parity portfolios - equal risk contribution from each portfolio component
- and show that their volatility is located between those of minimum
variance and equally-weighted portfolios. The set of all risk parity
solutions - solved by using convex optimization techniques - may contain
exponential number of solutions making them extremely hard to compute
(Sahni 1974; Murty \& Kabadi 1985; Cook 2000; Boyd \& Vandenberghe
2004; Fortnow 2009; Bertsekas 2009; 2015; End-notes \ref{enu:Convex-optimization-is};
\ref{enu:The-P-versus-NP}; \ref{enu:In-computational-complexity-NP-Hardness};
\ref{enu:In-computational-complexity-NP}). Bai et al., (2016) propose
an alternative non-convex least-square model which is more efficient
in terms of both speed of computation and accuracy of the result.
Much research is being conducted in terms of implementing risk parity
portfolios using advanced statistical technique (Bellini et al., 2021). 

Qian (2013) examines several risk-parity managers quantitatively using
return-based style analysis and finds that many do not practice true
risk parity. Despite the intuitive appeal of risk parity - since equalizing
estimated risk contributions seems like a good way to achieve the
goal of risk diversification and not having the need to estimate expected
returns, which are very error prone - more theoretical investigations
and empirical verifications are needed before it can be deemed superior
to traditional portfolio construction methods (Timmermann 2008; Thiagarajan
\& Schachter 2011; Rapach \& Zhou 2013). Braga et al., (2023) find
that instead of using variance or volatility, a kurtosis based risk-parity
strategy - that does not seek the minimization of kurtosis, but rather
its ‘fair diversification’ among assets - outperforms the traditional
risk parity according to several prominent risk-adjusted performance
measures. 

The principal components - once they are extracted - driving the variability
of assets in a portfolio can be interpreted as principal portfolios
representing the uncorrelated risk sources inherent in the assets
(Meucci 2009; Frahm \& Wiechers 2011; Shlens 2014; Jolliffe \& Cadima
2016; End-note \ref{enu:Principal-component-analysis}). A well-diversified
portfolio should have its overall risk evenly distributed across the
principal portfolios - or the sources of risk based on the principal
component decomposition. Using this approach the maximum diversification
is obtained from a risk parity strategy that is budgeting risk with
respect to the extracted principal portfolios rather than the underlying
portfolio assets (Lohre et al., 2012; 2014). 

\section{\label{sec:Risk-Parity:-The-Kryptonite}Risk Parity: The Kryptonite
To Alleviate Crypto Market Nightmares}

Alpha, Beta and Gamma are our three main funds with different levels
of risk and expected returns. Alpha will be more risky than Beta,
which will be more risky than Gamma. Investors will be able to combine
the three funds depending on their risk appetites using a suitable
Graphical User Interface (GUI - Figure \ref{fig:Deposit-Screen-Parity};
Martinez 2011; End-note \ref{enu:A-graphical-user-GUI}). Mixing Alpha,
Beta and Gamma will give the Risk Parity portfolio. Risk Parity will
generate returns for investors that factor the risk of the individual
assets, with each asset contributing equally to the overall risk of
the portfolio. The end result will be the first investment vehicle
that will tailor the preferences of each investor - providing diversified
and risk adjusted returns - entirely on a highly secure blockchain
environment.

A subtle aspect of our portfolio construction and the weight calculation
methodology described in Kashyap (2022) is that parity is already
accomplished in each of our individual funds Alpha, Beta and Gamma.
These investment products - Alpha, Beta, Gamma and Parity - will provide
risk managed access to several crypto assets and strategies. We have
adapted many of the well known safety mechanisms and investor protection
schemes that have evolved for several decades in traditional finance,
and combined them with many innovations that are unique to crypto
markets (Kashyap 2023).

Having mentioned that each of our sub funds already achieves risk
parity, we need to draw a distinction between mathematical parity
and conceptual parity. The assets weights are calculated based on
precise rules and mathematical operations and this brings parity to
each of our sub funds at the asset level. While this is still a huge
innovation to bring to the blockchain environment, we wish to proceed
further and bring parity also on a conceptual level.

To elaborate further, we create portfolios that perform satisfactorily
where mathematics can fall short of completely combating market uncertainty.
Broad categories of assets have slightly different risk and return
attributes. By grouping assets with similar responses to different
market regimes, we can ensure that the various groups counterbalance
one another under diverse market conditions. Hence, in addition to
mathematical parity, within each sub fund, each sub fund has an overall
risk return feature which is preferable to the other sub funds under
a particular market criterion.

Another motivation for creating these groups is because even if assets
at the individual level deviate from their risk and expected return
properties, such a misalignment is less likely at the group level.
A few assets in a bunch might display atypical behavior, but the majority
of them will be closer to their representative qualities. The result
is that the overall group can be expected to behave in a certain way
and offset other groups, which are constructed based on the same principle
of clubbing together similar assets, that have different attributes.
We term this fluctuating pseudo-equilibrium between groups of assets
conceptual parity.

The implication of constructing the sub-funds (Alpha, Beta and Gamma)
in this way ensures that when the overall market under performs, which
means Alpha and Beta will not deliver very high returns, Gamma will
still continue to provide acceptable returns because of its negative
correlation to Alpha and Beta. The manufacturing, and linking, of
Alpha, Beta and Gamma will then produce the most efficient set of
portfolios in terms of risk and return characteristics. We term this
collection of portfolios, the parity line.

Risk Parity is a unique advance in cross-chain DeFi asset management.
The Parity index class is where all our other index classes (Alpha,
Beta and Gamma) come together. In addition to the balanced portfolios
created by the other indexes, the Risk Parity strategy provides cover
against universal changes such as higher than expected inflation,
which could otherwise negatively affect most portfolios, both crypto
and traditional.

Risk Parity is an investment movement first pioneered in the All Weather
fund by Ray Dalio at Bridgewater Associates in 1996 (Prince 2011;
Fabozzi, et al. , 2021). It aims to return results better than holding
cash no matter the circumstances. Many investment managers try to
diversify assets but ignore the wider environmental risks, such as
inflation or low economic growth that can have a major impact on the
entire economy. Some examples of such events are: recessions, the
2008 financial crisis or indeed a global pandemic (Reinhart \& Rogoff
2008; Acharya \& Richardson 2009; Ciotti et al., 2020). It is generally
true to a large extent that the wider economic outlook is ‘priced
in’ to assets. The risk materializes when expectations aren’t met,
in either direction. Additional pointers regarding this topic are
mentioned in Section (\ref{sec:Areas-for-Further}).

Risk Parity was designed to enable a portfolio to perform whatever
the environment. Until now, this extraordinarily successful mechanism
from traditional finance could not be applied to crypto because the
tools didn’t exist. This work - and its sister papers - is the first
to engineer the necessary infrastructure, along with the required
foundational principles, to bring such a risk managed portfolio to
DeFi cross-chain assets.

\section{\label{sec:Risk-Parity:-Combining}On Par With Parity: The Mathematics
of Mixing Alpha, Beta and Gamma}

We modify well known techniques in portfolio management theory related
to obtaining portfolios on the efficient frontier. We note that our
problem boils down to being able to combine three funds: Alpha, Beta
and Gamma using a robust and simple methodology to obtain an ideal
fund that balances risk and return. We refer to this combined fund
as the Parity Portfolio. The properties of the individual funds are
discussed in separate articles (Kashyap 2022). The essence of the
properties of the three portfolios are that each fund caters to a
varying risk appetite. In addition, the Gamma fund is constructed
such that it shows negative correlation to Alpha to a certain extent
and also to Beta at times. Hence combining the three funds provides
an all weather portfolio, or a portfolio that outperforms the market
during both upturns and downturns.

The efficient frontier is an extremely elegant approach to finding
excellent combinations of risk and return. To simplify the corresponding
ideas for a blockchain environment and also to make sure the solutions
are more robust, we formulate the ideas as linear approximations.
The linear approximation becomes valid as we construct Gamma portfolios
that show negative correlation to Alpha and Beta. These approximations
can easily be modified in later versions of the application, but the
main idea is that this approach allows the weights corresponding to
specific combinations of risk and return to be calculated on-chain. 

The simplifications we adopt can also enable moving these calculations
onto an on-chain framework quite seamlessly using smart contracts
(Zou, et al., 2019; Zheng, et al., 2020; End-note \ref{enu:A-smart-contract}).
As better blockchain networks develop, we will need to see if the
above techniques we have created need modification. It might be possible
to use thousands of transactions for calculations on a blockchain
computing platform. To emphasize, committing thousands of transactions
to the blockchain record, or into a block, is already possible (Pierro
\& Tonelli 2022). The basics of computing make it clear that the more
data we wish to store and the more computations we need to perform,
the associated costs will increase (Dromey 1982). 

The other reasons for adopting linear approximations are because the
efficient frontier is based on estimates of risk and return, which
can be highly erroneous. So instead of estimating several parameters
to get a curve that acts as the efficient frontier, we fit a line
- through regression or other methodologies at a later stage - across
the set of diversified fund portfolios we have created (Alpha, Beta,
Gamma) over time. That is the as the points Alpha, Beta and Gamma
move over time, they give us a measure of how the Parity Line is changing
with time.

\subsection{\label{subsec:Parity-Line:-Moving}The Parity Line: The Final Frontier
Of Asset Management}

In terms of the three funds, Alpha is the most risky and hence has
the highest expected return. Beta comes second in terms of risk and
Gamma is the least risky of the three. Hence, Alpha also provides
the highest expected returns and so on. If we visualize these three
funds as points on a graph with risk (standard deviation of returns)
on the X-axis and expected return on the Y-axis (Figures \ref{fig:Parity-Line:-ABG-Scatter};
\ref{fig:Parity-Line:-ABG-Real-Market-Data}). Alpha would be the
right most point. Beta would be in the middle and Gamma would be the
leftmost point. 

As expectations change with time and as the market changes, risk and
expected return of the funds will change over time. Hence we need
to visualize the three funds as clouds of points in the right, middle
and left of the graph (Figure \ref{fig:Parity-Line:-ABG-Scatter}).
We can fit a regression line through these points and we will let
investors chose what level of risk and return they want based on the
slope and intercept of this line. This line is a really good estimate
of the best portfolios that investors can chose and we term this the
Parity Line. Since this is a modification of the concept of the efficient
frontier, this Parity Line becomes the final frontier of investments.
We denote this line as:
\begin{equation}
E\left(R\right)=\left(\Theta\right)\left(\sigma\right)+R_{F}\label{eq:Parity-Line}
\end{equation}
Here, $E\left(R\right)$ is the expected return. $\sigma$ is the
risk. $\Theta$ and $R_{F}$ are the slope and the intercept of the
line respectively. Notice that the intercept, $R_{F}$, is also the
risk free rate. \textbf{\textit{We calculate the slope and intercept
separately using regression or other methods and store them in the
smart contract.}}

Any investor can chose a level of risk and his expected return will
be fixed based on the efficient Parity Line. Likewise if someone chooses
a certain level of expected return, the risk they need to bear to
achieve that will be given by the efficient Parity Line. In addition,
investors can also choose the proportion of their wealth they can
allocate to Alpha, Beta and Gamma directly. Based on this allocation
of wealth, we determine the expected return and risk the investors
wealth would bear on the efficient Parity Line. $E\left(R_{\alpha}\right)$
is the expected return and $\sigma_{\alpha}$ is the risk of Alpha.
$E\left(R_{\beta}\right)$ is the expected return and $\sigma_{\beta}$
is the risk of Beta. $E\left(R_{\gamma}\right)$ is the expected return
and $\sigma_{\gamma}$ is the risk of Gamma. The weights invested
into Alpha, Beta and Gamma for investor $i$ are represented by $w_{\alpha}$
, $w_{\alpha}$ and $w_{\alpha}$ respectively. The expected return
and risk are estimated using Equations (\ref{eq:Return-Estimate};
\ref{eq:Volatility-Risk-Estimate}) in Section (\ref{sec:Return-and-Risk}).

For each investor we mint an NFT (Kugler 2021; Wang et al., 2021;
Bao \& Roubaud 2022; End-note \ref{enu:A-non-fungible-token}) that
stores at a minimum the following information: the amount being deposited,
their risk and expected return preferences and the weight allocation
of their wealth into Alpha, Beta and Gamma portfolios. We consider
the three scenarios below:
\begin{enumerate}
\item \label{enu:Investor-Chooses-Risk}Investor Chooses Risk Tolerance:
The investor can choose his / her level of risk tolerance, $\sigma_{i}$,
on a sliding scale. From this, we will calculate the other values
to store in the NFT.
\begin{equation}
E\left(R_{i}\right)=\left(\Theta\right)\left(\sigma_{i}\right)+R_{F}
\end{equation}
The weights to be invested into Alpha, Beta and Gamma are given by
calculating the distance to the Alpha, Beta and Gamma points from
the investors risk and expected return point, $P\equiv\left[\sigma_{i},E\left(R_{i}\right)\right]$,
on the efficient Parity Line. We denote the distance between two points
$A\equiv\left[\sigma_{A},E\left(R_{A}\right)\right]$ and $B\equiv\left[\sigma_{B},E\left(R_{B}\right)\right]$
as $d_{A,B}$ (Rudin 1953; End-note \ref{enu:In-mathematics-Metric-Space}).
\begin{equation}
w_{\alpha}=\frac{\left(\frac{1}{d_{\alpha,P}}\right)}{\left[\left(\frac{1}{d_{\alpha,P}}\right)+\left(\frac{1}{d_{\beta,P}}\right)+\left(\frac{1}{d_{\gamma,P}}\right)\right]}
\end{equation}
\begin{equation}
w_{\beta}=\frac{\left(\frac{1}{d_{\beta,P}}\right)}{\left[\left(\frac{1}{d_{\alpha,P}}\right)+\left(\frac{1}{d_{\beta,P}}\right)+\left(\frac{1}{d_{\gamma,P}}\right)\right]}
\end{equation}
\begin{equation}
w_{\gamma}=\frac{\left(\frac{1}{d_{\gamma,P}}\right)}{\left[\left(\frac{1}{d_{\alpha,P}}\right)+\left(\frac{1}{d_{\beta,P}}\right)+\left(\frac{1}{d_{\gamma,P}}\right)\right]}
\end{equation}
\begin{equation}
d_{A,B}=\sqrt{\left(\sigma_{A}-\sigma_{B}\right)^{2}+\left(E\left(R_{A}\right)-E\left(R_{B}\right)\right)^{2}}
\end{equation}
\begin{equation}
d_{\alpha,P}=\sqrt{\left(\sigma_{\alpha}-\sigma_{i}\right)^{2}+\left(E\left(R_{\alpha}\right)-E\left(R_{i}\right)\right)^{2}}
\end{equation}
Note that the weights satisfy the equation, 
\begin{equation}
w_{\alpha}+w_{\beta}+w_{\gamma}=1
\end{equation}
If any of the distances are zero, the corresponding weight becomes
1. That is, $w_{\alpha}=1$ if $d_{\alpha,P}=0$ and so on.
\item \label{enu:Investor-Chooses-Expected}Investor Chooses Expected Return
Level: The investor can choose his / her level of expected return,
$E\left(R_{i}\right)$, on a sliding scale. From this, we will calculate
the other values to store in the NFT.
\begin{equation}
\sigma_{i}=\frac{E\left(R_{i}\right)-R_{F}}{\left(\Theta\right)}
\end{equation}
The weights are calculated similar to Point (\ref{enu:Investor-Chooses-Risk})
above.
\item \label{enu:Investor-Chooses-Weights}Investor Chooses Alpha, Beta
and Gamma Weights: The investor can choose his allocations to each
of the three funds directly. He will choose two and the third is given
by the identify, $w_{\alpha}+w_{\beta}+w_{\gamma}=1$. From this,
we will calculate the other values to store in the NFT.
\begin{equation}
E\left(R_{i}\right)=\left[w_{\alpha}E\left(R_{\alpha}\right)\right]+\left[w_{\beta}E\left(R_{\beta}\right)\right]+\left[w_{\gamma}E\left(R_{\gamma}\right)\right]\label{eq:Expected-Return}
\end{equation}
\begin{equation}
\sigma_{i}=\left(w_{\alpha}\sigma_{\alpha}\right)^{2}+\left(w_{\beta}\sigma_{\beta}\right)^{2}+\left(w_{\gamma}\sigma_{\gamma}\right)^{2}+2w_{\alpha}\sigma_{\alpha}w_{\beta}\sigma_{\beta}\rho_{\alpha,\beta}+2w_{\beta}\sigma_{\beta}w_{\gamma}\sigma_{\gamma}\rho_{\beta,\gamma}+2w_{\alpha}\sigma_{\alpha}w_{\gamma}\sigma_{\gamma}\rho_{\alpha,\gamma}
\end{equation}
Here, $\rho_{\alpha,\beta}$ represents the correlation between Alpha
and Beta. If the correlation values are not readily available, we
can use our Parity Line to get the value of risk to store in the NFT
using Point (\ref{enu:Investor-Chooses-Expected}) and the expected
return from Equation (\ref{eq:Expected-Return}). This approximation
is valid to a reasonable extent since the Parity Line is the best
combination of risk and expected returns possible and the user's wealth
allocation will result in a point on this line. Figure (\ref{fig:Parity-Line:-ABG-Sample})
shows an example of calculating the combined portfolio return and
risk when the weights of Alpha, Beta and Gamma are calculated using
the volatilities of the funds themselves.
\end{enumerate}

\subsection{\label{subsec:Distancing-the-Distance}Distancing the Distance Function}

The approximation in Point (\ref{enu:Investor-Chooses-Weights}) in
Section (\ref{subsec:Parity-Line:-Moving}) also has a related problem.
If we use the Euclidean distance functions specified in Point (\ref{enu:Investor-Chooses-Risk})
and Section (\ref{subsec:Parity-Line:-Moving}) we note that there
are many possible solutions to go from weights to risk and expected
return as discussed in Point (\ref{enu:Investor-Chooses-Weights}).
Specifically, let us say we chose a particular risk and return combination
and then get the Alpha, Beta and Gamma weights corresponding to this
choice of risk and return. Now, if we use these weights and try to
calculate the risk and the return we may not get the same risk and
return combination we originally started with. This is because the
distance function has many points on the risk-return plane that satisfy
the corresponding constraints on the parity line. 

An alternative approach is to combine Alpha and Beta portfolios with
weights $w_{c\alpha}$ and $w_{c\beta}$ to form a combined portfolio,
or market portfolio, that lies on the parity line. We denote the expected
return and risk of this combined portfolio as $E\left(R_{c}\right)$
and $\sigma_{c}$ respectively. 
\begin{equation}
E\left(R_{c}\right)=\left[w_{c\alpha}E\left(R_{\alpha}\right)\right]+\left[w_{c\beta}E\left(R_{\beta}\right)\right]
\end{equation}
Using the slope and intercept of the parity line we get,
\begin{equation}
\sigma_{c}=\frac{E\left(R_{c}\right)-R_{F}}{\left(\Theta\right)}
\end{equation}
We can then select the weights for Alpha, Beta and Gamma based on
the risk and expected return preferences of any investor, $\sigma_{i}$
and $E\left(R_{i}\right)$. First we get the weights to be invested
in Gamma and the combined portfolio ($w_{c}$) as follows,
\begin{equation}
w_{c}=\frac{\left[E\left(R_{i}\right)-E\left(R_{\gamma}\right)\right]}{\left[E\left(R_{c}\right)-E\left(R_{\gamma}\right)\right]}
\end{equation}
\begin{equation}
w_{\gamma}=\frac{\left[E\left(R_{i}\right)-E\left(R_{c}\right)\right]}{\left[E\left(R_{\gamma}\right)-E\left(R_{c}\right)\right]}
\end{equation}
Note that the weights for Gamma and the combined portfolio can be
negative or greater than one if the user's risk and return preference
lies above the risk and return of the combined portfolio. This condition
can be expressed as,
\begin{equation}
\sigma_{i}>\sigma_{c}\text{ or }E\left(R_{i}\right)>E\left(R_{c}\right)
\end{equation}
We can trim the weights for Gamma and the combined portfolio when
they are negative or greater than one by using suitable bounds. This
gives the weight to be invested in Alpha and Beta as, 
\begin{equation}
w_{\alpha}=\left(w_{c}\right)\left(w_{c\alpha}\right)
\end{equation}
\begin{equation}
w_{\beta}=\left(w_{c}\right)\left(w_{c\beta}\right)
\end{equation}

\begin{prop}
\label{prop:Combined-Weights-Expected-Return}If the weights for Alpha,
Beta and Gamma are chosen as below 
\begin{equation}
w_{\alpha}=\left(w_{c}\right)\left(w_{c\alpha}\right)
\end{equation}
\begin{equation}
w_{\beta}=\left(w_{c}\right)\left(w_{c\beta}\right)
\end{equation}
\begin{equation}
w_{\gamma}=\frac{\left[E\left(R_{i}\right)-E\left(R_{c}\right)\right]}{\left[E\left(R_{\gamma}\right)-E\left(R_{c}\right)\right]}
\end{equation}
We then get the expected return chosen by the investor. That is, 
\begin{equation}
\left[w_{\alpha}E\left(R_{\alpha}\right)\right]+\left[w_{\beta}E\left(R_{\beta}\right)\right]+\left[w_{\gamma}E\left(R_{\gamma}\right)\right]=E\left(R_{i}\right)
\end{equation}
\end{prop}
\begin{proof}
Appendix (\ref{subsec:Parity-Line-Expected}).
\end{proof}
There is still the possibility that multiple combinations of weights
can give the same risk and return on the parity line. But the above
approach ensures that the weights we obtain for any risk and return
chosen will give back the same risk and return. 

Figure (\ref{fig:Parity-Line:-ABG-Real-Market-Data}) has numerical
examples illustrating this methodology of combining Alpha, Beta and
Gamma. Figures (\ref{fig:Parity-Line:-Return-Scenaios}; \ref{fig:Parity-Line:-Risk-Scenarios};
\ref{fig:Parity-Line:-Weights-Return-Scenaios}; \ref{fig:Parity-Line:-Weights-Risk-Scenaios})
show several scenarios where the risk and return of the parity portfolio
are calculated depending on the risk and return of Alpha, Beta and
Gamma with investors choosing their desired level of risk, return
or the weights of the sub-funds. In the scenarios either the investor
choice of portfolio risk or the return are increased in steps of one
percent or the weights are incremented suitably to get portfolio risk
or return increments of one percent.

\subsection{\label{subsec:Robust-and-Simple}Robust and Simple Estimation of
The Parity Line}

Note that the Parity Line will change over time. Hence, as time passes
we will periodically change the slope and intercept of the Parity
Line that is stored in the smart contract (Figure \ref{fig:Parity-Line:-ABG-Scatter}). 

One approach to decide when to seed the smart contract with new values
of the slope and intercept for the Parity Line are to use the volatility
of Alpha, Beta and Gamma and when a new point lies outside a circle
based on the corresponding volatility, it is a good indicator to change
the Parity Line.

A very simple approach to obtain the Parity Line is to run a regression
across different values of return and risk of Alpha, Beta and Gamma
over time. The issue with this approach is that the newer values of
Alpha, Beta and Gamma do not get any additional weight or priority.
To resolve this, we can use advanced regression techniques or take
a weighted average of the values giving higher weight to recent values
and run the regression with these weighted average values.

Other than regression methods, the Parity Line can be chosen as the
line passing through Alpha and Beta. This is for simplicity and also
because a simple regression line will not give more preference to
recent values of Alpha, Beta and Gamma. We are using Alpha and Beta
since by design they are higher than Gamma on the return or Y-axis.
But the Parity Line is the line passing through the two points that
have the highest and second highest return with positive slope. We
emphasize the positive slope since due to market movements and related
anomalies, Beta could be higher than Alpha resulting in negative slope.
In this case, the Parity Line will be the one that passes through
Alpha and Gamma or Beta and Gamma, if Beta is only slightly higher
than Alpha on the return axis.

The main point to note is that the choice of the parameters of the
Parity Line will be done off-chain and the calculated values will
then be stored in the smart contract and updated when deemed necessary.

Though, we need to know what choice the investor made and using that
choice we need to calculate the other values and update the NFT accordingly.
When the weights chosen by the investor and the weights implied by
the Parity Line are different, a simple rebalancing can be performed
(Section \ref{subsec:Parity-Rebalancing}).

\subsection{\label{subsec:Parity-Rebalancing}Parity Rebalancing}

Once an investor makes an investment into Parity and the weights for
Alpha, Beta and Gamma investment are finalized, the corresponding
amounts are invested into Alpha, Beta and Gamma. Hence, based on the
total investment and the individual weights, token quantities specific
to the investment amounts into each of Alpha, Beta and Gamma will
be allocated to the investor. This allocation can be stored in the
NFT or it can also be stored in the smart contract. This will depend
on the pros and cons of adding / storing more metadata into the NFT
and is a technical design consideration. 

The intrinsic value of the NFT, $PARITYINTRINSICVALUE_{it}$, held
by investor $i$ at any time $t$ is the sum of the Alpha, Beta and
Gamma quantities ($ALPHAQTY_{it}$, $BETAQTY_{it},$$GAMMAQTY_{it}$)
allocated to the investor multiplied by the corresponding asset prices
($ALPHAPRICE_{t},$$BETAPRICE_{t}$,$GAMMAPRICE_{t}$). 
\begin{align}
PARITYINTRINSICVALUE_{it} & =\left(ALPHAQTY_{it}\right)\left(ALPHAPRICE_{t}\right)\label{eq:Parity-Intrinsic_Partial}\\
 & +\left(BETAQTY_{it}\right)\left(BETAPRICE_{t}\right)\\
 & +\left(GAMMAQTY_{it}\right)\left(GAMMAPRICE_{t}\right)
\end{align}
There are three primary causes for changing the tokens allocated to
any investors: 
\begin{enumerate}
\item When the prices of the assets change over time, the contribution of
Alpha, Beta and Gamma to the intrinsic value of Parity will change.
Hence, there will be a deviation from the weights chosen by the investor.
A rebalancing mechanism is required that will bring the weights back
to the chosen weights. Note that this adjustment is similar to the
other two cases.
\item When the weights will change due to changes in the user preferences
or
\item Due to changes in the Parity Line.
\end{enumerate}
We consider how to move assets across Alpha, Beta and Gamma based
on the changes due to the above three contributing causes and also
to meet investor deposit and withdrawal requests in Section (\ref{subsec:Parity-Sequence}).
The goal of these fund movements are to keep investor aligned with
their risk and return objectives.

\subsubsection{\label{subsec:Parity-Sequence}Parity Sequence of Steps}

To outline the main sequence of steps, and the corresponding calculations,
for the Parity investment process we proceed with the steps given
in Algorithm (\ref{alg:Algorithm-Parity Sequence of Steps}) as follows:
\begin{lyxalgorithm}
\label{alg:Algorithm-Parity Sequence of Steps}The following algorithm
captures the sequence of steps that need to be carried out at periodic
intervals to ensure that the investments in Alpha, Beta and Gamma
can capture the risk and return preferences and also deposit and withdrawal
requests from all investors entirely on blockchain. The algorithm
also takes of the changes in the allocations - across Alpha, Beta
and Gamma- for the investors based on changes in the market environment
and the risk profiles of the individual funds.
\end{lyxalgorithm}
\begin{enumerate}
\item \label{enu:Parity-Step-One-Fund-Prices}Calculate new Net Asset Value
(NAV: Penman 1970; End-note \ref{EN:Net-Asset-Value}) or Fund Prices
for Alpha, Beta and Gamma based on Step (1) and Step (2) in Kashyap
(2023). 
\item \label{enu:INTRINSIC-PLUS-PENDING}In addition to the variables in
Equation (\ref{eq:Parity-Intrinsic_Partial}), we need to also consider
pending deposits and withdraw requests made by any investor to give
the complete intrinsic value of the Parity investment held by that
investor. 
\begin{enumerate}
\item We next calculate the value of the tokens plus any pending deposits
held by investor $i$ at time $t$, $PARITYDPLUST_{it}$.We then calculate
the value of the withdraw request made by investor $i$ at time $t$,
$PARITYWDRW_{it}$. The deposit amount which has still not been converted
to tokens for investor $i$ at time $t$ is $DEPOSIT_{it}$. The deposit
amount at this phase of the project will be in a stable token or denominated
in USD (Ante, Fiedler \& Strehle 2021; Grobys et al., 2021; Hoang
\& Baur 2021; Lyons \& Viswanath-Natraj 2023; End-note \ref{enu:Stablecoin}).
The withdraw amount will be specified as a percentage of the intrinsic
value held by the investor, $WITHDRAW_{it}$. Note that, $0\leq WITHDRAW_{it}\leq1$.
\begin{align}
PARITYDPLUST_{it} & =DEPOSIT_{it}+\left(ALPHAQTY_{it}\right)\left(ALPHAPRICE_{t}\right)\\
 & +\left(BETAQTY_{it}\right)\left(BETAPRICE_{t}\right)\\
 & +\left(GAMMAQTY_{it}\right)\left(GAMMAPRICE_{t}\right)
\end{align}
\begin{align}
PARITYWDRW_{it} & =WITHDRAW_{it}\left[DEPOSIT_{it}\right.\\
 & +\left(ALPHAQTY_{it}\right)\left(ALPHAPRICE_{t}\right)\\
 & +\left(BETAQTY_{it}\right)\left(BETAPRICE_{t}\right)\\
 & \left.+\left(GAMMAQTY_{it}\right)\left(GAMMAPRICE_{t}\right)\right]
\end{align}
\begin{align}
PARITYWDRW_{it} & =\left(WITHDRAW_{it}\right)\left(PARITYDPLUST_{it}\right)
\end{align}
The value available for future withdraw requests for investor $i$
at time $t$ is given by, $PARITYAVLWDRW_{it}$. 
\begin{align}
PARITYAVLWDRW_{it} & =\left(PARITYDPLUST_{it}\right)\left(1-WITHDRAW_{it}\right)
\end{align}
\item Additional deposits and withdraw requests, made before the sequence
of steps are run, results in the cumulation of the deposit and withdraw
amounts in the pending state. That is, the deposit is a summation
of the deposits from time $T$ to the present time $t$. With a slight
abuse of notation which is helpful while writing this in computer
code wherein the final calculated values can be saved back into the
original variables, we write this as, 
\begin{equation}
DEPOSIT_{it}=\sum_{l=T}^{l=t}DEPOSIT_{il}\quad\begin{cases}
T\leq l\leq t\end{cases}
\end{equation}
Here $T$ is the time when the previous rebalance is completed. All
deposits made by the investor since $T$ are added to the deposit
balance at $T$. Note that, the deposit balance at $T$ can be greater
than zero if a large deposit was made and it was not fully invested
into the sub-funds Alpha, Beta and Gamma when the sequence of steps
are completed at $T$. That is $DEPOSIT_{iT}\geq0$.
\item Likewise, the withdraw amount at time $t$, can be calculated by initially
setting the amount available to withdraw to the sum of the pending
deposit and token values. This is done when the first withdraw request
is made. We use this available amount to calculate the withdraw amounts
from subsequent requests and then cumulate them until time $t$
\begin{align}
PARITYAVLWDRW_{it} & =PARITYDPLUST_{it}
\end{align}
\begin{align}
PARITYWDRW_{it}=PARITYAVLWDRW_{it} & \left(WITHDRAW_{it}\right)
\end{align}
\begin{align}
PARITYAVLWDRW_{it} & =PARITYAVLWDRW_{it}\left(1-WITHDRAW_{it}\right)
\end{align}
\begin{align}
PARITYWDRW_{it}= & PARITYWDRW_{iT}\\
+ & \sum_{l=T}^{l=t}PARITYAVLWDRW_{it}\left(WITHDRAW_{it}\right)\\
 & \quad\begin{cases}
T\leq l\leq t\end{cases}
\end{align}
Here $T$ is the time when the previous rebalance is completed. All
deposits made by the investor since $T$ are added to the deposit
balance at $T$. Note that, the withdraw balance at $T$ can be greater
than zero if a large withdraw request was made and it was not fully
redeemed from the sub-funds Alpha, Beta and Gamma when the sequence
of steps are completed at $T$. That is, $PARITYWDRW_{iT}\geq0$.
\item The deposit and withdraw variables are incremented whenever the user
makes an action. This ensures that the latest deposit or withdraw
variable state is obtained and the gas fees for the operations to
keep these variables updated are paid by the investor as a part of
of the deposit or withdraw transaction.
\end{enumerate}
\item \label{enu:The-deposit-withdraw-user-netting}The deposit amount and
the withdraw amount will be netted for each user. Please note that
the netting can be done along with each deposit or withdraw user action.
It is tempting to do the netting only after the Parity sequence of
steps have commenced and the Alpha, Beta and Gamma prices have been
updated. But since we are only using the cash deposit made by the
user to fulfill his / her withdraw request, the netting can happen
whenever the user makes a deposit or withdraw action. 
\begin{enumerate}
\item \label{enu:Parity-Inflow-Outflow-Investor}We will have a net deposit
or withdraw depending on which cashflow is bigger. That is we check
the condition, $DEPOSIT_{it}\geq PARITYWDRW_{it}$ and calculate the
net inflow $PARITYINFLOW_{it}$ or outflow $PARITYOUTFLOW_{it}$ as
follows,
\begin{align}
PARITYINFLOW_{it} & =\begin{cases}
\begin{array}{cc}
DEPOSIT_{it}-PARITYWDRW_{it}\; & DEPOSIT_{it}\geq PARITYWDRW_{it}\\
0 & DEPOSIT_{it}<PARITYWDRW_{it}
\end{array}\end{cases}
\end{align}
\begin{align}
PARITYOUTFLOW_{it} & =\begin{cases}
\begin{array}{cc}
0 & DEPOSIT_{it}\geq PARITYWDRW_{it}\\
PARITYWDRW_{it}-DEPOSIT_{it}\; & DEPOSIT_{it}<PARITYWDRW_{it}
\end{array}\end{cases}
\end{align}
Or this can be written alternatively as,
\begin{align}
PARITYINFLOW_{it} & =\max\left(DEPOSIT_{it}-PARITYWDRW_{it},0\right)\label{eq:Parity-Inflow-Investor}
\end{align}
\begin{align}
PARITYOUTFLOW_{it} & =\max\left(PARITYWDRW_{it}-DEPOSIT_{it},0\right)\label{eq:Parity-Outflow-Investor}
\end{align}
Note that only one variable, $PARITYINFLOW_{it}$ or $PARITYOUTFLOW_{it}$
can be positive and the other variable has to be zero. While writing
this in computer code, the inflow and outflow values can be saved
back into the original variables. See sub-point (\textbf{\ref{enu:Parity-In-Out-Investor-Comment}})
for additional clarifications.
\item \textbf{\label{enu:Parity-In-Out-Investor-Comment}The Equations (\ref{eq:Parity-Inflow-Investor};
\ref{eq:Parity-Outflow-Investor}) in sub-point (\ref{enu:Parity-Inflow-Outflow-Investor})
need to be modified, to include the cash received from a previous
rebalancing for the investor along with the deposits, to get the full
inflow amount. }This is discussed in sub-point (\ref{enu:Cash-Allocation-Investor})
and given in Equations (\ref{eq:Parity-Inflow-Investor-Cash-Rebalance};
\ref{eq:Parity-Outflow-Investor-Cash-Rebalance}). The aggregate level
(or across investors) cash allocations are discussed in sub-point
(\ref{enu:Cash-Allocation-Aggregate}) and given in Equations (\ref{eq:Parity-Inflow-Aggregate-Positive-Rebalance};
\ref{eq:Parity-Outflow-Aggregate-Positive-Rebalance}). Also sub-points
(\ref{enu:PA-Inflow-Outflow}; \ref{enu:Parity-Fund-Level-Netting})
for aggregate (across investors) level calculations.
\item Using these variables gives the intrinsic value of Parity that includes
deposits and withdraw requests. The intrinsic value of the NFT, $PARITYINTRINSICVALUE_{it}$,
held by investor $i$ at any time $t$ is the sum of the Alpha, Beta
and Gamma quantities ($ALPHAQTY_{it}$, $BETAQTY_{it},$ $GAMMAQTY_{it}$)
allocated to the investor multiplied by the corresponding asset prices
($ALPHAPRICE_{t},$ $BETAPRICE_{t}$, $GAMMAPRICE_{t}$) including
the inflow and outflow of funds, $PARITYINFLOW_{it}$ or $PARITYOUTFLOW_{it}$.
\begin{align}
PARITYINTRINSICVALUE_{it} & =PARITYINFLOW_{it}\label{eq:Parity-Intrinsic_Complete}\\
 & +\left(ALPHAQTY_{it}\right)\left(ALPHAPRICE_{t}\right)\\
 & +\left(BETAQTY_{it}\right)\left(BETAPRICE_{t}\right)\\
 & +\left(GAMMAQTY_{it}\right)\left(GAMMAPRICE_{t}\right)\\
 & -PARITYOUTFLOW_{it}
\end{align}
\end{enumerate}
\item \label{enu:Weights-ABG-Parity-Investments}The weights of Alpha, Beta
and Gamma, for investor $i$ at time $t$ are ($ALPHAWGT_{it}$, $BETAWGT_{it}$,
$GAMMAWGT_{it}$). Based on these weights we can calculate the quantities
that are supposed to be invested into Alpha, Beta and Gamma, ($INVALPHAQT$,
$INVBETAQTY_{it}$, $INVGAMMAQTY_{it}$), for investor $i$ at time
$t$.
\begin{align}
INVALPHAQTY_{it} & =\frac{PARITYINTRINSICVALUE_{it}\left(ALPHAWGT_{it}\right)}{\left(ALPHAPRICE_{t}\right)}-ALPHAQTY_{it}
\end{align}
\begin{align}
INVBETAQTY_{it} & =\frac{PARITYINTRINSICVALUE_{it}\left(BETAWGT_{it}\right)}{\left(BETAPRICE_{t}\right)}-BETAQTY_{it}
\end{align}
\begin{align}
INVGAMMAQTY_{it} & =\frac{PARITYINTRINSICVALUE_{it}\left(GAMMAWGT_{it}\right)}{\left(GAMMAPRICE_{t}\right)}-GAMMAQTY_{it}
\end{align}

\begin{enumerate}
\item Clearly the above can be written in dollar notational values ($INVALPHAAMNT_{it}$,
$INVBETAAMNT_{it}$, $INVGAMMAAMNT_{it}$) to be invested into Alpha,
Beta and Gamma for investor $i$ at time $t$ as,
\begin{align}
INVALPHAAMNT_{it} & =PARITYINTRINSICVALUE_{it}\left(ALPHAWGT_{it}\right)\\
 & -ALPHAQTY_{it}\left(ALPHAPRICE_{t}\right)
\end{align}
\begin{align}
INVBETAAMNT_{it} & =PARITYINTRINSICVALUE_{it}\left(BETAWGT_{it}\right)\\
 & -BETAQTY_{it}\left(BETAPRICE_{t}\right)
\end{align}
\begin{align}
INVGAMMAAMNT_{it} & =PARITYINTRINSICVALUE_{it}\left(GAMMAWGT_{it}\right)\\
 & -GAMMAQTY_{it}\left(GAMMAPRICE_{t}\right)
\end{align}
\item Since the deposit amounts are in USD and the withdraw quantities are
denominated in tokens we can separate these out, for each investor
$i$ at time $t$, in USD investment amounts and token withdraw quantities.
This also ensures we are only storing positive values depending on
whether we are making a net deposit or withdraw as shown below,
\begin{align}
INVALPHAAMNT_{it} & =\max\left[PARITYINTRINSICVALUE_{it}\left(ALPHAWGT_{it}\right)\right.\\
 & -\left.ALPHAQTY_{it}\left(ALPHAPRICE_{t}\right),0\right]
\end{align}
\begin{align}
WDRWALPHAQTY_{it} & =\max\left[ALPHAQTY_{it}-\frac{PARITYINTRINSICVALUE_{it}\left(ALPHAWGT_{it}\right)}{\left(ALPHAPRICE_{t}\right)},0\right]
\end{align}
($WDRWALPHAQTY_{it}$, $WDRWBETAQTY_{it}$, $WDRWGAMMAQTY_{it}$)
are the quantities to be withdrawn from Alpha, Beta and Gamma for
investor $i$ at time $t$. Notice that only one of $INVALPHAAMNT_{it}$
or $WDRWALPHAQTY_{it}$ can be non-zero (and positive). The same applies
to the other pairs of variables for Beta and Gamma. 
\begin{align}
INVBETAAMNT_{it} & =\max\left[PARITYINTRINSICVALUE_{it}\left(BETAWGT_{it}\right)\right.\\
 & -\left.BETAQTY_{it}\left(BETAPRICE_{t}\right),0\right]
\end{align}
\begin{align}
WDRWBETAQTY_{it} & =\max\left[BETAQTY_{it}-\frac{PARITYINTRINSICVALUE_{it}\left(BETAWGT_{it}\right)}{\left(BETAPRICE_{t}\right)},0\right]
\end{align}
\begin{align}
INVGAMMAAMNT_{it} & =\max\left[PARITYINTRINSICVALUE_{it}\left(GAMMAWGT_{it}\right)\right.\\
 & -\left.GAMMAQTY_{it}\left(GAMMAPRICE_{t}\right),0\right]
\end{align}
\begin{align}
WDRWGAMMAQTY_{it} & =\max\left[GAMMAQTY_{it}-\frac{PARITYINTRINSICVALUE_{it}\left(GAMMAWGT_{it}\right)}{\left(GAMMAPRICE_{t}\right)},0\right]
\end{align}
\item Also the following simple identity must be satisfied when the deposit
and withdraw values are combined with the amounts being moved across
funds (rebalancing) so that the investor stays aligned with their
Alpha, Beta and Gamma weights.
\begin{align}
PARITYINFLOW_{it}+\left(ALPHAPRICE_{t}\right)WDRWALPHAQTY_{it}\label{eq:Parity-Fund-Flow-Identity-Investor}\\
+\left(BETAPRICE_{t}\right)WDRWBETAQTY_{it}+\left(GAMMAPRICE_{t}\right)WDRWGAMMAQTY_{it}\\
=\\
PARITYOUTFLOW_{it}+INVALPHAAMNT_{it}\\
+INVBETAAMNT_{it}+INVGAMMAAMNT_{it}
\end{align}
\end{enumerate}
\item \label{enu:Parity-Aggregate}We then aggregate the total deposit amounts
and withdraw quantities across all investors. There are two possible
groups to aggregate. 
\begin{enumerate}
\item \label{enu:First-all-users}First, across all users. This will be
done after a parity line change or a global rebalance for everyone. 
\item \textbf{Second, the other group of users will be those that have a
non zero deposit or withdraw. }This second group of users is used
on a periodic basis for getting the total amounts to deposit or withdraw
from Alpha, Beta and Gamma and also to allocate what is received from
the sub-funds. 
\item Clearly, any group of users can be chosen for the aggregation depending
on the situation dictated by gas fees or other considerations. This
applies to sub-point (\ref{enu:First-all-users}) as well if there
are a large number of users and we need to rebalance them in batches. 
\item Another way to prioritize users to rebalance would be to check their
total amount for rebalancing and order based on the users with the
highest amounts for rebalancing. The amount for rebalancing will be
given by the sum of the difference between the actual allocation and
expected allocation across all three funds: Alpha, Beta and Gamma
(ABG). We denote the total amount to rebalance for investor $i$ at
time $t$ , $TOTALRBLNC_{it}$,by the following,
\begin{align}
TOTALRBLNC_{it} & =PARITYINFLOW_{it}\\
 & +\left(ALPHAPRICE_{t}\right)WDRWALPHAQTY_{it}\\
 & +\left(BETAPRICE_{t}\right)WDRWBETAQTY_{it}\\
 & +\left(GAMMAPRICE_{t}\right)WDRWGAMMAQTY_{it}\\
 & +PARITYOUTFLOW_{it}+INVALPHAAMNT_{it}\\
 & +INVBETAAMNT_{it}+INVGAMMAAMNT_{it}
\end{align}
Note here that we also include deposits and withdraws pending since
that is a complete treatment of the total fund flows required for
investor $i$ at time $t$ .
\item \label{enu:PA-Inflow-Outflow}If we let $IC_{t}$ denote the total
number or count of investors being considered for aggregation at time
$t$ the below variables, prefixed with $PA...$ for Parity Aggregate,
represent the totals across all investors,
\begin{equation}
PAINVALPHAAMNT_{t}=\sum_{i=1}^{i=IC_{t}}INVALPHAAMNT_{it}
\end{equation}
\begin{align}
PAWDRWALPHAQTY_{t} & =\sum_{i=1}^{i=IC_{t}}WDRWALPHAQTY_{t}
\end{align}
\begin{align}
PAINVBETAAMNT_{t} & =\sum_{i=1}^{i=IC_{t}}INVBETAAMNT_{it}
\end{align}
\begin{align}
PAWDRWBETAQTY_{t} & =\sum_{i=1}^{i=IC_{t}}WDRWBETAQTY_{it}
\end{align}
\begin{align}
PAINVGAMMAAMNT_{t} & =\sum_{i=1}^{i=IC_{t}}INVGAMMAAMNT_{it}
\end{align}
\begin{align}
PAWDRWGAMMAQTY_{t} & =\sum_{i=1}^{i=IC_{t}}WDRWGAMMAQTY_{it}
\end{align}
The net fund flows across all investors in terms of deposits and withdraw
requests is given by,
\begin{equation}
PARITYINFLOW_{t}=DEPOSIT_{t}=\sum_{i=1}^{i=IC_{t}}DEPOSIT_{it}=\sum_{i=1}^{i=IC_{t}}PARITYINFLOW_{it}\label{eq:Inflow-Parity-Aggregate}
\end{equation}
\begin{equation}
PARITYOUTFLOW_{t}=PARITYWDRW_{t}=\sum_{i=1}^{i=IC_{t}}PARITYWDRW_{it}=\sum_{i=1}^{i=IC_{t}}PARITYOUTFLOW_{it}\label{eq:Outflow-Parity-Aggregate}
\end{equation}
\textbf{Note that we need to include any cash received from a previous
rebalancing for all the investors along with the deposits to get the
full inflow amount. }These aggregate (across investors) level cash
inclusions are discussed in sub-point (\ref{enu:Cash-Allocation-Aggregate})
and given in Equations (\ref{eq:Parity-Inflow-Aggregate-Positive-Rebalance};
\ref{eq:Parity-Outflow-Aggregate-Positive-Rebalance}). The investor
level cash allocations are discussed in sub-point (\ref{enu:Cash-Allocation-Investor})
and given in Equations (\ref{eq:Parity-Inflow-Investor-Cash-Rebalance};
\ref{eq:Parity-Outflow-Investor-Cash-Rebalance}).See sub-point (\ref{enu:Parity-Fund-Level-Netting})
for aggregate level calculations such that only the investment or
withdraw side is positive. See sub-points (\ref{enu:Parity-Inflow-Outflow-Investor};
\ref{enu:Parity-In-Out-Investor-Comment}) for additional clarifications
related to the investor level deposit and cash inclusions.
\item \label{enu:Parity-Fund-Level-Netting}Notice that we net again at
the fund level (Parity Aggregate) such that only one of $PAINVALPHAAMNT_{t}$
or $PAWDRWALPHAQTY_{t}$ can be non-zero (and positive). See sub-point
(\ref{enu:PA-Inflow-Outflow}) for additional context. The same applies
to the other pairs of variables for Beta and Gamma. We highlight here
that the modified values that include pending amounts need to be considered
as shown in Point (\ref{enu:Pending-Parity-Deposit-Withdraw}). The
calculations in this step are helpful to arrive at the Parity fund
flow identity in Point (\ref{enu:Parity-Fund-Flow-Identity}) and
the allocations in Point (\ref{enu:Allocations-Parity}).
\begin{align}
PAINVALPHAAMNT_{t} & =\max\left[\sum_{i=1}^{i=IC_{t}}INVALPHAAMNT_{it}\right.\\
 & -\left.\left(ALPHAPRICE_{t}\right)\sum_{i=1}^{i=IC_{t}}WDRWALPHAQTY_{it},0\vphantom{\sum_{i=1}^{i=IC_{t}}}\right]
\end{align}
\begin{align}
PAWDRWALPHAQTY_{t} & =\max\left[\sum_{i=1}^{i=IC_{t}}WDRWALPHAQTY_{it}\right.\\
 & -\left.\frac{1}{\left(ALPHAPRICE_{t}\right)}\sum_{i=1}^{i=IC_{t}}INVALPHAAMNT_{it},0\vphantom{\sum_{i=1}^{i=IC_{t}}}\right]
\end{align}
\begin{align}
PAINVBETAAMNT_{t} & =\max\left[\sum_{i=1}^{i=IC_{t}}INVBETAAMNT_{it}\right.\\
 & -\left.\left(BETAPRICE_{t}\right)\sum_{i=1}^{i=IC_{t}}WDRWBETAQTY_{it},0\vphantom{\sum_{i=1}^{i=IC_{t}}}\right]
\end{align}
\begin{align}
PAWDRWBETAQTY_{t} & =\max\left[\sum_{i=1}^{i=IC_{t}}WDRWBETAQTY_{it}\right.\\
 & -\left.\frac{1}{\left(BETAPRICE_{t}\right)}\sum_{i=1}^{i=IC_{t}}INVBETAAMNT_{it},0\vphantom{\sum_{i=1}^{i=IC_{t}}}\right]
\end{align}
\begin{align}
PAINVGAMMAAMNT_{t} & =\max\left[\sum_{i=1}^{i=IC_{t}}INVGAMMAAMNT_{it}\right.\\
 & -\left.\left(GAMMAPRICE_{t}\right)\sum_{i=1}^{i=IC_{t}}WDRWGAMMAQTY_{it},0\vphantom{\sum_{i=1}^{i=IC_{t}}}\right]
\end{align}
\begin{align}
PAWDRWGAMMAQTY_{t} & =\max\left[\sum_{i=1}^{i=IC_{t}}WDRWGAMMAQTY_{it}\right.\\
 & -\left.\frac{1}{\left(GAMMAPRICE_{t}\right)}\sum_{i=1}^{i=IC_{t}}INVGAMMAAMNT_{it},0\vphantom{\sum_{i=1}^{i=IC_{t}}}\right]
\end{align}
The inflow and outflow equations become,
\begin{align}
PARITYINFLOW_{t} & =\max\left(\sum_{i=1}^{i=IC_{t}}\left[PARITYINFLOW_{it}-PARITYOUTFLOW_{it}\right],0\right)\label{eq:Parity-Inflow-Aggregate-Positive}
\end{align}
\begin{align}
PARITYOUTFLOW_{t} & =\max\left(\sum_{i=1}^{i=IC_{t}}\left[PARITYOUTFLOW_{it}-PARITYINFLOW_{it}\right],0\right)\label{eq:Parity-Outflow-Aggregate-Positive}
\end{align}
\textbf{Note that we need to include any cash received from a previous
rebalancing for all the investors along with the deposits to get the
full inflow amount. }These aggregate (across investors) level cash
inclusions are discussed in sub-point (\ref{enu:Cash-Allocation-Aggregate})
and given in Equations (\ref{eq:Parity-Inflow-Aggregate-Positive-Rebalance};
\ref{eq:Parity-Outflow-Aggregate-Positive-Rebalance}). The investor
level cash allocations are discussed in sub-point (\ref{enu:Cash-Allocation-Investor})
and given in Equations (\ref{eq:Parity-Inflow-Investor-Cash-Rebalance};
\ref{eq:Parity-Outflow-Investor-Cash-Rebalance}).
\item \label{enu:Parity-Fund-Flow-Identity}Another simple identity, similar
to the investor fund flow identity (Equation \ref{eq:Parity-Fund-Flow-Identity-Investor}),
can be arrived at for the parity aggregate level as shown below,
\begin{align}
PARITYINFLOW_{t}+\left(ALPHAPRICE_{t}\right)PAWDRWALPHAQTY_{t}\label{eq:Parity-Fund-Flow-Identity-Aggregate}\\
+\left(BETAPRICE_{t}\right)PAWDRWBETAQTY_{t}+\left(GAMMAPRICE_{t}\right)PAWDRWGAMMAQTY_{t}\\
=\\
PARITYOUTFLOW_{t}+PAINVALPHAAMNT_{t}\\
+PAINVBETAAMNT_{t}+PAINVGAMMAAMNT_{t}
\end{align}
Figure (\ref{fig:Parity-Sequence-Steps-Aggregate-I-II-III}) shows
that the Parity Aggregate identity is satisfied based on the scenarios
shown in the corresponding investor illustrations.
\item \label{enu:Pending-Parity-Deposit-Withdraw}We need to calculate the
actual excess values we have to invest into or withdraw, from Alpha,
Beta and Gamma, which modify the equations in Point (\ref{enu:Parity-Fund-Level-Netting}).
To do this, we need to consider the pending deposits, denominated
in USD, or withdraws, denominated in number of tokens, made by Parity
into Alpha, Beta and Gamma while calculating the netted total amounts:
$PDGALPHADPST_{t}$, $PDGBETADPST_{t}$, $PDGGAMMADPST_{t}$ are deposits
denominated in USD and $PDGALPHAWDRW_{t}$, $PDGBETAWDRW_{t}$, $PDGGAMMAWDRW_{t}$
are withdraw requests denominated in number of tokens. Notice that
the additional amount we need to invest or withdraw has to be decreased
by the pending amount if the pending value is smaller otherwise we
do not need to invest or withdraw any additional amounts. A cancel
of the excess amount invested or being withdrawn can be made. Though
currently investors can only cancel their entire withdraw. But special
provisions for Parity can be made at a later stage.
\begin{align}
INVALPHAAMNT_{t} & =\max\left[\max\left\{ \left(\sum_{i=1}^{i=IC_{t}}INVALPHAAMNT_{it}-PDGALPHADPST_{t}\right),0\right\} \right.\\
- & \left(ALPHAPRICE_{t}\right)\\
 & \left.\max\left\{ \left(\sum_{i=1}^{i=IC_{t}}WDRWALPHAQTY_{it}-PDGALPHAWDRW_{t}\right),0\right\} ,0\right]
\end{align}
\begin{align}
WDRWALPHAQTY_{t} & =\max\left[\max\left\{ \left(\sum_{i=1}^{i=IC_{t}}WDRWALPHAQTY_{it}-PDGALPHAWDRW_{t}\right),0\right\} \right.\\
- & \frac{1}{\left(ALPHAPRICE_{t}\right)}\\
 & \left.\max\left\{ \left(\sum_{i=1}^{i=IC_{t}}INVALPHAAMNT_{it}-PDGALPHADPST_{t}\right),0\right\} ,0\right]
\end{align}
\begin{align}
INVBETAAMNT_{t} & =\max\left[\max\left\{ \left(\sum_{i=1}^{i=IC_{t}}INVBETAAMNT_{it}-PDGBETADPST_{t}\right),0\right\} \right.\\
- & \left(BETAPRICE_{t}\right)\\
 & \left.\max\left\{ \left(\sum_{i=1}^{i=IC_{t}}WDRWBETAQTY_{it}-PDGBETAWDRW_{t}\right),0\right\} ,0\right]
\end{align}
\begin{align}
WDRWBETAQTY_{t} & =\max\left[\max\left\{ \left(\sum_{i=1}^{i=IC_{t}}WDRWBETAQTY_{it}-PDGBETAWDRW_{t}\right),0\right\} \right.\\
- & \frac{1}{\left(BETAPRICE_{t}\right)}\\
 & \left.\max\left\{ \left(\sum_{i=1}^{i=IC_{t}}INVBETAAMNT_{it}-PDGBETADPST_{t}\right),0\right\} ,0\right]
\end{align}
\begin{align}
INVGAMMAAMNT_{t} & =\max\left[\max\left\{ \left(\sum_{i=1}^{i=IC_{t}}INVGAMMAAMNT_{it}-PDGGAMMADPST_{t}\right),0\right\} \right.\\
- & \left(GAMMAPRICE_{t}\right)\\
 & \left.\max\left\{ \left(\sum_{i=1}^{i=IC_{t}}WDRWGAMMAQTY_{it}-PDGGAMMAWDRW_{t}\right),0\right\} ,0\right]
\end{align}
\begin{align}
WDRWGAMMAQTY_{t} & =\max\left[\max\left\{ \left(\sum_{i=1}^{i=IC_{t}}WDRWGAMMAQTY_{it}-PDGGAMMAWDRW_{t}\right),0\right\} \right.\\
- & \frac{1}{\left(GAMMAPRICE_{t}\right)}\\
 & \left.\max\left\{ \left(\sum_{i=1}^{i=IC_{t}}INVGAMMAAMNT_{it}-PDGGAMMADPST_{t}\right),0\right\} ,0\right]
\end{align}
\item The separation of deposit and investment values need not be done at
the investor level. But the separation can be done in this step while
aggregating across the group of users we have chosen, so that we do
not have to deal with negative quantities, by checking conditions
such as $\frac{\left(PARITYINTRINSICVALUE_{it}\right)\left(ALPHAWGT_{it}\right)}{\left(ALPHAPRICE_{t}\right)}\geq ALPHAQTY_{it}$.
\item \textbf{There is an important point to be discussed here. }The totals
calculated here are good estimates of the exact amounts of USD, Alpha,
Beta and Gamma tokens we will need when we actually receive them from
the sub-funds. The calculations, which are estimates, are only used
to invest and redeem from Alpha, Beta and Gamma. The actual allocations
in a later step (\ref{enu:Allocations-Parity}) are based on what
will be actually received from Alpha, Beta and Gamma and what the
current weights and requirements of investors are. What we receive
could be different from what was requested due to various reasons
such as more deposits or withdraws being made, changes in preferences
and also limits on the amounts that can be deposited or withdrawn
from Alpha, Beta and Gamma.
\end{enumerate}
\item \label{enu:Using-Parity-Aggregate}Using the aggregate totals across
all investors, calculated in Point (\ref{enu:Pending-Parity-Deposit-Withdraw}),we
invest or withdraw from Alpha, Beta and Gamma (ABG) accordingly. The
USD that is available for deposit, $AVLDEPOSIT_{t}$, after taking
away the pending deposits, needs to be split among the investment
requirements across Alpha, Beta and Gamma. Clearly there are several
ways in which this can be done. 
\begin{align}
AVLDEPOSIT_{t} & =\max\left[PARITYINFLOW_{t}-PDGALPHADPST_{t}\right.\\
 & \left.-PDGBETADPST_{t}-PDGGAMMADPST_{t},0\right]
\end{align}

\begin{enumerate}
\item \label{enu:A-simple-rule-allocation}A simple rule that minimizes
transactions, so that we invest as much as possible into one fund
before investing in another fund, can be implemented as follows:
\begin{equation}
INVALPHAAMNT_{t}=\begin{cases}
\begin{array}{cc}
\left(INVALPHAAMNT_{t}\right)\min\left(\frac{AVLDEPOSIT_{t}}{INVALPHAAMNT_{t}},1\right) & \;INVALPHAAMNT_{t}>0\\
0 & \;INVALPHAAMNT_{t}=0
\end{array}\end{cases}
\end{equation}
The cash available for the next fund (Beta) is decremented and a similar
ratio is used to calculate the deposit to be made in Beta,
\begin{equation}
AVLDEPOSIT_{t}=AVLDEPOSIT_{t}-INVALPHAAMNT_{t}
\end{equation}
\begin{equation}
INVBETAAMNT_{t}=\begin{cases}
\begin{array}{cc}
\left(INVBETAAMNT_{t}\right)\min\left(\frac{AVLDEPOSIT_{t}}{INVBETAAMNT_{t}},1\right) & \;INVBETAAMNT_{t}>0\\
0 & \;INVBETAAMNT_{t}=0
\end{array}\end{cases}
\end{equation}
\begin{equation}
AVLDEPOSIT_{t}=AVLDEPOSIT_{t}-INVBETAAMNT_{t}
\end{equation}
The cash remaining is deposited into the last fund (Gamma) if it is
greater than the amount to be invested in Gamma otherwise we use a
condition similar to the earlier funds, as shown below,
\begin{equation}
INVGAMMAAMNT_{t}=\begin{cases}
\begin{array}{cc}
0 & INVGAMMAAMNT_{t}=0\\
AVLDEPOSIT_{t} & INVGAMMAAMNT_{t}\geq AVLDEPOSIT_{t}\\
\left\{ \vphantom{*\min\left(\frac{DEPOSIT_{t}}{INVGAMMAAMNT_{t}},1\right)}\left(INVGAMMAAMNT_{t}\right)\right. & INVGAMMAAMNT_{t}<AVLDEPOSIT_{t}\\
\left.*\min\left(\frac{AVLDEPOSIT_{t}}{INVGAMMAAMNT_{t}},1\right)\right\} 
\end{array}\end{cases}
\end{equation}
\item The issue with the simple rule above in sub-point (\ref{enu:A-simple-rule-allocation}),
that fills the first fund and moves to the next, is that the fund
prices will fluctuate and it would mean higher costs for people that
invest in the later funds. This issue arises especially when prices
are volatile compared to the frequency at which the sequence of steps
are run. Hence we use the investment need for each fund after adjusting
for pending deposits, as a proportion of the total across all three
funds, $TOTALINVABGAMNT_{t}$, made as below,
\begin{align}
TOTALINVABGAMNT_{t} & =INVALPHAAMNT_{t}+INVBETAAMNT_{t}\\
 & +INVGAMMAAMNT_{t}
\end{align}
\begin{equation}
INVALPHAAMNT_{t}=\begin{cases}
\begin{array}{cc}
\left(AVLDEPOSIT_{t}\right)\left(\frac{INVALPHAAMNT_{t}}{TOTALINVABGAMNT_{t}}\right) & \;TOTALINVABGAMNT_{t}>0\\
0 & \;TOTALINVABGAMNT_{t}=0
\end{array}\end{cases}
\end{equation}
\begin{equation}
INVBETAAMNT_{t}=\begin{cases}
\begin{array}{cc}
\left(AVLDEPOSIT_{t}\right)\left(\frac{INVBETAAMNT_{t}}{TOTALINVABGAMNT_{t}}\right) & \;TOTALINVABGAMNT_{t}>0\\
0 & \;TOTALINVABGAMNT_{t}=0
\end{array}\end{cases}
\end{equation}
\begin{equation}
INVGAMMAAMNT_{t}=\begin{cases}
\begin{array}{cc}
\left(AVLDEPOSIT_{t}\right)\left(\frac{INVGAMMAAMNT_{t}}{TOTALINVABGAMNT_{t}}\right) & \;TOTALINVABGAMNT_{t}>0\\
0 & \;TOTALINVABGAMNT_{t}=0
\end{array}\end{cases}
\end{equation}
\item A combination of the above two rules would check if the value of the
investment, to be sent in a particular transaction to one of the three
funds (ABG), is above a certain minimum threshold. Hence, this approach
considers the investment need of each fund adjusted for pending deposits
combined with a minimum amount for each transaction. Such a minimum
amount can also be set for the withdraw quantities. This ensures that
we do not make very small deposit and withdraw requests.
\end{enumerate}
\item \label{enu:Allocations-Parity}When tokens or cash are received from
Alpha, Beta and Gamma they are distributed to individual investors
based on the percentage of their request as compared to the overall
tokens or cash received for all investors. 
\begin{enumerate}
\item There is an important point to be discussed here. The totals calculated
in Point (\ref{enu:Parity-Aggregate}) are good estimates of the exact
amounts of USD, Alpha, Beta and Gamma tokens we will need when we
actually receive them from the sub-funds. The calculations, in Point
(\ref{enu:Parity-Aggregate}) which are estimates, are only used to
invest and redeem from Alpha, Beta and Gamma. The actual allocations
in this step (Step \ref{enu:Allocations-Parity}) are based on the
tokens and cash actually received from Alpha, Beta and Gamma and what
the current weights and requirements of investors are. 
\item The key point to remember is that the steps from Step (\ref{enu:INTRINSIC-PLUS-PENDING})
to Step (\ref{enu:Using-Parity-Aggregate}) can be performed in one
blockchain transaction but within multiple computer science programming
functions. Some of the calculations mentioned in Step (\ref{enu:INTRINSIC-PLUS-PENDING})
and Step (\ref{enu:The-deposit-withdraw-user-netting}) happen when
a user deposit or withdraw action is performed. Step (\ref{enu:Parity-Step-One-Fund-Prices})
happens as a separate transaction in Alpha, Beta and Gamma before
we start Step (\ref{enu:INTRINSIC-PLUS-PENDING}) for Parity. Once
Steps (\ref{enu:INTRINSIC-PLUS-PENDING}) to (\ref{enu:Using-Parity-Aggregate})
are completed, this step (Step \ref{enu:Allocations-Parity}) is performed
as a separate transaction. 
\item There may be a way to automatically detect incoming tokens or cash
from Alpha, Beta and Gamma and perform this step without manager or
external intervention.
\item The calculations from Step (\ref{enu:INTRINSIC-PLUS-PENDING}) to
Step (\ref{enu:Parity-Aggregate}) can be repeated before doing this
step (Step \ref{enu:Allocations-Parity}). Even the price updates
from Step (\ref{enu:Parity-Step-One-Fund-Prices}) can be performed
before this step. This depends on the amount of time that has elapsed
since deposit and withdraw requests into Alpha, Beta and Gamma have
been made and cash or tokens have been received from ABG. It is also
worth emphasizing that the calculations in Step (\ref{enu:INTRINSIC-PLUS-PENDING})
to Step (\ref{enu:Parity-Aggregate}) can be repeated for any group
of investors before proceeding to this step (Step \ref{enu:Allocations-Parity}).
\item \label{enu:Parity-Token-Allocations}Let the number of tokens and
the amount of cash received from Alpha, Beta and Gamma be denoted
by $ALPHARCVD_{t}$, $BETARCVD_{t}$, $GAMMARCVD_{t}$ and $CASHRCVD_{t}$.
The ABG allocations, $ALPHAALLOC_{it}$, $BETAALLOC_{it}$, $GAMMAALLOC_{it}$,
for investor $i$ at time $t$ are given by,
\begin{align}
ALPHAALLOC_{it} & =\left[ALPHARCVD_{t}\right]\\
 & \begin{cases}
\begin{array}{cc}
\min\left[\left(\frac{INVALPHAAMNT_{it}}{PAINVALPHAAMNT_{t}}\right),1\right] & \;\left(PAINVALPHAAMNT_{t}\right)>0\\
0 & \;\left(PAINVALPHAAMNT_{t}\right)=0
\end{array}\end{cases}
\end{align}
\begin{equation}
INVALPHAAMNT_{it}=INVALPHAAMNT_{it}-ALPHAALLOC_{it}\left(ALPHAPRICE_{t}\right)
\end{equation}
\begin{equation}
ALPHARCVD_{t}=ALPHARCVD_{t}-ALPHAALLOC_{it}
\end{equation}
\begin{align}
BETAALLOC_{it} & =\left[BETARCVD_{t}\right]\\
 & \begin{cases}
\begin{array}{cc}
\min\left[\left(\frac{INVBETAAMNT_{it}}{PAINVBETAAMNT_{t}}\right),1\right] & \;\left(PAINVBETAAMNT_{t}\right)>0\\
0 & \;\left(PAINVBETAAMNT_{t}\right)=0
\end{array}\end{cases}
\end{align}
\begin{equation}
INVBETAAMNT_{it}=INVBETAAMNT_{it}-BETAALLOC_{it}\left(BETAPRICE_{t}\right)
\end{equation}
\begin{equation}
BETARCVD_{t}=BETARCVD_{t}-BETAALLOC_{it}
\end{equation}
\begin{align}
GAMMAALLOC_{it} & =\left[GAMMARCVD_{t}\right]\\
 & \begin{cases}
\begin{array}{cc}
\min\left[\left(\frac{INVGAMMAAMNT_{it}}{PAINVGAMMAAMNT_{t}}\right),1\right] & \;\left(PAINVGAMMAAMNT_{t}\right)>0\\
0 & \;\left(PAINVGAMMAAMNT_{t}\right)=0
\end{array}\end{cases}
\end{align}
\begin{equation}
INVGAMMAAMNT_{it}=INVGAMMAAMNT_{it}-GAMMAALLOC_{it}\left(GAMMAPRICE_{t}\right)
\end{equation}
\begin{equation}
GAMMARCVD_{t}=GAMMARCVD_{t}-GAMMAALLOC_{it}
\end{equation}
\textbf{Notice that we need to increment the Alpha quantity already
allocated to this investor by the new quantity that just got allocated
to this investor out of the Alpha tokens just received by Parity}.
The same is done for Beta and Gamma. Computer science optimizations
such as incrementing the quantity allocated to the investor directly
or using one temporary variable to calculate the allocation from this
recent event and adding it for each investor can be implemented.
\begin{equation}
ALPHAQTY_{it}=ALPHAQTY_{it}+ALPHAALLOC_{it}
\end{equation}
\begin{equation}
BETAQTY_{it}=BETAQTY_{it}+BETAALLOC_{it}
\end{equation}
\begin{equation}
GAMMAQTY_{it}=GAMMAQTY_{it}+GAMMAALLOC_{it}
\end{equation}
\item \label{enu:Parity-Cash-Allocations}Notice that we will receive cash
from withdraw requests made into Alpha, Beta and Gamma separately.
But we can club them together and handle it as one lump-sum for convenience.
$PAWDRWABGAMNT_{t}$is the total dollar withdraw request across all
investors from Alpha, Beta and Gamma. 
\begin{align}
PAWDRWABGAMNT_{t} & =\left(PAWDRWALPHAQTY_{t}\right)\left(ALPHAPRICE_{t}\right)\\
 & +\left(PAWDRWBETAQTY_{t}\right)\left(BETAPRICE_{t}\right)\\
 & +\left(PAWDRWGAMMAQTY_{t}\right)\left(GAMMAPRICE_{t}\right)
\end{align}
Let ($WDRWALPHAAMNT_{it}$, $WDRWBETAAMNT_{it}$, $WDRWGAMMAAMNT_{it}$)
be the dollar amounts to be withdrawn from Alpha, Beta and Gamma for
investor $i$ at time $t$. The cash allocation, $CASHALLOC_{it}$
for investor $i$ at time $t$ is given by,
\begin{align}
CASHALLOC_{it} & =\left[CASHRCVD_{t}\right]\\
 & \begin{cases}
\begin{array}{cc}
\min\left[\left(\frac{WDRWALPHAAMNT_{it}+WDRWBETAAMNT_{it}+WDRWGAMMAAMNT_{it}}{PAWDRWABGAMNT_{t}}\right),1\right]\\
0
\end{array}\end{cases}\\
 & \begin{cases}
\begin{array}{cc}
\;\left(PAWDRWABGAMNT_{t}\right)>0\\
\;\left(PAWDRWABGAMNT_{t}\right)=0
\end{array}\end{cases}
\end{align}
\begin{align}
WDRWALPHAQTY_{it} & =WDRWALPHAQTY_{it}\\
 & -\left(\frac{CASHALLOC_{it}}{ALPHAPRICE_{t}}\right)\\
 & \begin{cases}
\begin{array}{cc}
\left[\frac{WDRWALPHAAMNT_{it}}{WDRWALPHAAMNT_{it}+WDRWBETAAMNT_{it}+WDRWGAMMAAMNT_{it}}\right]\\
0
\end{array}\end{cases}\\
 & \begin{cases}
\begin{array}{cc}
WDRWALPHAAMNT_{it}+WDRWBETAAMNT_{it}\\
+WDRWGAMMAAMNT_{it}>0\\
WDRWALPHAAMNT_{it}+WDRWBETAAMNT_{it}\\
+WDRWGAMMAAMNT_{it}=0
\end{array}\end{cases}
\end{align}
\begin{align}
WDRWBETAQTY_{it} & =WDRWBETAQTY_{it}\\
 & -\left(\frac{CASHALLOC_{it}}{BETAPRICE_{t}}\right)\\
 & \begin{cases}
\begin{array}{cc}
\left[\frac{WDRWBETAAMNT_{it}}{WDRWALPHAAMNT_{it}+WDRWBETAAMNT_{it}+WDRWGAMMAAMNT_{it}}\right]\\
0
\end{array}\end{cases}\\
 & \begin{cases}
\begin{array}{cc}
WDRWALPHAAMNT_{it}+WDRWBETAAMNT_{it}\\
+WDRWGAMMAAMNT_{it}>0\\
WDRWALPHAAMNT_{it}+WDRWBETAAMNT_{it}\\
+WDRWGAMMAAMNT_{it}=0
\end{array}\end{cases}
\end{align}
\begin{align}
WDRWGAMMAQTY_{it} & =WDRWGAMMAQTY_{it}\\
 & -\left(\frac{CASHALLOC_{it}}{GAMMAPRICE_{t}}\right)\\
 & \begin{cases}
\begin{array}{cc}
\left[\frac{WDRWGAMMAAMNT_{it}}{WDRWALPHAAMNT_{it}+WDRWBETAAMNT_{it}+WDRWGAMMAAMNT_{it}}\right]\\
0
\end{array}\end{cases}\\
 & \begin{cases}
\begin{array}{cc}
WDRWALPHAAMNT_{it}+WDRWBETAAMNT_{it}\\
+WDRWGAMMAAMNT_{it}>0\\
WDRWALPHAAMNT_{it}+WDRWBETAAMNT_{it}\\
+WDRWGAMMAAMNT_{it}=0
\end{array}\end{cases}
\end{align}
\begin{equation}
CASHRCVD_{t}=CASHRCVD_{t}-CASHALLOC_{it}
\end{equation}
\item \label{enu:Cash-Allocation-Investor}The cash allocation received
by investor $i$ at time $t$, $CASHALLOC_{it}$, is aggregated along
with the deposit and withdraw in Equations (\ref{eq:Parity-Inflow-Investor};
\ref{eq:Parity-Outflow-Investor}) and sub-point (\ref{enu:Parity-Inflow-Outflow-Investor})
and put back into the fund in the next iteration. 
\begin{align}
PARITYINFLOW_{it} & =\max\left(DEPOSIT_{it}+CASHALLOC_{it}-PARITYWDRW_{it},0\right)\label{eq:Parity-Inflow-Investor-Cash-Rebalance}
\end{align}
\begin{align}
PARITYOUTFLOW_{it} & =\max\left(PARITYWDRW_{it}-DEPOSIT_{it}-CASHALLOC_{it},0\right)\label{eq:Parity-Outflow-Investor-Cash-Rebalance}
\end{align}
\textbf{We need separate variables for the cash allocation and deposit
amounts only if we wish to display the deposit amounts separately
to the user or maintain these separately for other reasons. This deposit
amount can be shown as a pending deposit before it enters the three
sub-funds and gets converted to ABG tokens. Having more granularity
in terms of storing data is always to be preferred, but within the
blockchain realm it might be prudent to reduce data stored. }
\item \label{enu:Cash-Allocation-Aggregate}The inclusion of cash at the
investor level also affects the parity aggregate inflow and outflow
Equations (\ref{eq:Parity-Inflow-Aggregate-Positive}; \ref{eq:Parity-Outflow-Aggregate-Positive})
in sub-point (\ref{enu:Parity-Fund-Level-Netting}).
\begin{align}
PARITYINFLOW_{t} & =\max\left(\sum_{i=1}^{i=IC_{t}}\left[PARITYINFLOW_{it}+CASHALLOC_{it}-PARITYOUTFLOW_{it}\right],0\right)\label{eq:Parity-Inflow-Aggregate-Positive-Rebalance}
\end{align}
\begin{align}
PARITYOUTFLOW_{t} & =\max\left(\sum_{i=1}^{i=IC_{t}}\left[PARITYOUTFLOW_{it}-PARITYINFLOW_{it}-CASHALLOC_{it}\right],0\right)\label{eq:Parity-Outflow-Aggregate-Positive-Rebalance}
\end{align}
\end{enumerate}
\end{enumerate}
Figures (\ref{fig:Parity-Sequence-Steps-Investor-Variables-I}; \ref{fig:Parity-Sequence-Steps-Investor-Amounts-Tokens-II};
\ref{fig:Parity-Sequence-Steps-Investor-Raw-III}; \ref{fig:Parity-Sequence-Steps-Aggregate-I-II-III})
illustrate different investors performing different actions - such
as withdraws, deposits, risk-return preference changes - with different
states of fund prices and their existing investments in the fund.
The illustrations show how investor allocations of Alpha, Beta and
Gamma tokens change and also how the fund level aggregate figures
are calculated. 

\subsection{\label{subsec:Efficient-Frontier-Parabolic}Illustration of The Efficient
Frontier Parabolic Equation}

This section is mostly concerned with sending a strong message to
blockchain investors about the powerful investment vehicle that can
be created for them using several decades of wisdom from the traditional
financial markets. We are using the Parity Line as an enhancement
to the concept of the efficient frontier. Hence, we are calling it
the final frontier of investing. But the efficient frontier is well
known in financial and investment circles. Hence, we will use the
below formulation to show a suitable parabolic curve below the Parity
Line on the GUI (Figure \ref{fig:Deposit-Screen-Parity}; Loney 1897;
End-note \ref{enu:In-mathematics,-Parabola}). 

We represent the Parity Line by an equation as 
\begin{equation}
y=mx+c
\end{equation}
Comparing this to the Parity Line in Equation (\ref{eq:Parity-Line})
shows the following: $y=E\left(R\right)$ is the expected return;
$x=\sigma$ is the risk; $m=\Theta$ and $c=R_{F}$ are the slope
and the intercept of the line respectively. 

Using the Parity Line, we calculate the following elements for the
purpose of graphing: the parabolic curve, the point where the Parity
Line is tangent to the parabolic curve, the maximum and minimum display
area in terms of X and Y co-ordinates. A subtle point to be noted
is that given a line and a point, we can have an infinite number of
parabolic curves such that the line is tangent to the curve at the
point. In our case, we narrow this down by choosing a parabolic curve
that has its axis parallel to the X-axis, or its directrix is perpendicular
to the X-axis, and the Y co-ordinate of its vertex has some order
of magnitude compared to the Y co-ordinate of the tangency point.
\begin{enumerate}
\item The maximum display area in terms of X and Y co-ordinates, $\left(X_{V},Y_{V}\right)$,
is given by:
\begin{align}
X_{V} & =\frac{V_{C}}{m}\\
Y_{V} & =V_{C}+x
\end{align}
Here, $V_{C}$ is the view constant which we configure to determine
how many points we want to have on the X axis. A suggested value is
$V_{C}=50$. 
\item The tangency point on the Parity Line, $\left(X_{T},Y_{T}\right)$,
is chosen so that its X co-ordinate is a certain order of magnitude
compared to the X co-ordinate of the maximum display area.
\begin{align}
X_{T} & =X_{V}T_{C}\\
Y_{T} & =mX_{V}T_{C}+c
\end{align}
Here, $T_{C}$ is the tangency constant which we configure to determine
the order of magnitude the X co-ordinate of the tangency point on
the Parity Line compared to the X co-ordinate of the maximum display
area. A suggested value is $T_{C}=0.5$. 
\item The equation of any parabolic curve such that its directrix is perpendicular
to the X-axis is 
\begin{equation}
\left(y-h\right)^{2}=A\left(x-k\right)\label{eq:Parabolic-Curve}
\end{equation}
We need to determine $\left(h,k,A\right)$ to satisfy the constraints
in our case as follows,
\begin{equation}
h=\frac{Y_{T}}{H_{C}}
\end{equation}
Here, $H_{C}$ is the height constant which we configure to determine
the displacement of the Y co-ordinate of the vertex of the parabola
compared to the Y co-ordinate of the tangency point. A suggested value
is $H_{C}=2$. Differentiating Equation (\ref{eq:Parabolic-Curve})
with respect to $x$ gives,
\begin{equation}
2\left(y-h\right)\left(\frac{dy}{dx}\right)=A
\end{equation}
Since the derivative of Equation (\ref{eq:Parabolic-Curve}) at the
point, $\left(X_{T},Y_{T}\right)$ is the slope of the Parity Line
we have,
\begin{equation}
A=2mY_{T}\left(1-\frac{1}{H_{C}}\right)
\end{equation}
Lastly, substituting the values of $\left(h,A\right)$ in Equation
(\ref{eq:Parabolic-Curve}) and solving gives, 
\begin{equation}
k=X_{T}-\frac{Y_{T}}{2m}\left(1-\frac{1}{H_{C}}\right)
\end{equation}
Notice that when plotting the parabolic curve, for any value of the
X co-ordinate, $x_{1}$, we will have two values of the Y co-ordinate,
$y_{1}$. That is,
\begin{equation}
y_{1}=h\pm\sqrt{A\left(x_{1}-k\right)}
\end{equation}
Also, we only have valid Y Co-ordinate values when $x_{1}>k$.
\end{enumerate}
Figures (\ref{fig:Parity-Line-Parabolic-Curve-Points}; \ref{fig:Parity-Line-Parabolic-Curve-Plots})
illustrate the graphs of the efficient-frontier and the final frontier
- parabolic curve and parity line - discussed in this section. Figure
(\ref{fig:Parity-Line-Parabolic-Curve-Points}) gives the numerical
values of the X and Y co-ordinates corresponding to the Parity Line
and the Parabolic Curve. representing the efficient frontier. Figure
(\ref{fig:Parity-Line-Parabolic-Curve-Plots}) is a graphical plot
of the Parity Line and the Parabolic Curve based on the numerical
example from Figure (\ref{fig:Parity-Line-Parabolic-Curve-Points}).

The following disclaimer or clarification needs to be added in the
Parity GUI at a suitable location, ideally somewhere at the bottom
perhaps: “The line in the graph shown on the Parity Deposit page,
termed the “Parity Line”, is the best combination of risk and returns
that investors can expect by blending together Alpha, Beta and Gamma.
This line represents a set of investment portfolios that are expected
to provide the highest returns at a given level of risk. In other
words, there is no other portfolio that offers higher returns for
a lower or equal amount of risk. The parabolic curve below the Parity
Line has been included for visual emphasis only and it represents
the efficient frontier curve when reasonable alternatives for the
risk free rate are not available. The careful construction and mixing
of our Alpha, Beta and Gamma portfolios, allows the Parity Line to
transcend the risk return combinations available from the creation
of all other portfolios.”

\section{\label{sec:Areas-for-Further}Implementation Pointers and Areas for
Further Research}

To perform some of the calculations we have discussed - such as risk
and return calculations plus other related enhancements using the
averaging techniques we have outlined in Kashyap (2023) - requires
being able to access a large number of historical transactions as
well. Providing such a large amount of input data to the decentralized
computer is still an area of active research (Wu et al., 2019; Kurt
Peker et al., 2020; Fan, Niu \& Liu 2022; End-note \ref{enu:Ethereum,-which-was}). 

Any intensive computations needed, to clarify the decision process
and arrive at the decision outcomes, can be done outside the blockchain
world, but the essential fund movements are better suited to happen
on a blockchain environment for security reasons. The interaction
between on-chain and off-chain components is a delicate balance involving
several trade-offs such as blockchain computational cost and not revealing
proprietary investment strategies (Garvey \& Murphy 2005; Pardo 2011;
Nuti et al., 2011). 

We believe that the efficient frontier is a moving target - even in
the traditional financial world - with assets being added or removed,
their risk-return properties undergoing alterations and even entire
markets getting transformed. This is all the more the case with the
rapidly evolving crypto landscape, where many new protocols and projects
are appearing on the scene. 

We have other analytical estimates and results in subsequent works
- pertaining to optimal fund flows and the periodicity of rebalancing
intervals - since this paper has a lot of innovations and material
already. But if necessary, we can include some of the mathematical
estimates as propositions - with the proofs provided in the appendix
or in the related paper. Other analytical work can be done regarding
the characteristics of the conceptual parity portfolio. The use of
various distance functions to get risk return profiles from different
weights of the sub-funds and vice versa can also be a fruitful endeavor.
Estimation of the parity line and how often to enforce the changed
parameters onto all users is a costly affair and hence a lot of further
research can be done to find out the cost-benefits and related optimality
conditions.

Any investment fund, whether on blockchain or outside, exists to generate
excess returns for its investors. Several excellent investment strategies
have been utilized in traditional investment funds to obtain higher
returns. To implement similar investment ideas on blockchain would
require considering each strategy as an overlay within a larger fund
(Mulvey, Ural \& Zhang 2007; Mohanty, Mohanty \& Ivanof 2021). As
time goes on, several overlay strategies can be added to the basic
fund so that we can benefit from any potential opportunities that
open up. 

It will be helpful - almost important - to understand newer blockchain
protocols and add them to our investment funds, which would render
them as highly diversified cross chain collectors of wealth appreciation
venues.   In addition - on each protocol - we need to continuously
evaluate new projects and - if they pass certain due diligence standards
- include them in our portfolio. All of this follows from standard
security analysis procedures from traditional finance. A team of researchers
and investment specialists need to continually scour the blockchain
investment landscape to identify ways to generate profits. 

It is helpful to consider exposure to derivative instruments and physical
assets such as gold, real estate, and so on, as and when they become
available. The implication of this is that investors in these funds
will be getting better returns and lower risks, as the funds seek
out varied sources of risk adjusted returns. As more sophisticated
derivatives start to become available as decentralized securities,
incorporating them could be challenging yet rewarding. The development
of new networks, and derivative providers within networks, will enable
the use of options as a hedging mechanism (Hull 2003). This will help
to protect from market crashes and to reduce the portfolio volatility.
Also, derivative strategies combined with rigorous risk management
can help to gain additional returns (Huberts 2004; Madan \& Sharaiha
2015). Numerous other areas for improvement, in terms of portfolio
weight calculations, rebalancing, trade execution risk management
and so on, are listed in Kashyap (2022). 

\subsection{\label{subsec:Beating-Benchmarks-with}Beating Benchmarks with The
Perfect Blend of Contrasting Correlations}

All major environmental economic risks can be reduced and classified
into one of the boxes in Figure (\ref{fig:Classification-of-Environmental}).
Each box represents one scenario in terms of whether expectations
regarding two key drivers, economic growth and inflation, turn out
to be lower or higher than expected. The trick is to distribute assets
so that whatever box we find ourselves in in the future we will maximize
returns, at the very least to perform better than cash. To do this,
we assign an equal 25\% risk allotment to each box since, after all,
we cannot predict the future and hence it is hard to know which of
the four categories we will end up in. This sameness in terms of the
treatment of risk gives the name: Risk Parity.  A detailed discussion
of all the components necessary to bring a robust risk management
approach to DeFi - including the design of the overall framework and
associated algorithms - is given in Kashyap (2022).

We will assign assets to each box - in Figure (\ref{fig:Classification-of-Environmental})
- so that the amount of risk is the same for each box, whatever the
returns from the box corresponding to the relevant market conditions.
For instance in Figure (\ref{fig:Portfolio-Risk-Concentrations}),
the dollar amount granted to Alpha assets will be lower than the appropriation
to Gamma since Gamma assets are less risky. Investors can choose a
desired level of risk - or equivalent to a target return - and the
distribution of their funds to each of the boxes, to create a customized
portfolio, will be done by the investment machinery we outline in
this paper. This personalized portfolio - issued as a Parity index
token on blockchain - will supply returns based upon the specific
risk appetite of the investor, no matter the environment.

Though less likely, if growth and inflation meet expectations, that
is a good problem to have. In this case, all four quadrants will perform
satisfactorily and the combined portfolio will still meet the stated
objectives.

\subsection{\label{subsec:Crypto-Environmental-Nuances}Crypto Environmental
Nuances}

While economic growth can be considered a basic environmental factor
for traditional assets, growth in the crypto universe can also be
linked to two parallel dimensions related to: trustworthiness and
benefits of transacting with a particular cryptocurrency as the medium
of exchange. Trustworthiness is heavily dependent on the extent of
computing power (or nodes)  deployed on a particular platform to verify
transactions. The transaction benefits are proportional to size of
the network or the number of active users (wallets) and how well accepted
the particular currency. The allocation to the four boxes will have
to balance out any fluctuations in the trends of these two key crypto-movers
over time.

In conventional investing techniques, selection of assets to balance
risks depends on their mutual correlations. Practical experience from
managing portfolios shows that correlations are inherently unstable.
Hence, as our models evolve, our approach to measuring alikeness will
use metrics that capture higher dimensions of similarity (and variability)
between assets by looking at attributes well beyond risk, return and
correlations.

The Parity indexes are composites of the Alpha, Beta and Gamma index
classes, weighted so that the risk from each bucket is limited to
25\% of the total. The final result is a individualized matrix of
aggregated cross-chain tokens and yield farming strategies (Figure
\ref{fig:Risk-Parity-Highlights}). It is algorithmically controlled
by smart contracts, regularly rebalancing to keep optimal coverage
and maintain the 25\% ratio. It is possible that most investors might
chose high risk - or low risk - as their preference. In this case,
it would be difficult to allot equal weight to all the four environmental
boxes in Figure (\ref{fig:Classification-of-Environmental}) unless
we limit the maximum risk - and hence the return - that can be obtained
from the overall fund.

The online investing process captures the member’s individual risk
appetite (directly or from their DeFi asset manager) and uses it to
weight the cross-chain Alpha, Beta and Gamma assets appropriately.
A Parity index token is then minted. It is personalized to the member’s
needs, representing the best risk-adjusted De-Fi investment in the
market. The Parity index token can itself then be staked, exchanged
or traded without any limit for further gain.

\subsection{\label{subsec:Fortune-Favors-The}Fortune Favors The Prepared Portfolio}

Current views on investing hold many reservations regarding the use
of Leverage. While excessive leverage can result in losses, a modest
amount of leverage, when applied to a small portion of the overall
portfolio, has many benefits. In conventional portfolios that are
chasing a certain level of return, the majority of the holdings will
be concentrated in assets that are closer to the benchmark in terms
of their returns. The risk characteristics of these assets are derived
from similar sources with the end result being that the overall portfolio
will not be well diversified.

A moderate amount of leverage is included in the Parity index class
due to the leverage within the Alpha index class. Our Alpha suite
will have assets that aim for spectacular returns, but could have
risks derived from similar fundamental properties. Hence, to mitigate
the concentration risk, we will mix in assets possessing distinctive
risk features. These non-identical assets could have lower return
profiles, but when they are levered up and combined into the portfolio,
they provide an excellent source of diverging movements that offset
the overall risk.

This amplification of the returns from chosen assets results in an
investment profile that can have lower risk than those of conventional
portfolios, which concentrate their holding around assets with similar
risk return configurations, providing the same level of returns. The
use of leverage in the DeFi world can be seamless due to the high
degree of automation. Smart contracts will monitor the level of leverage
and automatically trigger events to offset extreme and adverse shifts.

\subsection{\label{subsec:Risk-Management-Is}Risk Management Is But Taming The
Volatility Skew}

The Parity portfolios carry a risk level adjusted to the needs of
the member. Each Parity index token is personalized and minted with
the appropriate weightings of Alpha, Beta and Gamma index class investments.

The most crucial aspect of asset management is to have a rigorous
process to calculate risk.  Risk is defined and understood in several
ways. It is also important to differentiate between risk and uncertainty,
which we will explore in subsequent articles. We consider risk as
an indicator of the extent of variation in the returns of assets.
One of the commonly used metrics as a gauge for risk is the volatility.
A particularly popular approach to volatility is computing the standard
deviation of returns. 

The main issue with minimizing risk using volatility is that volatility
is an unobserved variable. This means volatility can only be estimated
over historical periods or forecasted over future times since it cannot
be directly seen. To mitigate this issue, our models will calculate
asset weights based on a range of volatilities rather than trying
to pin down one exact volatility number. The additional benefit from
this approach will be reduced rebalancing efforts, which will decrease
the corresponding blockchain gas costs.

A secondary drawback of volatility is that it fails to capture the
effect of changes in the direction of any variable (Kashyap 2021).
That is both rising or falling prices are treated equally and hence
if we are long a security, upward movements in security prices could
get penalized as excess volatility in portfolio management decisions.
As an alternative, we devise a metric to measure the path taken by
the variable to arrive at the current value, over the last few time
periods. To articulate the intuition, any security that had steady
upward growth, over a particular time period, is ranked higher than
a security with similar growth, over the same time period but with
more ups and downs in the path, or changes in direction. Assets are
penalized for having ups and downs in their price process. But unidirectional
upward jumps suffer no such penalty.

\subsection{\label{subsec:Sharpening-the-Sharpe}Sharpening the Sharpe Ratio}

A key metric used to assess the performance of portfolios is known
as the Sharpe Ratio (SR: Sharpe 1966; 1994; End-note \ref{enu:Sharpe-Ratio}).
The ratio is the average return (historical or expected) earned on
a portfolio in excess of the risk-free rate per unit of volatility
or total risk. The widespread use of this metric can be explained
due to its simplicity and the powerful insights it provides to investors
in understanding how their goals related to risk and return are being
met. 

An important limitation of the SR is that averages and standard deviations
are not enough to completely understand the randomness inherent in
the returns generated by any asset. While this limitation is severe
enough in traditional portfolios, it is all the more crucial in DeFi.
Given the greater volatilities of crypto-assets there are well justified
reasons to develop complementary techniques to measure portfolio performance.
As a refinement of the SR, we will pay close attention to additional
statistical properties called higher moments, also known as Skew and
Kurtosis, which show whether there is a greater tendency for either
positive or negative movement of returns and whether extreme events
are more likely when compared to a normal distribution (Anson, Ho
\& Silberstein 2007; Grigoletto \& Lisi 2011; Theodossiou \& Savva
2016; End-notes \ref{enu:In-probability-theory-Skew}; \ref{enu:In-probability-theory-Kurtosis}).

\subsection{\label{subsec:Investor-Experience-on}Investor Experience on Blockchain}

The user experience has been designed such that investors can tailor
their wealth allocations to their preferred risk appetites (Figure
\ref{fig:Deposit-Screen-Parity}).  Users can select either their
preferred level of risk or return. Investors can also directly decide
how much of their wealth they want to allocate to the three funds:
Alpha, Beta and Gamma. Once either of the three routes are selected,
(Risk or Return or Weights of Alpha, Beta and Gamma), the other parameters
are automatically calculated and saved into an NFT, which the investor
will hold for the life of the investment. 

The preferences can be changed anytime by investors and this will
trigger a readjustment of their sub fund allocations. Our investment
specialists will also monitor the markets and, as the relationship
between risk and return changes, will fine tune the parameters of
the parity line and update the parameters of the portfolio allocations.
This will guarantee that all investors are getting the best possible
outcomes customized for their desired wealth management objectives.

The challenge will be to ensure that the user interactions are intuitive,
and yet their preferences are precisely captured in the investment
decisions. This has been accomplished by letting someone who does
not wish to be bothered with all the settings, or a novice investor,
have the simple option of choosing the default, depositing his funds
and forgetting about everything else. If this is the option chosen,
the portfolio will select a low level of risk and calculate the other
parameters accordingly. Advanced users can choose their risk level
or their expected return, or the weights they want to assign to each
of the sub funds. The other parameters will be automatically calculated
based on the discussion in Section (\ref{sec:Risk-Parity:-Combining}).

The outcome of these innovations is an investment machinery that responds
to investor preferences and adapts to changing market conditions. 
All of this can be viewed as a natural progression from the conventional
efficient frontier to a progressive final frontier, which will continue
to transcend itself.

\subsection{\label{subsec:Parity-Fees}Parity Fees}

The following points capture the main fee structure for Parity investors.
Levying some kind of fees is the way in which investment funds - for
example: hedge funds, mutual funds, ETFs - cover their costs and make
profits (Golec 1996; Elton, Gruber \& Blake 2003; Guasoni \& Obłój
2016; Khorana et al., 2009; End-notes \ref{enu:A-hedge-fund}; \ref{enu:Open-end-mutual-funds};
\ref{enu:Mutual-Funds-Fees}; \ref{enu:An-exchange-traded-fund}).
In the blockchain environment the investment fund has additional costs
related to performing transactions on the network. 

We wish to highlight that there is a separate component that will
share a percentage of the profits - levied from the fees - with the
investors to ensure wealth management firms stay true to the spirit
of decentralization (Singh \& Kim 2019; Wang et al., 2019; Santana
\& Albareda 2022; Barbereau \& Bodó 2023; Rikken et al., 2023; Kashyap
2022; 2023; End-note \ref{enu:A-decentralized-autonomous-DAO}). All
the fee rates, time and other variables - the sample values shown
below are based on what we have used in various investment scenarios
- should be configurable unless otherwise specified.
\begin{enumerate}
\item \label{enu:Parity-investors-will}Parity investors can be charged
a deposit fees as a percentage of the amount being deposited.
\item \label{enu:Whenever-the-user}Whenever the user changes their preferences,
and updates their corresponding NFT, a percentage fees proportional
to the total deposited amount can also be levied. The user preferences
change fees should be capped at a certain maximum amount.
\item Whenever a new amount is deposited, the minimum of the fee component
for deposits or the fee amount for user preference changes can be
levied (Points \ref{enu:Parity-investors-will}; \ref{enu:Whenever-the-user}).
It is possible to separate the deposit action and the user preference
action in which case the deposit fees and user preference fees can
be set to be different depending on the network gas fees, fund expenses
and other factors (Zarir et a., 2021; Donmez \& Karaivanov 2022; Laurent
et al., 2022; Kashyap 2023; End-note \ref{enu:Gas-Fee}).
\item The withdraw fees or the redemption penalty for Parity can be in three
tiers for three different holding time periods. The redemption penalty
rates and the corresponding time periods should be configurable. We
provide the below only as a numerical example. Since these values
have to be configured depending on various considerations.
\begin{enumerate}
\item 5\% early redemption penalty for investing less than 30 days.
\item 4\% early redemption penalty for investing between 30 to 60 days.
\item 3\% early redemption penalty for investing between 60 to 90 days.
\item No redemption penalty for investing more than 90 days.
\end{enumerate}
\item There is no lockout period - a minimum investment period - for Parity
unlike Alpha, Beta and Gamma. There should be no lockout period for
Parity when it invests in Alpha, Beta and Gamma. The lockout periods
for Alpha, Beta and Gamma should be comparable to the no redemption
penalty time horizon for Parity.
\item Parity will be charged different deposit fee rates for investing into
Alpha, Beta and Gamma as compared to other investors who might invest
directly into Alpha, Beta and Gamma. This will allow investors who
invest into Parity rather than investing directly into Alpha, Beta
and Gamma to benefit from lower total fees. For example right now,
the Alpha deposit rate is 0.5\%. We would then have a different Parity
Alpha deposit rate, such as 0.3\%. This 0.3\% is an example only and
a provision has to made to so that it can be changed from a suitable
GUI by the administrators of the investment platform. Setting the
actual rates rather than setting a discount rate provides more flexibility
and the rate can also be made zero at times, which is impossible with
a discount rate.
\item The fees that Parity gets charged when investing into Alpha, Beta
and Gamma will then become a cost for Parity accordingly. Any benefits
from netting within Parity then adds to the profits generated by Parity
and will be sent to the corresponding treasury, which will be a separate
wallet or smart contract.
\end{enumerate}

\section{\label{sec:Conclusion}Conclusion: The Blockchain Revolution For
Wealth Gains Without The Pains}

We have created several novel techniques to bring many mechanisms
that have worked well in the traditional financial wealth management
arena to the blockchain space. We have given detailed algorithmic
steps to help with technical implementation of the methodologies we
have developed. 

We have discussed how to overcome several issues with the efficient
frontier. We have pioneered several solutions designed not only to
bring risk parity to the blockchain environment, but also provided
ways in which the efficient set of assets can continue to evolve -
becoming the final frontier of wealth management. 

Risk Parity investing for cross chain DeFi will be a new development
in decentralized investing. Members receive a personalized balanced
cross chain portfolio constructed algorithmically according to the
Risk Parity strategy. It includes varying measures of three funds
- index exposures to Alpha, Beta and Gamma, which have varying risk
profiles - risk-adjusted to the requirements of each investor. It
is engineered to provide long term high yields protected from market
downturns and environmental shocks. 

The innovations we have designed to realize risk parity on blockchain
also address techniques to combine the efficient frontier with equal
risk allocation to the assets within a portfolio. We have termed these
set of concepts: ``conceptual parity'' and detailed elaborate methods
to simplify its implementation in decentralized technology using the
notion of what we have descried as the ``Parity Line''.

We have given detailed mathematical formulations, and technical pointers,
to be able to implement the mechanisms we have created as blockchain
smart contracts. Our approach overcomes numerous blockchain bottlenecks
and takes the power of smart contracts much further. The numerical
illustrations we have given depict various scenarios regarding risk
return combinations - of Alpha, Beta and Gamma - to obtain parity
portfolios on blockchain.

The fundamental doctrine of DeFi is equal access and rights to all
participants coupled with complete freedom from any central controlling
authority. Bringing this notion of equality to investing is the core
impetus behind initiating the risk parity movement in DeFi. This translates
to providing sophisticated asset management - with risk mitigation
- techniques, insights and wealth appreciation schemes to all investors.
To become trailblazer 's on this voyage, is doing the needful - and
implementing some of the ideas discussed in this paper - so that accumulating
a Fortune is just a mouse click away from everyone. 

Mutual funds, hedge funds and other traditional investments have had
a significant impact in the lives of many individuals across the world.
Despite their popularity, there are many concerns regarding their
transparency and ease of access for everyone. Blockchain technology
is extremely well suited to mitigate, if not entirely eliminate, those
concerns. Decentralized ledger concepts and the technological advancements
over the last several decades allow us to combine the best features
of both traditional investment vehicles and the universal accessibility
of blockchain. 

\textbf{\textit{Equal Wealth Generation Opportunities are Finally
Here.}} 

\section{\label{sec:References}References}
\begin{itemize}
\item Acharya, V. V., \& Richardson, M. P. (Eds.). (2009). Restoring financial
stability: how to repair a failed system (Vol. 542). John Wiley \&
Sons.
\item Ambachtsheer, K. P. (1987). Pension fund asset allocation: In defense
of a 60/40 equity/debt asset mix. Financial Analysts Journal, 43(5),
14-24.
\item Ammann, M., Coqueret, G., \& Schade, J. P. (2016). Characteristics-based
portfolio choice with leverage constraints. Journal of Banking \&
Finance, 70, 23-37.
\item Anderson, R. M., Bianchi, S. W., \& Goldberg, L. R. (2012). Will my
risk parity strategy outperform?. Financial Analysts Journal, 68(6),
75-93.
\item Anson, M., Ho, H., \& Silberstein, K. (2007). Building a hedge fund
portfolio with kurtosis and skewness. The Journal of Alternative Investments,
10(1), 25.
\item Ante, L., Fiedler, I., \& Strehle, E. (2021). The influence of stablecoin
issuances on cryptocurrency markets. Finance Research Letters, 41,
101867.
\item Asness, C. S., Frazzini, A., \& Pedersen, L. H. (2012). Leverage aversion
and risk parity. Financial Analysts Journal, 68(1), 47-59.
\item Ballestero, E. (2005). Mean‐semivariance efficient frontier: a downside
risk model for portfolio selection. Applied Mathematical Finance,
12(1), 1-15.
\item Baldry, C. (1999). Space-the final frontier. Sociology, 33(3), 535-553.
\item Bai, X., Scheinberg, K., \& Tutuncu, R. (2016). Least-squares approach
to risk parity in portfolio selection. Quantitative Finance, 16(3),
357-376.
\item Bao, H., \& Roubaud, D. (2022). Non-fungible token: A systematic review
and research agenda. Journal of Risk and Financial Management, 15(5),
215.
\item Barbereau, T., \& Bodó, B. (2023). Beyond financial regulation of
crypto-asset wallet software: In search of secondary liability. Computer
Law \& Security Review, 49, 105829.
\item Barberis, N., \& Thaler, R. (2003). A survey of behavioral finance.
Handbook of the Economics of Finance, 1, 1053-1128.
\item Bellini, Fabio, Francesco Cesarone, Christian Colombo, and Fabio Tardella.
\textquotedbl Risk parity with expectiles.\textquotedbl{} European
journal of operational research 291, no. 3 (2021): 1149-1163.
\item Bender, J., Briand, R., Nielsen, F., \& Stefek, D. (2010). Portfolio
of risk premia: A new approach to diversification. The Journal of
Portfolio Management, 36(2), 17-25.
\item Bertsekas, D. (2009). Convex optimization theory (Vol. 1). Athena
Scientific. 
\item Bertsekas, D. (2015). Convex optimization algorithms. Athena Scientific.
\item Best, M. J., \& Hlouskova, J. (2000). The efficient frontier for bounded
assets. Mathematical methods of operations research, 52, 195-212.
\item Bhansali, V. (2011). Beyond risk parity. The Journal of Investing,
20(1), 137-147. 
\item Bi, J., Jin, H., \& Meng, Q. (2018). Behavioral mean-variance portfolio
selection. European Journal of Operational Research, 271(2), 644-663.
\item Black, F., \& Litterman, R. (1992). Global portfolio optimization.
Financial analysts journal, 48(5), 28-43.
\item Bolstad, W. M., \& Curran, J. M. (2016). Introduction to Bayesian
statistics. John Wiley \& Sons.
\item Bond, S. A., \& Satchell, S. E. (2002). Statistical properties of
the sample semi-variance. Applied Mathematical Finance, 9(4), 219-239.
\item Boyd, S. P., \& Vandenberghe, L. (2004). Convex optimization. Cambridge
university press. 
\item Braga, M. D., Nava, C. R., \& Zoia, M. G. (2023). Kurtosis-based risk
parity: methodology and portfolio effects. Quantitative Finance, 23(3),
453-469.
\item Britten‐Jones, M. (1999). The sampling error in estimates of mean‐variance
efficient portfolio weights. The Journal of Finance, 54(2), 655-671.
\item Brode, D., \& Brode, S. T. (Eds.). (2015). The Star Trek Universe:
Franchising the Final Frontier. Rowman \& Littlefield.
\item Buttell, A. E. (2010). Harry M. Markowitz on modern portfolio theory,
the efficient frontier, and His Life's Work. Journal of Financial
Planning, 23(5), 18.
\item Calvo, C., Ivorra, C., \& Liern, V. (2012). On the computation of
the efficient frontier of the portfolio selection problem. Journal
of Applied Mathematics, 2012.
\item Calvo, C., Ivorra, C., \& Liern, V. (2016). Fuzzy portfolio selection
with non-financial goals: exploring the efficient frontier. Annals
of Operations Research, 245(1-2), 31-46.
\item Camerer, C. (1999). Behavioral economics: Reunifying psychology and
economics. Proceedings of the National Academy of Sciences, 96(19),
10575-10577.
\item Chaves, D., Hsu, J., Li, F., \& Shakernia, O. (2011). Risk parity
portfolio vs. other asset allocation heuristic portfolios. Journal
of Investing, 20(1), 108. 
\item Chen, L., He, S., \& Zhang, S. (2011). Tight bounds for some risk
measures, with applications to robust portfolio selection. Operations
Research, 59(4), 847-865.
\item Christodoulakis, G. A. (2002). Bayesian optimal portfolio selection:
the black-litterman approach. Unpublished paper.
\item Ciotti, M., Ciccozzi, M., Terrinoni, A., Jiang, W. C., Wang, C. B.,
\& Bernardini, S. (2020). The COVID-19 pandemic. Critical reviews
in clinical laboratory sciences, 57(6), 365-388.
\item Clarke, R., De Silva, H., \& Thorley, S. (2013). Risk parity, maximum
diversification, and minimum variance: An analytic perspective. The
Journal of Portfolio Management, 39(3), 39-53. 
\item Cochrane, J. (2009). Asset pricing: Revised edition. Princeton university
press. 
\item Cohen, K. J., \& Pogue, J. A. (1967). An empirical evaluation of alternative
portfolio-selection models. The Journal of Business, 40(2), 166-193.
\item Constantinides, G. M., \& Malliaris, A. G. (1995). Portfolio theory.
Handbooks in operations research and management science, 9, 1-30.
\item Cook, S. (2000). The P versus NP problem. Clay Mathematics Institute,
2. 
\item Dannen, C. (2017). Introducing Ethereum and solidity (Vol. 1, pp.
159-160). Berkeley: Apress.
\item Da Silva, A. S., Lee, W., \& Pornrojnangkool, B. (2009). The Black–Litterman
model for active portfolio management. The Journal of Portfolio Management,
35(2), 61-70.
\item De Souza, C., \& Smirnov, M. (2004). Dynamic leverage. Journal of
Portfolio Management, 31, 25-39.
\item Donmez, A., \& Karaivanov, A. (2022). Transaction fee economics in
the Ethereum blockchain. Economic Inquiry, 60(1), 265-292.
\item Drobetz, W. (2001). How to avoid the pitfalls in portfolio optimization?
Putting the Black-Litterman approach at work. Financial Markets and
Portfolio Management, 15(1), 59. 
\item Dromey, R. G. (1982). How to Solve it by Computer. Prentice-Hall,
Inc..
\item Dubois, D., \& Prade, H. (2001). Possibility theory, probability theory
and multiple-valued logics: A clarification. Annals of mathematics
and Artificial Intelligence, 32, 35-66.
\item Dubois, D. (2006). Possibility theory and statistical reasoning. Computational
statistics \& data analysis, 51(1), 47-69.
\item Duffie, D., \& Pan, J. (1997). An overview of value at risk. Journal
of derivatives, 4(3), 7-49.
\item Elton, E. J., \& Gruber, M. J. (1973). Estimating the dependence structure
of share prices-{}-implications for portfolio selection. The Journal
of Finance, 28(5), 1203-1232.
\item Elton, E. J., Gruber, M. J., \& Padberg, M. W. (1976). Simple criteria
for optimal portfolio selection. The Journal of finance, 31(5), 1341-1357.
\item Elton, E. J., Gruber, M. J., \& Padberg, M. W. (1977a). Simple criteria
for optimal portfolio selection with upper bounds. Operations Research,
25(6), 952-967.
\item Elton, E. J., Gruber, M. J., \& Padberg, M. W. (1977b). Simple rules
for optimal portfolio selection: the multi group case. Journal of
Financial and Quantitative Analysis, 12(3), 329-345.
\item Elton, E. J., Gruber, M. J., \& Padberg, M. W. (1978). Simple criteria
for optimal portfolio selection: tracing out the efficient frontier.
The Journal of Finance, 33(1), 296-302.
\item Elton, E. J., Gruber, M. J., \& Padberg, M. W. (1979). Simple criteria
for optimal portfolio selection: the multi-index case. Portfolio theory,
25, 7-19.
\item Elton, E. J., \& Gruber, M. J. (1997). Modern portfolio theory, 1950
to date. Journal of banking \& finance, 21(11-12), 1743-1759. 
\item Elton, E. J., Gruber, M. J., \& Blake, C. R. (2003). Incentive fees
and mutual funds. The Journal of Finance, 58(2), 779-804. 
\item Elton, E. J., Gruber, M. J., Brown, S. J., \& Goetzmann, W. N. (2009).
Modern portfolio theory and investment analysis. John Wiley \& Sons.
\item Fan, X., Niu, B., \& Liu, Z. (2022). Scalable blockchain storage systems:
research progress and models. Computing, 104(6), 1497-1524. 
\item Fabozzi, F. J., Gupta, F., \& Markowitz, H. M. (2002). The legacy
of modern portfolio theory. The journal of investing, 11(3), 7-22. 
\item Fabozzi, F. A., Simonian, J., \& Fabozzi, F. J. (2021). Risk parity:
The democratization of risk in asset allocation. The Journal of Portfolio
Management, 47(5), 41-50.
\item Fama, E. F., \& MacBeth, J. D. (1973). Risk, return, and equilibrium:
Empirical tests. Journal of political economy, 81(3), 607-636. 
\item Fama, E. F., \& French, K. R. (2004). The capital asset pricing model:
Theory and evidence. Journal of economic perspectives, 18(3), 25-46.
\item Fernández, A., \& Gómez, S. (2007). Portfolio selection using neural
networks. Computers \& operations research, 34(4), 1177-1191.
\item Fisher, G. S., Maymin, P. Z., \& Maymin, Z. G. (2015). Risk parity
optimality. The Journal of Portfolio Management, 41(2), 42-56.
\item Fortnow, L. (2009). The status of the P versus NP problem. Communications
of the ACM, 52(9), 78-86. 
\item Frahm, G., \& Wiechers, C. (2011). On the diversification of portfolios
of risky assets (No. 2/11). Discussion Papers in Statistics and Econometrics.
\item Frost, P. A., \& Savarino, J. E. (1986). An empirical Bayes approach
to efficient portfolio selection. Journal of Financial and Quantitative
Analysis, 21(3), 293-305.
\item Garas, A., Argyrakis, P., \& Havlin, S. (2008). The structural role
of weak and strong links in a financial market network. The European
Physical Journal B, 63, 265-271.
\item Garvey, R., \& Murphy, A. (2005). Entry, exit and trading profits:
A look at the trading strategies of a proprietary trading team. Journal
of Empirical Finance, 12(5), 629-649.
\item Gelman, A., \& Shalizi, C. R. (2013). Philosophy and the practice
of Bayesian statistics. British Journal of Mathematical and Statistical
Psychology, 66(1), 8-38. 
\item Goldfarb, D., \& Iyengar, G. (2003). Robust portfolio selection problems.
Mathematics of operations research, 28(1), 1-38.
\item Golec, J. H. (1996). The effects of mutual fund managers' characteristics
on their portfolio performance, risk and fees. Financial Services
Review, 5(2), 133-147. 
\item Green, R. C., \& Hollifield, B. (1992). When will mean‐variance efficient
portfolios be well diversified?. The Journal of Finance, 47(5), 1785-1809.
\item Grigoletto, M., \& Lisi, F. (2011). Practical implications of higher
moments in risk management. Statistical Methods \& Applications, 20,
487-506.
\item Grobys, K., Junttila, J., Kolari, J. W., \& Sapkota, N. (2021). On
the stability of stablecoins. Journal of Empirical Finance, 64, 207-223.
\item Gunjan, A., \& Bhattacharyya, S. (2023). A brief review of portfolio
optimization techniques. Artificial Intelligence Review, 56(5), 3847-3886.
\item Guasoni, P., \& Obłój, J. (2016). The incentives of hedge fund fees
and high‐water marks. Mathematical Finance, 26(2), 269-295. 
\item Harvey, C. R., Hoyle, E., Korgaonkar, R., Rattray, S., Sargaison,
M., \& Van Hemert, O. (2018). The impact of volatility targeting.
The Journal of Portfolio Management, 45(1), 14-33.
\item Harvey, C. R., Ramachandran, A., \& Santoro, J. (2021). DeFi and the
Future of Finance. John Wiley \& Sons. 
\item He, G., \& Litterman, R. (2002). The intuition behind Black-Litterman
model portfolios. Available at SSRN 334304.
\item Hirshleifer, D. (2015). Behavioral finance. Annual Review of Financial
Economics, 7, 133-159.
\item Hoang, L. T., \& Baur, D. G. (2021). How stable are stablecoins?.
The European Journal of Finance, 1-17.
\item Huang, X. (2006). Fuzzy chance-constrained portfolio selection. Applied
mathematics and computation, 177(2), 500-507.
\item Huang, X. (2007a). Portfolio selection with fuzzy returns. Journal
of Intelligent \& Fuzzy Systems, 18(4), 383-390.
\item Huang, X. (2007b). Two new models for portfolio selection with stochastic
returns taking fuzzy information. European Journal of Operational
Research, 180(1), 396-405.
\item Huang, X. (2008a). Risk curve and fuzzy portfolio selection. Computers
\& mathematics with applications, 55(6), 1102-1112.
\item Huang, X. (2008b). Portfolio selection with a new definition of risk.
European Journal of operational research, 186(1), 351-357.
\item Huberts, L. C. (2004). Overlay Speak. The Journal of Investing, 13(3),
22-30.
\item Hull, J. C. (2003). Options futures and other derivatives. Pearson
Education India.
\item Hursh, S. R. (1984). Behavioral economics. Journal of the experimental
analysis of behavior, 42(3), 435-452.
\item Jacobs, B. I., \& Levy, K. N. (2012). Leverage aversion and portfolio
optimality. Financial Analysts Journal, 68(5), 89-94.
\item Jolliffe, I. T., \& Cadima, J. (2016). Principal component analysis:
a review and recent developments. Philosophical transactions of the
royal society A: Mathematical, Physical and Engineering Sciences,
374(2065), 20150202.
\item Jorion, P. (1996). Risk2: Measuring the risk in value at risk. Financial
analysts journal, 52(6), 47-56. 
\item Jorion, P. (2007). Value at risk: the new benchmark for managing financial
risk. McGraw-Hill.
\item Kalymon, B. A. (1971). Estimation risk in the portfolio selection
model. Journal of Financial and Quantitative Analysis, 6(1), 559-582.
\item Kashyap, R. (2021). Behavioural Bias Benefits: Beating Benchmarks
By Bundling Bouncy Baskets. Accounting \& Finance, 61(3), 4885-4921.
\item Kashyap, R. (2022). Bringing Risk Parity To The DeFi Party: A Complete
Solution To The Crypto Asset Management Conundrum. Working Paper.
\item Kashyap, R. (2023). The Democratization of Wealth Management: Hedged
Mutual Fund Blockchain Protocol. Available at SSRN 4543339.
\item Khorana, A., Servaes, H., \& Tufano, P. (2009). Mutual fund fees around
the world. The Review of Financial Studies, 22(3), 1279-1310.
\item Kim, J. H., Lee, Y., Kim, W. C., \& Fabozzi, F. J. (2021). Mean–variance
optimization for asset allocation. The Journal of Portfolio Management,
47(5), 24-40.
\item Kocheturov, A., Batsyn, M., \& Pardalos, P. M. (2014). Dynamics of
cluster structures in a financial market network. Physica A: Statistical
Mechanics and its Applications, 413, 523-533.
\item Konstantinov, G. S. (2021). What Portfolio in Europe Makes Sense?.
The Journal of Portfolio Management.
\item Kroll, Y., Levy, H., \& Rapoport, A. (1988). Experimental tests of
the mean-variance model for portfolio selection. Organizational Behavior
and Human Decision Processes, 42(3), 388-410.
\item Krugman, P. (1998). Space: the final frontier. Journal of Economic
perspectives, 12(2), 161-174.
\item Kugler, L. (2021). Non-fungible tokens and the future of art. Communications
of the ACM, 64(9), 19-20.
\item Kurt Peker, Y., Rodriguez, X., Ericsson, J., Lee, S. J., \& Perez,
A. J. (2020). A cost analysis of internet of things sensor data storage
on blockchain via smart contracts. Electronics, 9(2), 244. 
\item Kwakernaak, H. (1978). Fuzzy random variables—I. Definitions and theorems.
Information sciences, 15(1), 1-29. 
\item Laurent, A., Brotcorne, L., \& Fortz, B. (2022). Transaction fees
optimization in the Ethereum blockchain. Blockchain: Research and
Applications, 3(3), 100074.
\item Ledoit, O., \& Wolf, M. (2003). Improved estimation of the covariance
matrix of stock returns with an application to portfolio selection.
Journal of empirical finance, 10(5), 603-621.
\item Leland, H. E. (1999). Beyond mean–variance: Performance measurement
in a nonsymmetrical world (corrected). Financial analysts journal,
55(1), 27-36.
\item Li, B., \& Hoi, S. C. (2014). Online portfolio selection: A survey.
ACM Computing Surveys (CSUR), 46(3), 1-36.
\item Li, B., Hoi, S. C., Sahoo, D., \& Liu, Z. Y. (2015). Moving average
reversion strategy for on-line portfolio selection. Artificial Intelligence,
222, 104-123.
\item Lindley, D. V. (1972). Bayesian statistics: A review. Society for
industrial and applied mathematics. 
\item Liu, Y. J., \& Zhang, W. G. (2013). Fuzzy portfolio optimization model
under real constraints. Insurance: Mathematics and Economics, 53(3),
704-711.
\item Lohre, H., Neugebauer, U., \& Zimmer, C. (2012). Diversified risk
parity strategies for equity portfolio selection. The Journal of Investing,
21(3), 111-128. 
\item Lohre, H., Opfer, H., \& Orszag, G. (2014). Diversifying risk parity.
Journal of Risk, 16(5), 53-79.
\item Loney, S. L. (1897). The elements of coordinate geometry. Macmillan
and Company.
\item Lyons, R. K., \& Viswanath-Natraj, G. (2023). What keeps stablecoins
stable?. Journal of International Money and Finance, 131, 102777.
\item Madan, D. B., \& Sharaiha, Y. M. (2015). Option overlay strategies.
Quantitative Finance, 15(7), 1175-1190. 
\item Maillard, S., Roncalli, T., \& Teïletche, J. (2010). The properties
of equally weighted risk contribution portfolios. The Journal of Portfolio
Management, 36(4), 60-70.
\item Mao, J. C., \& Särndal, C. E. (1966). A decision theory approach to
portfolio selection. Management Science, 12(8), B-323.
\item Markowitz, H. M. (1952). Portfolio selection. Journal of Finance 7
(1), 77-91.
\item Markowitz, H. M. (1959). Portfolio Selection: Efficient Diversification
of Investments. Cowles Foundation Monograph 16, New York: John Wiley
\& Sons.
\item Martinez, W. L. (2011). Graphical user interfaces. Wiley Interdisciplinary
Reviews: Computational Statistics, 3(2), 119-133.
\item Meucci, A. (2009). Managing diversification. Risk, 74-79.
\item McNeil, A. J., Frey, R., \& Embrechts, P. (2015). Quantitative risk
management: concepts, techniques and tools-revised edition. Princeton
university press.
\item Merton, R. C. (1972). An analytic derivation of the efficient portfolio
frontier. Journal of financial and quantitative analysis, 7(4), 1851-1872.
\item Michaud, R. O. (1989). The Markowitz optimization enigma: Is ‘optimized’optimal?.
Financial analysts journal, 45(1), 31-42.
\item Mohanty, S. S., Mohanty, O., \& Ivanof, M. (2021). Alpha enhancement
in global equity markets with ESG overlay on factor-based investment
strategies. Risk Management, 23(3), 213-242.
\item Momen, O., Esfahanipour, A., \& Seifi, A. (2019). Collective mental
accounting: an integrated behavioural portfolio selection model for
multiple mental accounts. Quantitative Finance, 19(2), 265-275. 
\item Momen, O., Esfahanipour, A., \& Seifi, A. (2020). A robust behavioral
portfolio selection: model with investor attitudes and biases. Operational
Research, 20, 427-446.
\item Monrat, A. A., Schelén, O., \& Andersson, K. (2019). A survey of blockchain
from the perspectives of applications, challenges, and opportunities.
IEEE Access, 7, 117134-117151. 
\item Mullainathan, S., \& Thaler, R. H. (2000). Behavioral Economics (No.
7948). National Bureau of Economic Research, Inc.
\item Mulvey, J. M., Ural, C., \& Zhang, Z. (2007). Improving performance
for long-term investors: wide diversification, leverage, and overlay
strategies. Quantitative Finance, 7(2), 175-187. 
\item Murty, K. G., \& Kabadi, S. N. (1985). Some NP-complete problems in
quadratic and nonlinear programming.
\item Nahmias, S. (1978). Fuzzy variables. Fuzzy sets and systems, 1(2),
97-110.
\item Nakamoto, S. (2008). Bitcoin: A peer-to-peer electronic cash system.
Decentralized Business Review, 21260. 
\item Namaki, A., Shirazi, A. H., Raei, R., \& Jafari, G. R. (2011). Network
analysis of a financial market based on genuine correlation and threshold
method. Physica A: Statistical Mechanics and its Applications, 390(21-22),
3835-3841.
\item Nantell, T. J., \& Price, B. (1979). An analytical comparison of variance
and semivariance capital market theories. Journal of Financial and
Quantitative Analysis, 14(2), 221-242.
\item Narayanan, A., \& Clark, J. (2017). Bitcoin’s academic pedigree. Communications
of the ACM, 60(12), 36-45.
\item Nuti, G., Mirghaemi, M., Treleaven, P., \& Yingsaeree, C. (2011).
Algorithmic trading. Computer, 44(11), 61-69. 
\item Palczewski, A., \& Palczewski, J. (2014). Theoretical and empirical
estimates of mean–variance portfolio sensitivity. European Journal
of Operational Research, 234(2), 402-410.
\item Pardo, R. (2011). The evaluation and optimization of trading strategies.
John Wiley \& Sons. 
\item Pástor, Ľ. (2000). Portfolio selection and asset pricing models. The
Journal of Finance, 55(1), 179-223.
\item Pearson, R., \& Davies, M. M. (2014). Star Trek and American Television.
Univ of California Press.
\item Penman, S. H. (1970). What Net Asset Value?-{}-An Extension of a Familiar
Debate. The Accounting Review, 45(2), 333-346.
\item Peralta, G., \& Zareei, A. (2016). A network approach to portfolio
selection. Journal of Empirical Finance, 38, 157-180.
\item Pierro, G. A., \& Tonelli, R. (2022, March). Can Solana be the solution
to the blockchain scalability problem?. In 2022 IEEE International
Conference on Software Analysis, Evolution and Reengineering (SANER)
(pp. 1219-1226). IEEE.
\item Piñeiro-Chousa, J., López-Cabarcos, M. Á., Sevic, A., \& González-López,
I. (2022). A preliminary assessment of the performance of DeFi cryptocurrencies
in relation to other financial assets, volatility, and user-generated
content. Technological Forecasting and Social Change, 181, 121740.
\item Pogue, G. A. (1970). An extension of the Markowitz portfolio selection
model to include variable transactions' costs, short sales, leverage
policies and taxes. The Journal of Finance, 25(5), 1005-1027.
\item Polson, N. G., \& Tew, B. V. (2000). Bayesian portfolio selection:
An empirical analysis of the S\&P 500 index 1970–1996. Journal of
Business \& Economic Statistics, 18(2), 164-173.
\item Prince, B. (2011). Risk Parity Is About Balance. Bridgewater Associates,
LP.
\item Qian, E. (2011). Risk parity and diversification. The Journal of Investing,
20(1), 119-127. 
\item Qian, E. (2013). Are risk-parity managers at risk parity?. The Journal
of Portfolio Management, 40(1), 20-26. 
\item Qian, J. (2015). An introduction to asset pricing theory.
\item Rapach, D., \& Zhou, G. (2013). Forecasting stock returns. In Handbook
of economic forecasting (Vol. 2, pp. 328-383). Elsevier.
\item Reinhart, C. M., \& Rogoff, K. S. (2008). This time is different:
A panoramic view of eight centuries of financial crises (No. w13882).
National Bureau of Economic Research.
\item Rikken, O., Janssen, M., \& Kwee, Z. (2023). The ins and outs of decentralized
autonomous organizations (DAOs) unraveling the definitions, characteristics,
and emerging developments of DAOs. Blockchain: Research and Applications,
100143.
\item Ritter, J. R. (2003). Behavioral finance. Pacific-Basin finance journal,
11(4), 429-437.
\item Roncalli, T. (2013). Introduction to risk parity and budgeting. CRC
Press.
\item Roncalli, T., \& Weisang, G. (2016). Risk parity portfolios with risk
factors. Quantitative Finance, 16(3), 377-388.
\item Rosenthal, J. S. (2006). First Look At Rigorous Probability Theory,
A. World Scientific Publishing Company.
\item Rubinstein, M. (2002). Markowitz's\textquotedbl{} portfolio selection\textquotedbl :
A fifty-year retrospective. The Journal of finance, 57(3), 1041-1045.
\item Rudin, W. (1953). Principles of mathematical analysis.
\item Sahni, S. (1974). Computationally related problems. SIAM Journal on
computing, 3(4), 262-279.
\item Santana, C., \& Albareda, L. (2022). Blockchain and the emergence
of Decentralized Autonomous Organizations (DAOs): An integrative model
and research agenda. Technological Forecasting and Social Change,
182, 121806.
\item Schlecht, L., Schneider, S., \& Buchwald, A. (2021). The prospective
value creation potential of Blockchain in business models: A delphi
study. Technological Forecasting and Social Change, 166, 120601.
\item Schuhmacher, F., Kohrs, H., \& Auer, B. R. (2021). Justifying mean-variance
portfolio selection when asset returns are skewed. Management Science,
67(12), 7812-7824.
\item Sharpe, W. F. (1966). Mutual fund performance. The Journal of business,
39(1), 119-138. 
\item Sharpe, W. F. (1994). The sharpe ratio. Journal of portfolio management,
21(1), 49-58. 
\item Shlens, J. (2014). A tutorial on principal component analysis. arXiv
preprint arXiv:1404.1100. 
\item Singh, M., \& Kim, S. (2019). Blockchain technology for decentralized
autonomous organizations. In Advances in computers (Vol. 115, pp.
115-140). Elsevier.
\item Sipser, M. (2006). Introduction to the Theory of Computation (Vol.
2). Boston: Thomson Course Technology. 
\item Tapscott, D., \& Tapscott, A. (2016). Blockchain revolution: how the
technology behind bitcoin is changing money, business, and the world.
Penguin. 
\item Thaler, R. H. (1999). The end of behavioral finance. Financial Analysts
Journal, 55(6), 12-17.
\item Thaler, R. H. (2016). Behavioral economics: Past, present, and future.
American economic review, 106(7), 1577-1600.
\item Thaler, R. H. (2017). Behavioral economics. Journal of Political Economy,
125(6), 1799-1805.
\item Theodossiou, P., \& Savva, C. S. (2016). Skewness and the relation
between risk and return. Management Science, 62(6), 1598-1609.
\item Thiagarajan, S. R., \& Schachter, B. (2011). Risk parity: Rewards,
risks, and research opportunities. The Journal of Investing, 20(1),
79-89. 
\item Timmermann, A. (2008). Elusive return predictability. International
Journal of Forecasting, 24(1), 1-18.
\item Tomer, J. F. (2007). What is behavioral economics?. The Journal of
Socio-Economics, 36(3), 463-479.
\item Treynor, J. L., \& Black, F. (1973). How to use security analysis
to improve portfolio selection. The journal of business, 46(1), 66-86.
\item Tsao, C. Y. (2010). Portfolio selection based on the mean–VaR efficient
frontier. Quantitative Finance, 10(8), 931-945.
\item van de Schoot, R., Depaoli, S., King, R., Kramer, B., Märtens, K.,
Tadesse, M. G., ... \& Yau, C. (2021). Bayesian statistics and modelling.
Nature Reviews Methods Primers, 1(1), 1.
\item Wang, S., Ding, W., Li, J., Yuan, Y., Ouyang, L., \& Wang, F. Y. (2019).
Decentralized autonomous organizations: Concept, model, and applications.
IEEE Transactions on Computational Social Systems, 6(5), 870-878.
\item Wang, Q., Li, R., Wang, Q., \& Chen, S. (2021). Non-fungible token
(NFT): Overview, evaluation, opportunities and challenges. arXiv preprint
arXiv:2105.07447.
\item Wood, G. (2014). Ethereum: A secure decentralised generalised transaction
ledger. Ethereum project yellow paper, 151(2014), 1-32.
\item Wu, H., Cao, J., Yang, Y., Tung, C. L., Jiang, S., Tang, B., ... \&
Deng, Y. (2019, July). Data management in supply chain using blockchain:
Challenges and a case study. In 2019 28th International Conference
on Computer Communication and Networks (ICCCN) (pp. 1-8). IEEE. 
\item Xu, J., \& Feng, Y. (2022). Reap the Harvest on Blockchain: A Survey
of Yield Farming Protocols. IEEE Transactions on Network and Service
Management. 
\item Yan, W., Miao, R., \& Li, S. (2007). Multi-period semi-variance portfolio
selection: Model and numerical solution. Applied Mathematics and Computation,
194(1), 128-134.
\item Zadeh, L. A. (1965). Fuzzy sets. Information and control, 8(3), 338-353.
\item Zadeh, L. A. (1978). Fuzzy sets as a basis for a theory of possibility.
Fuzzy sets and systems, 1(1), 3-28.
\item Zarir, A. A., Oliva, G. A., Jiang, Z. M., \& Hassan, A. E. (2021).
Developing cost-effective blockchain-powered applications: A case
study of the gas usage of smart contract transactions in the ethereum
blockchain platform. ACM Transactions on Software Engineering and
Methodology (TOSEM), 30(3), 1-38.
\item Zetzsche, D. A., Arner, D. W., \& Buckley, R. P. (2020). Decentralized
finance (defi). Journal of Financial Regulation, 6, 172-203.
\item Zhang, W. G., Wang, Y. L., Chen, Z. P., \& Nie, Z. K. (2007). Possibilistic
mean–variance models and efficient frontiers for portfolio selection
problem. Information Sciences, 177(13), 2787-2801.
\item Zhang, W. G., \& Wang, Y. L. (2008). An analytic derivation of admissible
efficient frontier with borrowing. European Journal of Operational
Research, 184(1), 229-243.
\item Zhang, Y., Li, X., \& Guo, S. (2018). Portfolio selection problems
with Markowitz’s mean–variance framework: a review of literature.
Fuzzy Optimization and Decision Making, 17, 125-158.
\item Zhao, D., Bai, L., Fang, Y., \& Wang, S. (2022). Multi‐period portfolio
selection with investor views based on scenario tree. Applied Mathematics
and Computation, 418, 126813.
\item Zheng, Z., Xie, S., Dai, H. N., Chen, W., Chen, X., Weng, J., \& Imran,
M. (2020). An overview on smart contracts: Challenges, advances and
platforms. Future Generation Computer Systems, 105, 475-491.
\item Zou, W., Lo, D., Kochhar, P. S., Le, X. B. D., Xia, X., Feng, Y.,
... \& Xu, B. (2019). Smart contract development: Challenges and opportunities.
IEEE Transactions on Software Engineering, 47(10), 2084-2106.
\end{itemize}
\begin{doublespace}
\begin{center}
\pagebreak{}
\par\end{center}
\end{doublespace}

\part{\label{part:Supplementary-Online-Material}Supplementary Online Material}

\section{\label{sec:Parity-Flow-Flow}Appendix: Parity Flow Flow Chart}

The flow chart in Figure (\ref{fig:Sequences-of-Steps}) corresponds
to all the steps mentioned in Section (\ref{sec:Risk-Parity:-Combining}).

\begin{figure}[H]
\includegraphics[width=17cm,height=10cm]{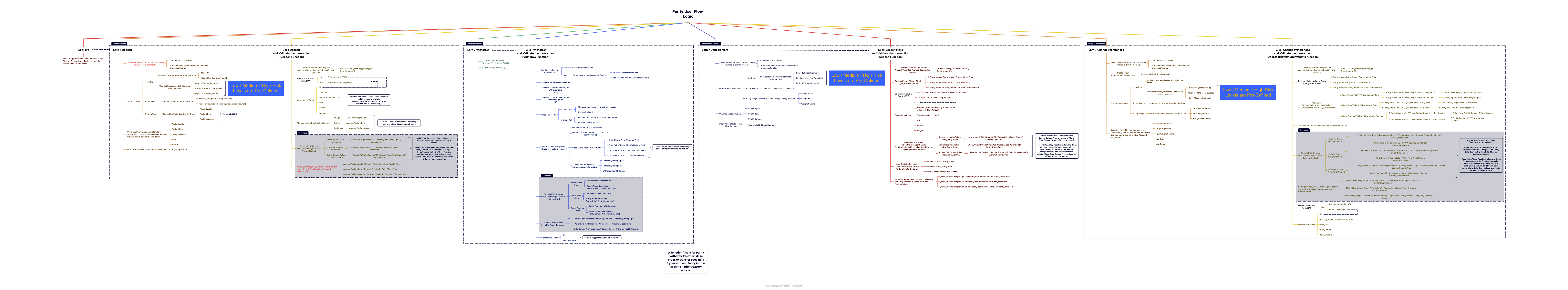}

\caption{Parity Flow Flow Chart: Sequences of Steps for Periodic Fund Management\label{fig:Sequences-of-Steps}}
\end{figure}

\section{\label{sec:Numerical-Results}Appendix of Numerical Results}

Each of the tables in this section are referenced in the main body
of the article. Below, we provide supplementary descriptions for each
table. 
\begin{itemize}
\item Figure (\ref{fig:Parity-Line:-ABG-Scatter}) shows three clouds of
points representing the movement of Alpha, Beta and Gamma in Risk
and Return space. Alpha is the right most, Beta is in the middle and
Gamma is the left most set of points. The X-Axis is the risk axis
and Y-axis is for returns.
\end{itemize}
\begin{figure}[H]
\includegraphics[width=18cm]{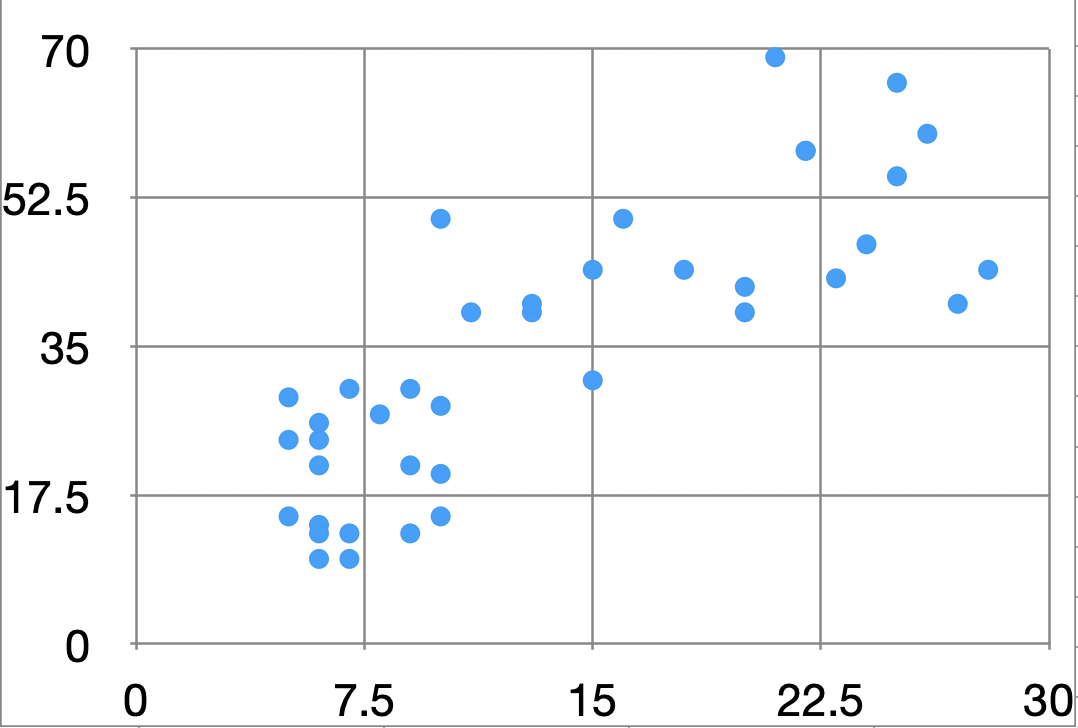}\caption{\label{fig:Parity-Line:-ABG-Scatter}Parity Line: Alpha, Beta Gamma
Combination Scatter Plot}
\end{figure}

\begin{itemize}
\item The Table in Figure (\ref{fig:Parity-Line:-ABG-Sample}) shows numerical
examples related to combining Alpha, Beta and Gamma to get the Parity
portfolio given their risk, return and correlation matrix. The weight
of Alpha, $w_{\alpha}$is given by the formula,
\begin{equation}
w_{\alpha}=\left(\frac{\frac{1}{\sigma_{\alpha t}}}{\frac{1}{\sigma_{\alpha t}}+\frac{1}{\sigma_{\beta t}}+\frac{1}{\sigma_{\gamma t}}}\right)
\end{equation}
Here, $\sigma_{\alpha t}$,$\sigma_{\beta t}$,$\sigma_{\gamma t}$
represent the volatilities of Alpha, Beta and Gamma at time $t$.
A similar formula applies for Beta and Gamma.
\end{itemize}
\begin{figure}[H]
\includegraphics[width=18cm]{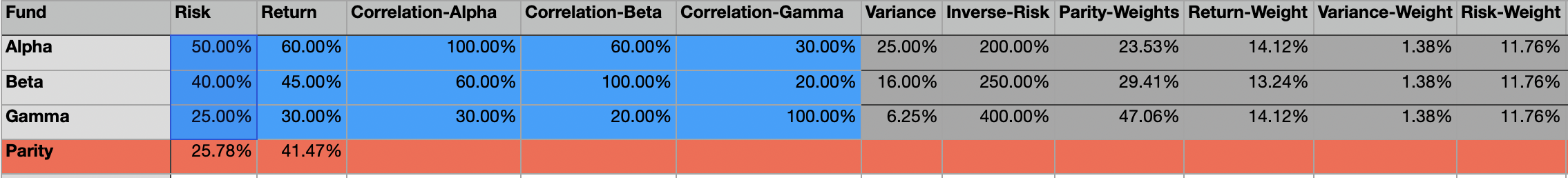}\caption{\label{fig:Parity-Line:-ABG-Sample}Parity Line: Alpha, Beta Gamma
Combination Sample Data}
\end{figure}

\begin{itemize}
\item The Table in Figure (\ref{fig:Parity-Line:-ABG-Real-Market-Data})
shows numerical examples using actual market data to combine Alpha,
Beta and Gamma to get the Parity portfolio using the method in Section
(\ref{subsec:Distancing-the-Distance}). Investors choose risk, return
or directly the weights of Alpha and Beta. The numerical values presented
below are from November 2022. The risk and return figures were estimated
on a rolling 90 day basis. The security compositions of Alpha, Beta
and Gamma are available upon request.
\end{itemize}
\begin{figure}[H]
\includegraphics[width=18cm]{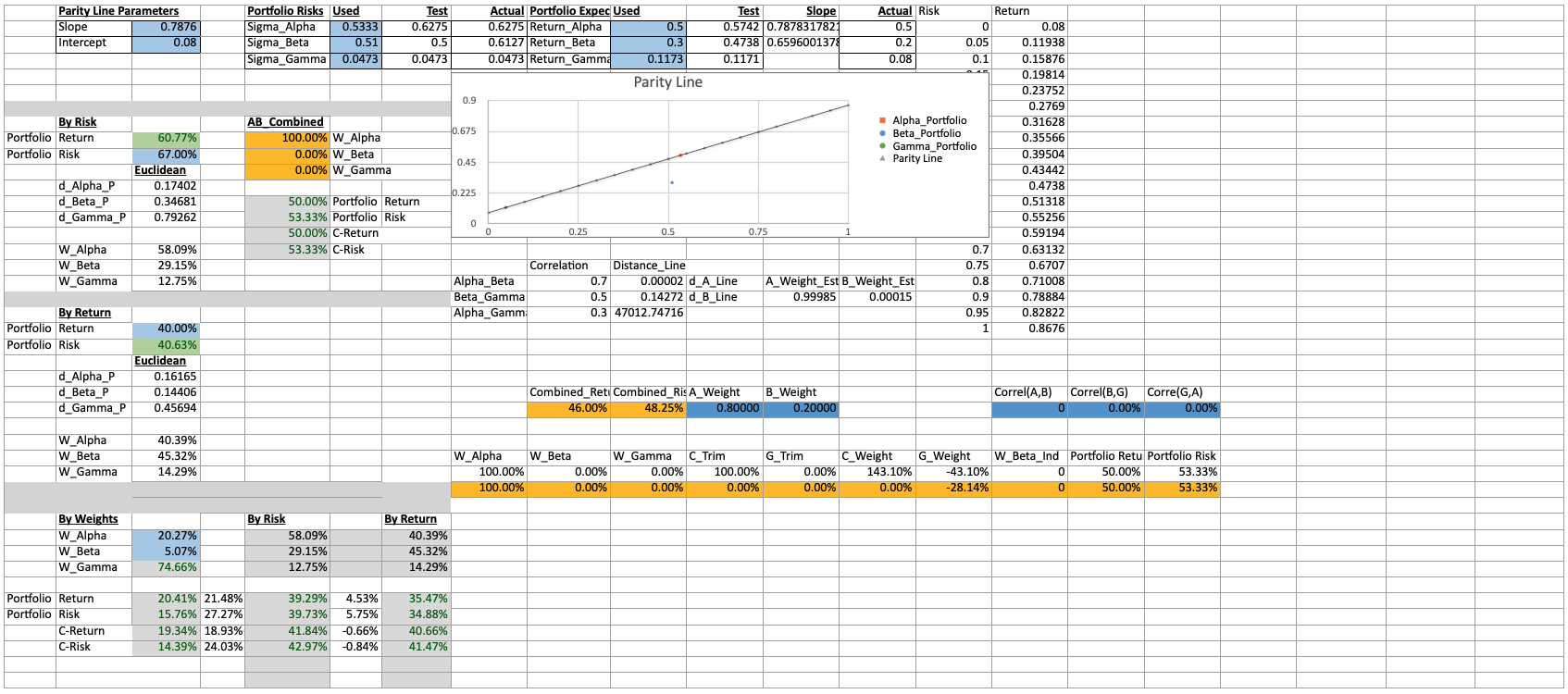}\caption{\label{fig:Parity-Line:-ABG-Real-Market-Data}Parity Line: Alpha,
Beta Gamma Combination Real Market Data}
\end{figure}

\begin{itemize}
\item The Table in Figure (\ref{fig:Parity-Line:-Return-Scenaios}) shows
numerical examples illustrating scenarios when investor return choices
will vary and the portfolio risk and return are calculated based on
the formulations in Section (\ref{subsec:Distancing-the-Distance}). 
\begin{itemize}
\item The columns in Figure (\ref{fig:Parity-Line:-Return-Scenaios}) represent
the following information respectively: 
\begin{enumerate}
\item \textbf{Weight\_Alpha} corresponds to the calculated weight of Alpha.
\item The other columns have similar nomenclatures and meaning.
\end{enumerate}
\end{itemize}
\end{itemize}
\begin{figure}[H]
\includegraphics[width=18cm]{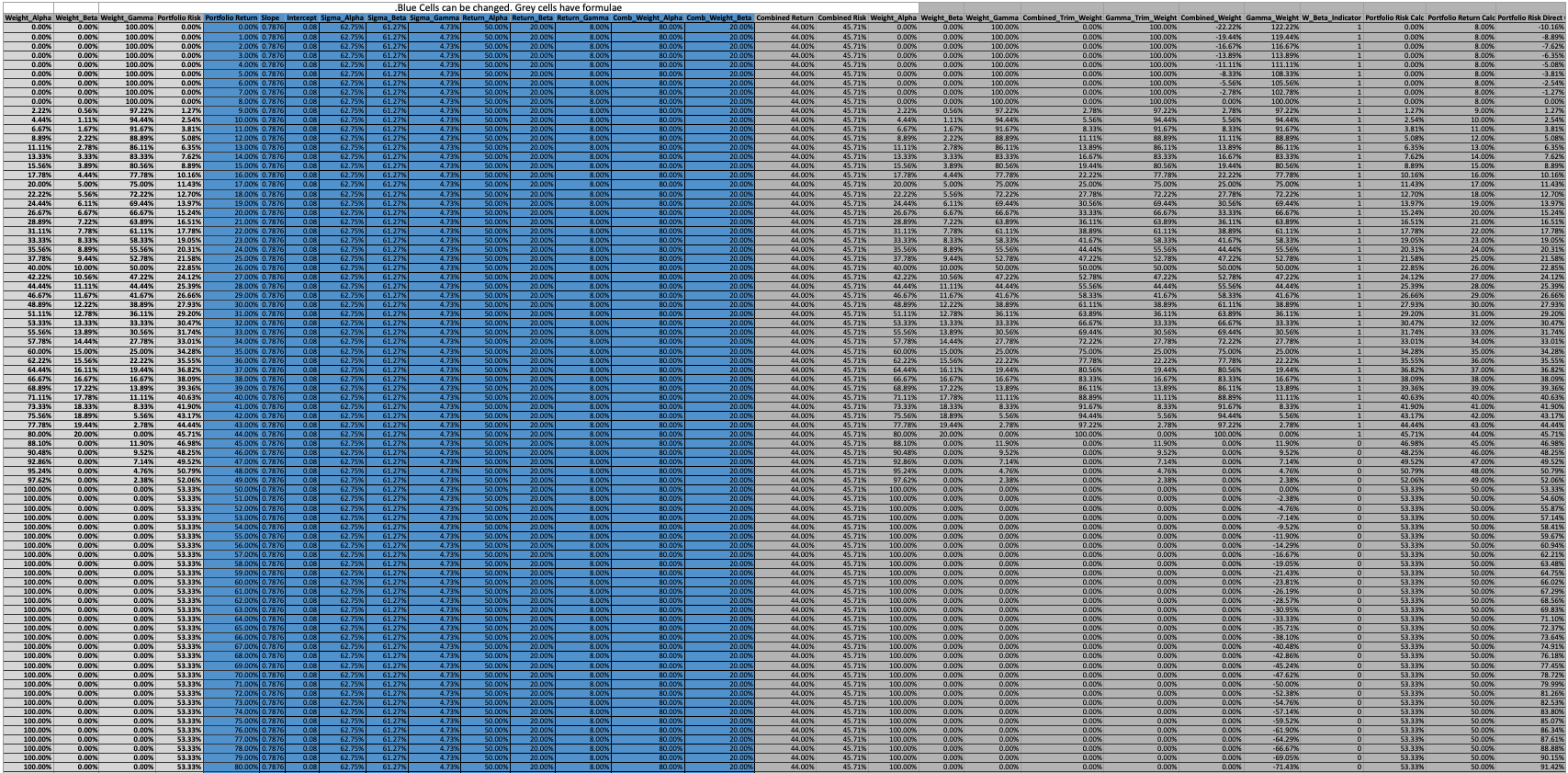}\caption{\label{fig:Parity-Line:-Return-Scenaios}Parity Line: Return Scenario
Analysis}
\end{figure}

\begin{itemize}
\item The Table in Figure (\ref{fig:Parity-Line:-Risk-Scenarios}) shows
numerical examples illustrating scenarios when investor risk choices
will vary and the portfolio risk and return are calculated based on
the formulations in Section (\ref{subsec:Distancing-the-Distance}). 
\begin{itemize}
\item The columns in Figure (\ref{fig:Parity-Line:-Risk-Scenarios}) represent
the following information respectively: 
\begin{enumerate}
\item \textbf{Weight\_Alpha} corresponds to the calculated weight of Alpha.
\item The other columns have similar nomenclatures and meaning.
\end{enumerate}
\end{itemize}
\end{itemize}
\begin{figure}[H]
\includegraphics[width=18cm]{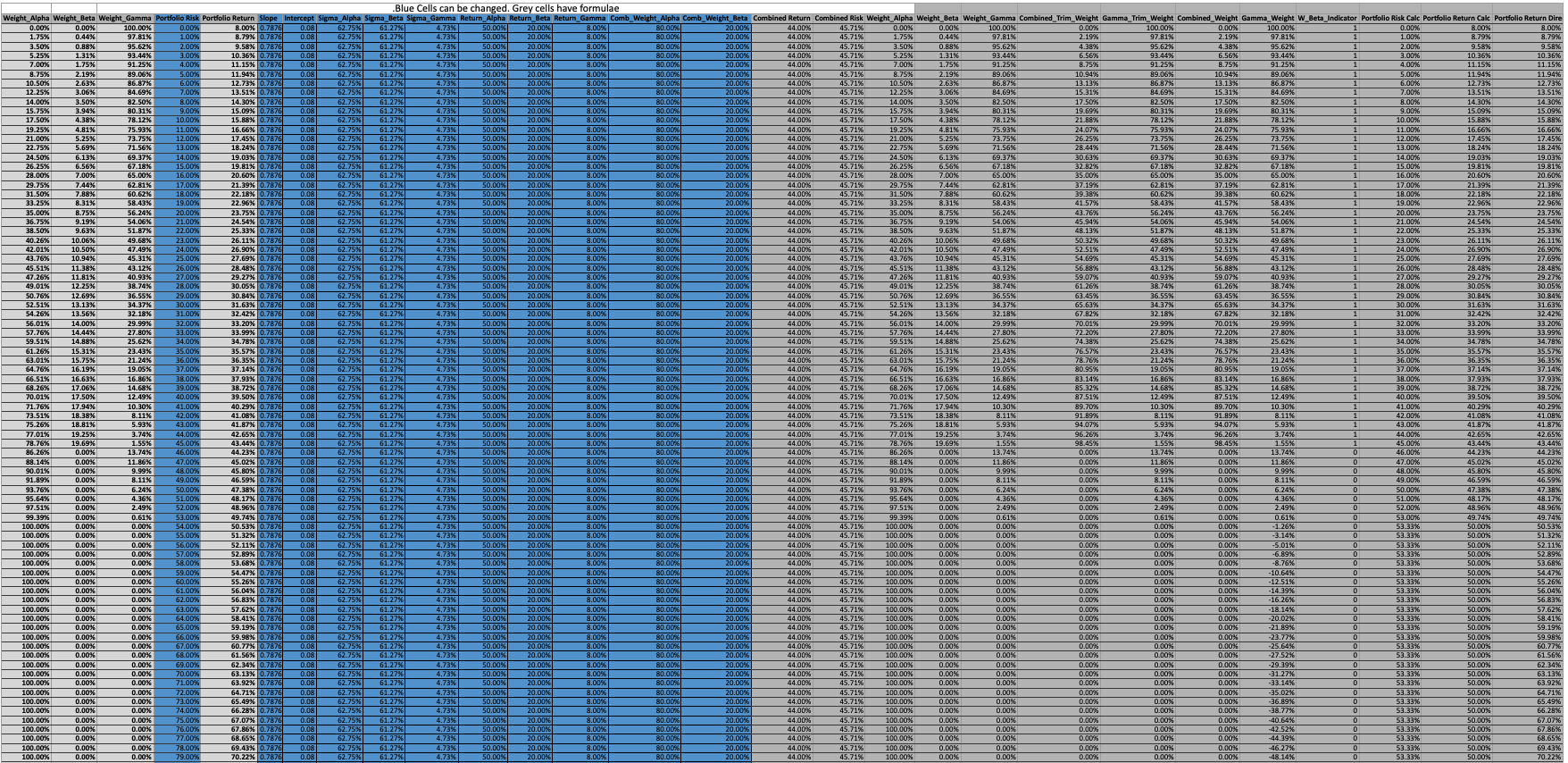}\caption{\label{fig:Parity-Line:-Risk-Scenarios}Parity Line: Risk Scenario
Analysis}
\end{figure}

\begin{itemize}
\item The Table in Figure (\ref{fig:Parity-Line:-Weights-Return-Scenaios})
shows numerical examples illustrating scenarios when investor weight
choices will vary and the portfolio return is incremented in steps
of one percent The portfolio risk calculations are based on the formulations
in Section (\ref{subsec:Distancing-the-Distance}). 
\begin{itemize}
\item The columns in Figure (\ref{fig:Parity-Line:-Weights-Return-Scenaios})
represent the following information respectively: 
\begin{enumerate}
\item \textbf{Weight\_Alpha} corresponds to the calculated weight of Alpha.
\item The other columns have similar nomenclatures and meaning.
\end{enumerate}
\end{itemize}
\end{itemize}
\begin{figure}[H]
\includegraphics[width=18cm]{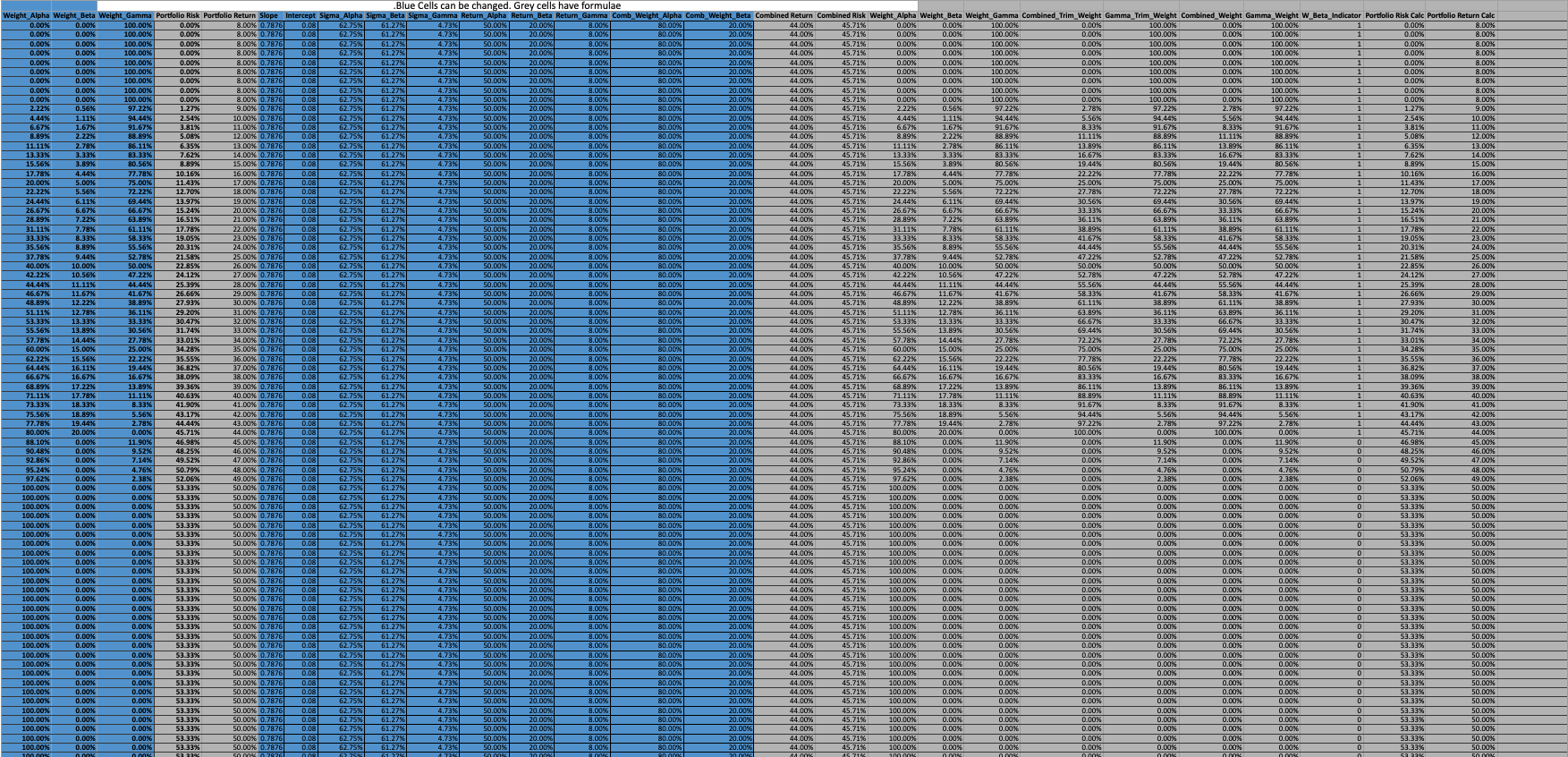}\caption{\label{fig:Parity-Line:-Weights-Return-Scenaios}Parity Line: Weights
and Return Scenario Analysis}
\end{figure}

\begin{itemize}
\item The Table in Figure (\ref{fig:Parity-Line:-Weights-Risk-Scenaios})
shows numerical examples illustrating scenarios when investor weight
choices will vary and the portfolio risk is incremented in steps of
one percent The portfolio return calculations are based on the formulations
in Section (\ref{subsec:Distancing-the-Distance}). 
\begin{itemize}
\item The columns in Figure (\ref{fig:Parity-Line:-Weights-Risk-Scenaios})
represent the following information respectively: 
\begin{enumerate}
\item \textbf{Weight\_Alpha} corresponds to the calculated weight of Alpha.
\item The other columns have similar nomenclatures and meaning.
\end{enumerate}
\end{itemize}
\end{itemize}
\begin{figure}[H]
\includegraphics[width=16cm,height=10cm]{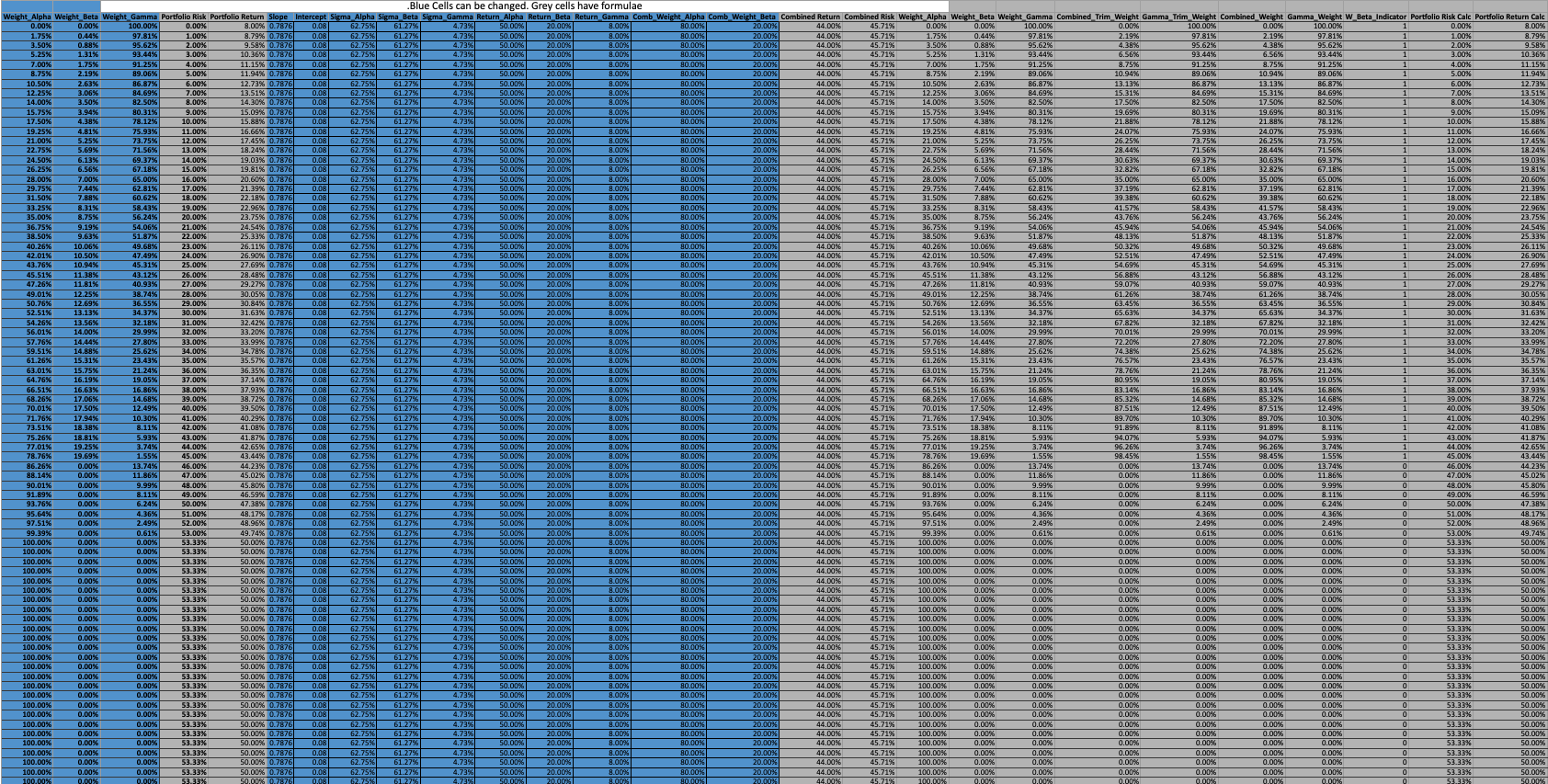}\caption{\label{fig:Parity-Line:-Weights-Risk-Scenaios}Parity Line: Weights
and Risk Scenario Analysis}
\end{figure}

\begin{itemize}
\item The Table in Figure (\ref{fig:Parity-Sequence-Steps-Investor-Variables-I})
shows numerical examples related to the method described in 
\begin{itemize}
\item The columns in Figure (\ref{fig:Parity-Sequence-Steps-Investor-Variables-I})
represent the following information respectively: 
\begin{enumerate}
\item \textbf{Time Step} corresponds to the rebalancing event when the following
actions happen by the corresponding investors and the state of their
investment and the fund prices.
\item The other columns have similar nomenclatures and meaning.
\end{enumerate}
\end{itemize}
\end{itemize}
\begin{figure}[H]
\includegraphics[width=18cm]{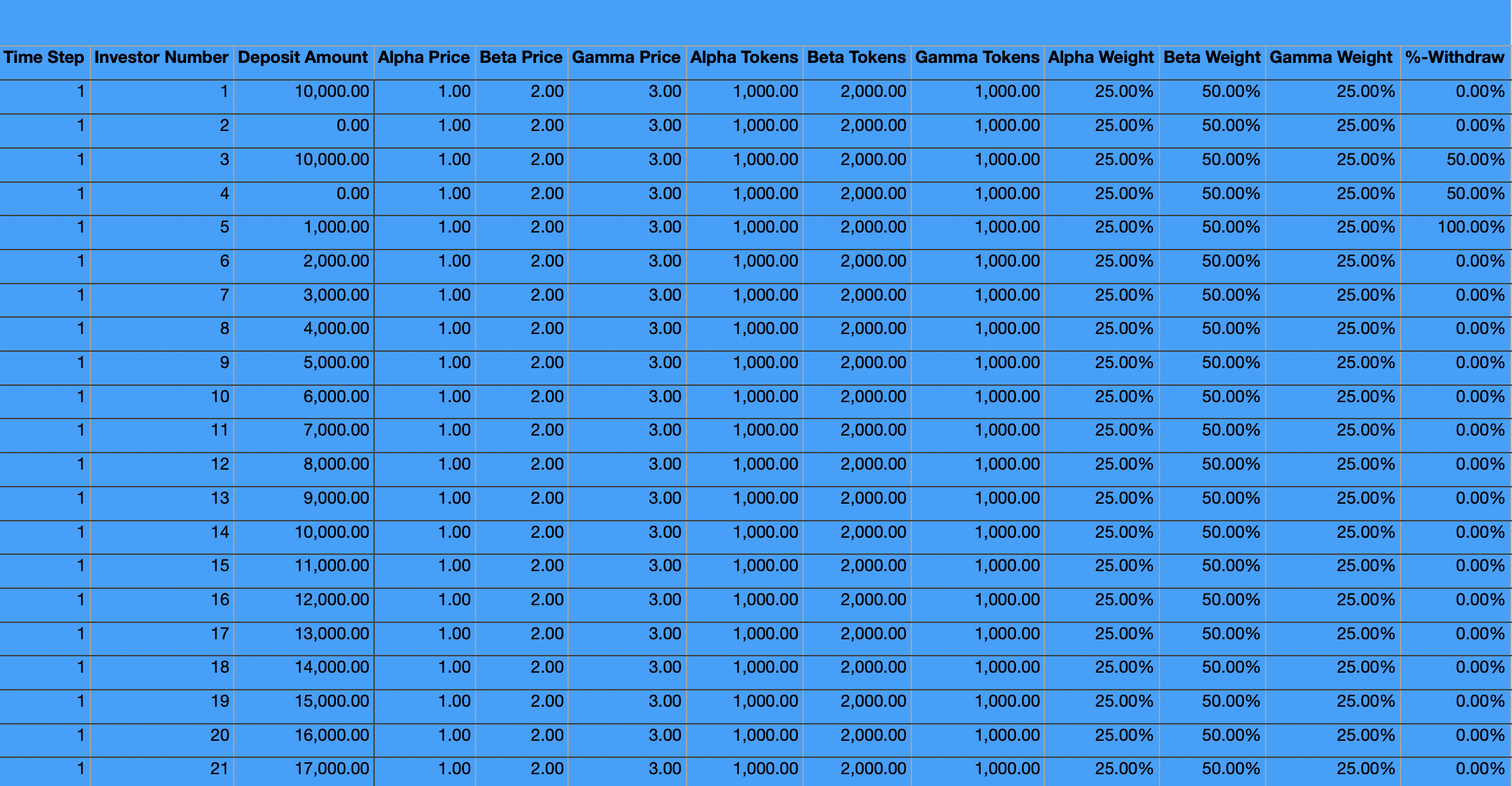}\caption{\label{fig:Parity-Sequence-Steps-Investor-Variables-I}Parity Sequence
of Steps: Investor Level Variables}
\end{figure}

\begin{itemize}
\item The Table in Figure (\ref{fig:Parity-Sequence-Steps-Investor-Amounts-Tokens-II})
shows numerical examples related to the method described in 
\begin{itemize}
\item The columns in Figure (\ref{fig:Parity-Sequence-Steps-Investor-Amounts-Tokens-II})
represent the following information respectively: 
\begin{enumerate}
\item \textbf{\$DepositPlusTokens} corresponds to the dollar value of the
deposit plus tokens for the corresponding investor.
\item The other columns have similar nomenclatures and meaning.
\end{enumerate}
\end{itemize}
\end{itemize}
\begin{figure}[H]
\includegraphics[width=18cm]{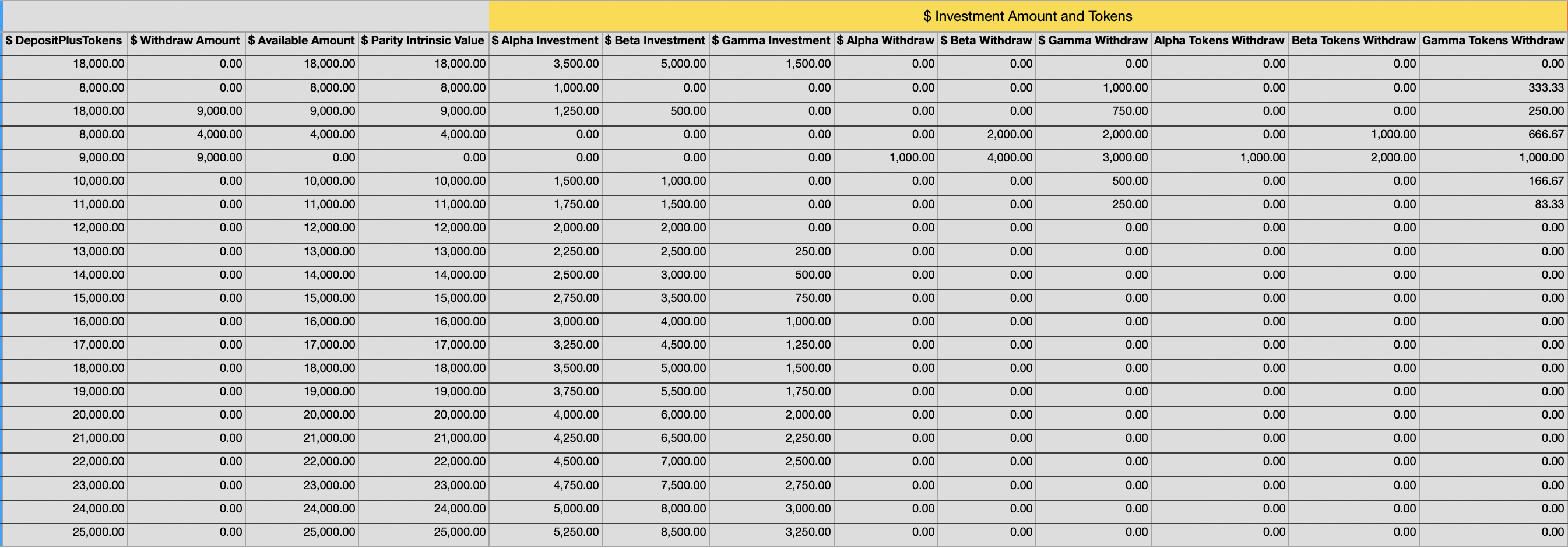}\caption{\label{fig:Parity-Sequence-Steps-Investor-Amounts-Tokens-II}Parity
Sequence of Steps: Investor Level Amounts and Tokens}
\end{figure}

\begin{itemize}
\item The Table in Figure (\ref{fig:Parity-Sequence-Steps-Investor-Raw-III})
shows numerical examples related to the method described in 
\begin{itemize}
\item The columns in Figure (\ref{fig:Parity-Sequence-Steps-Investor-Raw-III})
represent the following information respectively: 
\begin{enumerate}
\item \textbf{\$AIR} is the raw value of the investment needed in Alpha
that does not check if a corresponding withdraw is also happened for
the investor in this row. Notice \textbf{\$AIR} is positive and \textbf{\$AWR}
is negative in the first row. Hence only a net investment will happen
into Alpha for this investor. 
\item \textbf{\$ATW} is the raw value of the number of tokens that need
to withdrawn from this investor's allocation. In the first row the
values are negative hence this investor only receives a net deposit
of tokens - and no withdraws - based on his chosen risk and return
preferences.
\item The other columns have similar nomenclatures and meaning.
\end{enumerate}
\end{itemize}
\end{itemize}
\begin{figure}[H]
\includegraphics[width=18cm]{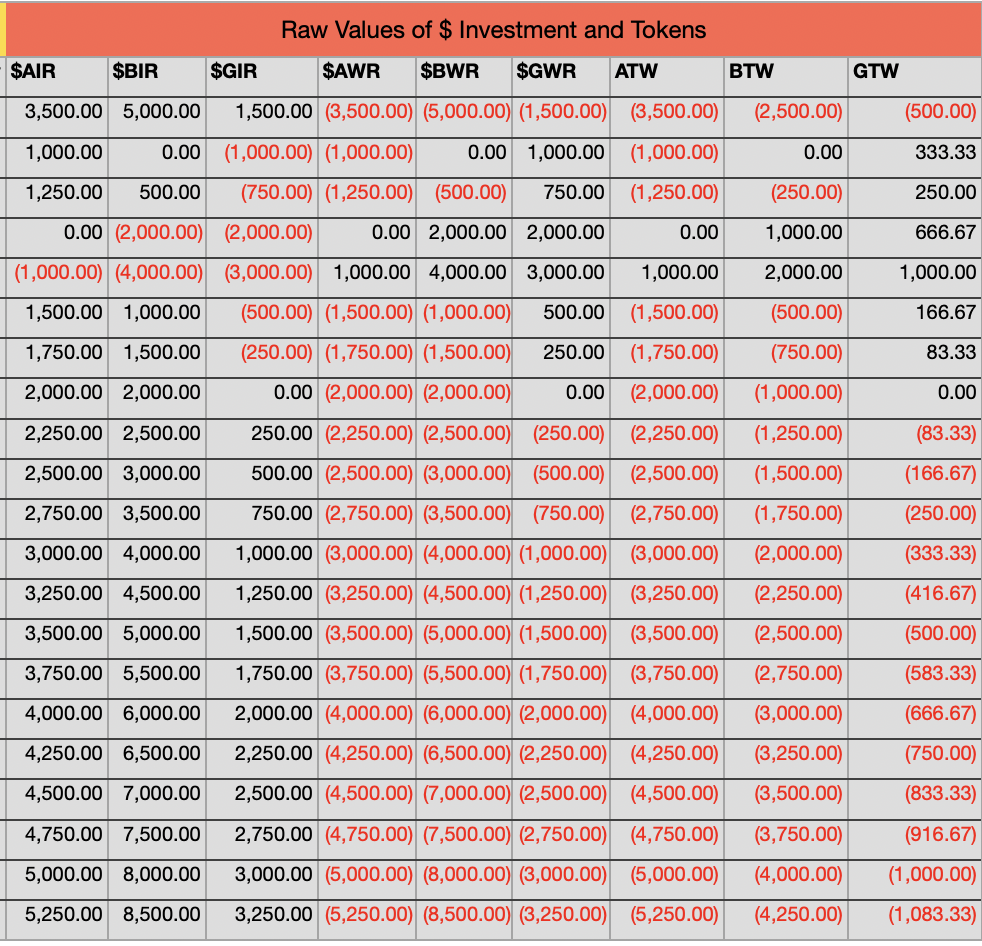}\caption{\label{fig:Parity-Sequence-Steps-Investor-Raw-III}Parity Sequence
of Steps: Investor Level Raw Values}
\end{figure}

\begin{itemize}
\item The Table in Figure (\ref{fig:Parity-Sequence-Steps-Aggregate-I-II-III})
shows numerical examples related to the method described in Point
(\ref{enu:Parity-Fund-Flow-Identity}) in Algorithm (\ref{alg:Algorithm-Parity Sequence of Steps})
in Section (\ref{subsec:Parity-Sequence}). Notice that the Parity
aggregate identity from Point (\ref{enu:Parity-Fund-Flow-Identity})
in Algorithm (\ref{alg:Algorithm-Parity Sequence of Steps}) in Section
(\ref{subsec:Parity-Sequence}) is satisfied.
\begin{itemize}
\item The columns in Figure (\ref{fig:Parity-Sequence-Steps-Aggregate-I-II-III})
represent the following information respectively: 
\begin{enumerate}
\item \textbf{\$Deposit} corresponds to the total dollar value of the deposit
received by the parity fund from all investors.
\item The columns in Figure (\ref{fig:Parity-Sequence-Steps-Aggregate-I-II-III})
have similar nomenclatures and meaning as the investor level illustrations
in Figures (\ref{fig:Parity-Sequence-Steps-Investor-Variables-I};
\ref{fig:Parity-Sequence-Steps-Investor-Amounts-Tokens-II}; \ref{fig:Parity-Sequence-Steps-Investor-Raw-III}).
\end{enumerate}
\end{itemize}
\end{itemize}
\begin{figure}[H]
\includegraphics[width=18cm]{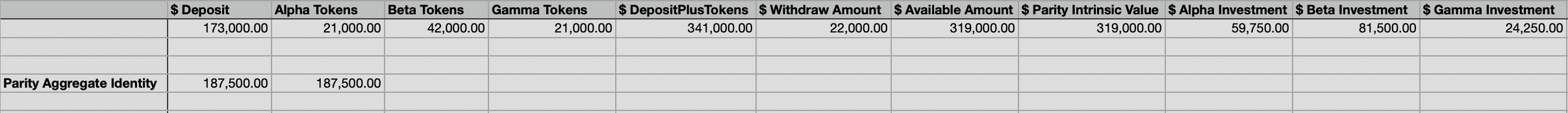}

\includegraphics[width=18cm]{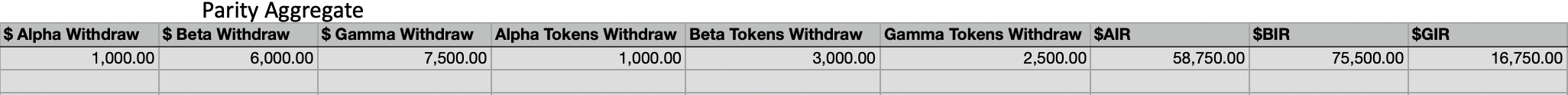}

\includegraphics[width=18cm]{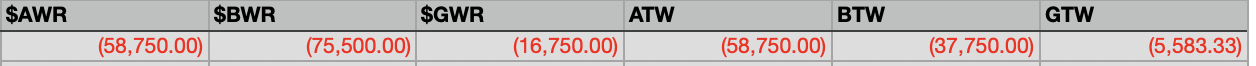}\caption{\label{fig:Parity-Sequence-Steps-Aggregate-I-II-III}Parity Sequence
of Steps: Aggregate Deposits, Withdraws and Raw Values}
\end{figure}

\begin{itemize}
\item The Table in Figure (\ref{fig:Parity-Line-Parabolic-Curve-Points})
shows numerical examples related to the method described in Section
(\ref{subsec:Efficient-Frontier-Parabolic}). 
\begin{itemize}
\item The cell names and values in Figure (\ref{fig:Parity-Line-Parabolic-Curve-Points})
represent variables from Section (\ref{subsec:Efficient-Frontier-Parabolic}).
\begin{enumerate}
\item \textbf{Slope (m)} corresponds to the slope of the Parity Line as
discussed in Section (\ref{subsec:Efficient-Frontier-Parabolic}).
\item The other cell names and values in Figure (\ref{fig:Parity-Line-Parabolic-Curve-Points})
have similar nomenclatures and meaning as the variables outlined in
Section (\ref{subsec:Efficient-Frontier-Parabolic}).
\end{enumerate}
\end{itemize}
\end{itemize}
\begin{figure}[H]
\includegraphics[width=18cm]{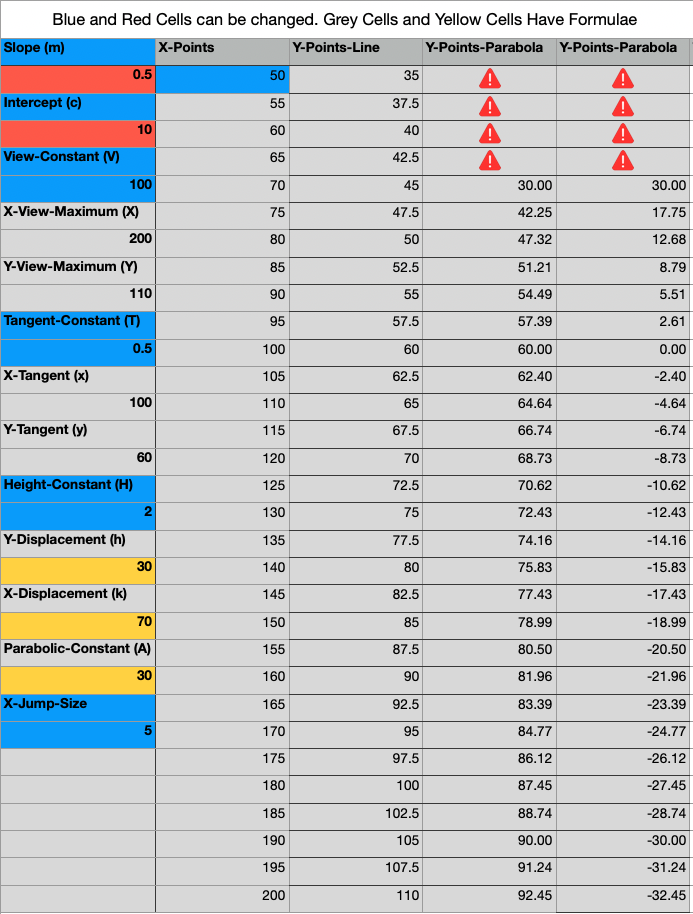}

\caption{\label{fig:Parity-Line-Parabolic-Curve-Points}Parity Line and Parabolic
Curve: X-Axis and Y-Axis Numerical Points}
\end{figure}

\begin{itemize}
\item The Illustrations in Figure (\ref{fig:Parity-Line-Parabolic-Curve-Plots})
shows graphs related to the method described in Section (\ref{subsec:Efficient-Frontier-Parabolic})
corresponding to the points given in Figure (\ref{fig:Parity-Line-Parabolic-Curve-Points}). 
\begin{itemize}
\item The X-Axis and Y-Axis in Figure (\ref{fig:Parity-Line-Parabolic-Curve-Plots})
represent the risk and return values respectively.
\end{itemize}
\end{itemize}
\begin{figure}[H]
\includegraphics[width=18cm]{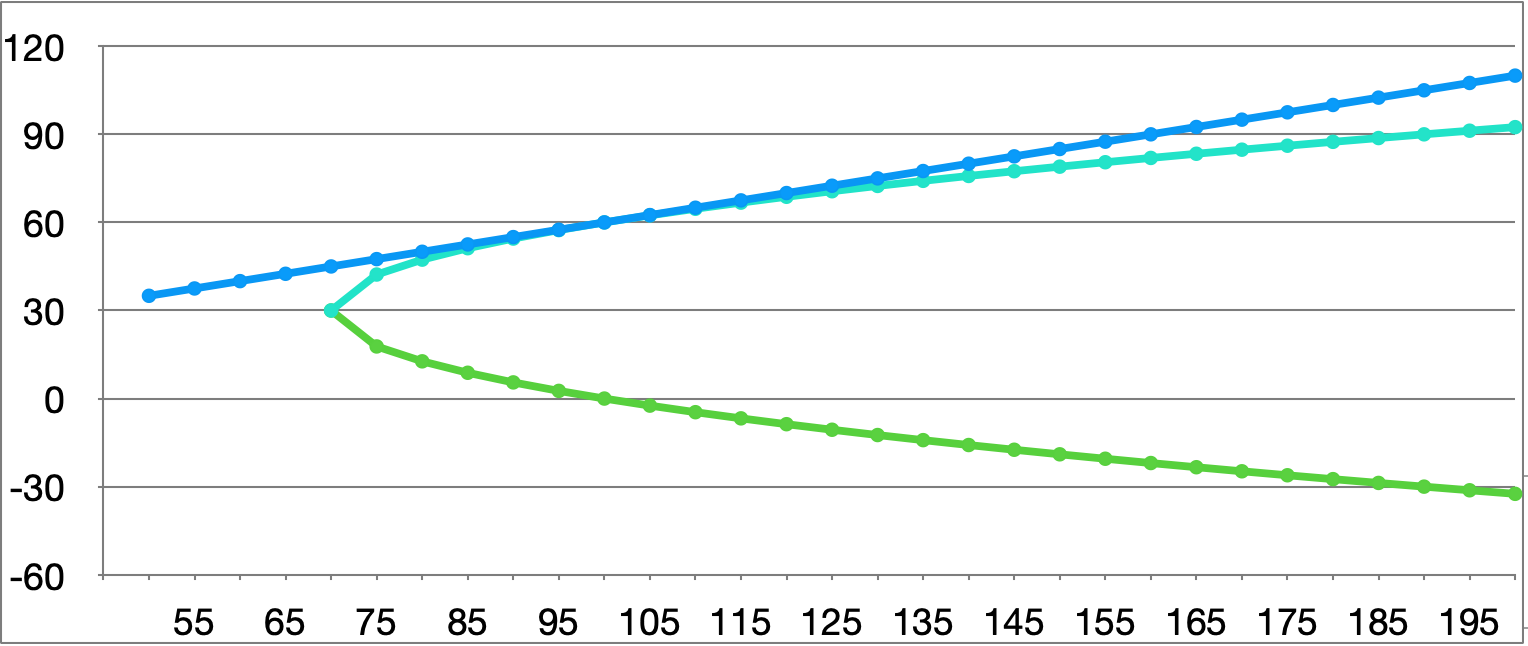}

\caption{\label{fig:Parity-Line-Parabolic-Curve-Plots}Parity Line and Parabolic
Curve: Graphical Plots}
\end{figure}

\section{\label{sec:End-notes}Appendix of Explanations and End-notes }
\begin{enumerate}
\item \label{enu:Efficient-Frontier}In modern portfolio theory, the efficient
frontier (or portfolio frontier) is an investment portfolio which
occupies the \textquotedbl efficient\textquotedbl{} parts of the
risk–return spectrum. Formally, it is the set of portfolios which
satisfy the condition that no other portfolio exists with a higher
expected return but with the same standard deviation of return (i.e.,
the risk). \href{https://en.wikipedia.org/wiki/Efficient_frontier}{Efficient Frontier,  Wikipedia Link}
\item \label{enu:Risk-parity-(or}Risk parity (or risk premia parity) is
an approach to investment management which focuses on allocation of
risk, usually defined as volatility, rather than allocation of capital.
\href{https://en.wikipedia.org/wiki/Risk_parity}{Risk Parity,  Wikipedia Link}
\item \label{enu:Star-Trek-is}Star Trek is an American science fiction
television series created by Gene Roddenberry that follows the adventures
of the starship USS Enterprise (NCC-1701) and its crew. The show is
set in the Milky Way galaxy, c. 2266–2269. The ship and crew are led
by Captain James T. Kirk (William Shatner), First Officer and Science
Officer Spock (Leonard Nimoy), and Chief Medical Officer Leonard H.
\textquotedbl Bones\textquotedbl{} McCoy (DeForest Kelley). Shatner's
voice-over introduction during each episode's opening credits stated
the starship's purpose:
\begin{itemize}
\item Space: the final frontier. These are the voyages of the starship Enterprise.
Its five-year mission: to explore strange new worlds, to seek out
new life and new civilizations, to boldly go where no man has gone
before. \href{https://en.wikipedia.org/wiki/Star_Trek:_The_Original_Series}{Star Trek: The Original Series,  Wikipedia Link}
\end{itemize}
\item \label{enu:Decentralized-finance}Decentralized finance (often stylized
as DeFi) offers financial instruments without relying on intermediaries
such as brokerages, exchanges, or banks by using smart contracts on
a blockchain. \href{https://en.wikipedia.org/wiki/Decentralized_finance}{Decentralized Finance (DeFi), Wikipedia Link}
\item \label{enu:Types-Yield-Enhancement-Services}The following are the
four main types of blockchain decentralized financial products or
services. We can also consider them as the main types of yield enhancement,
or return generation, vehicles available in decentralized finance: 
\begin{enumerate}
\item \label{enu:Single-sided-staking-allows}Single-Sided Staking: This
allows users to earn yield by providing liquidity for one type of
asset, in contrast to liquidity provisioning on AMMs, which requires
a pair of assets. \href{https://docs.saucerswap.finance/features/single-sided-staking}{Single Sided Staking,  SuacerSwap Link}
\begin{enumerate}
\item Bancor is an example of a provider who supports single sided staking.
Bancor natively supports Single-Sided Liquidity Provision of tokens
in a liquidity pool. This is one of the main benefits to liquidity
providers that distinguishes Bancor from other DeFi staking protocols.
Typical AMM liquidity pools require a liquidity provider to provide
two assets. Meaning, if you wish to deposit \textquotedbl TKN1\textquotedbl{}
into a pool, you would be forced to sell 50\% of that token and trade
it for \textquotedbl TKN2\textquotedbl . When providing liquidity,
your deposit is composed of both TKN1 and TKN2 in the pool. Bancor
Single-Side Staking changes this and enables liquidity providers to:
Provide only the token they hold (TKN1 from the example above) Collect
liquidity providers fees in TKN1. \href{https://docs.bancor.network/about-bancor-network/faqs/single-side-liquidity}{Single Sided Staking,  Bancor Link}
\end{enumerate}
\item \label{enu:AMM-Liquidity-Pairs}AMM Liquidity Pairs (AMM LP): A constant-function
market maker (CFMM) is a market maker with the property that that
the amount of any asset held in its inventory is completely described
by a well-defined function of the amounts of the other assets in its
inventory (Hanson 2007). \href{https://en.wikipedia.org/wiki/Constant_function_market_maker}{Constant Function Market Maker,  Wikipedia Link}

This is the most common type of market maker liquidity pool. Other
types of market makers are discussed in Mohan (2022). All of them
can be grouped under the category Automated Market Makers. Hence the
name AMM Liquidity Pairs. A more general discussion of AMMs, without
being restricted only to the blockchain environment, is given in (Slamka,
Skiera \& Spann 2012).
\item \label{enu:LP-Token-Staking:}LP Token Staking: LP staking is a valuable
way to incentivize token holders to provide liquidity. When a token
holder provides liquidity as mentioned earlier in Point (\ref{enu:AMM-Liquidity-Pairs})
they receive LP tokens. LP staking allows the liquidity providers
to stake their LP tokens and receive project tokens tokens as rewards.
This mitigates the risk of impermanent loss and compensates for the
loss. \href{https://defactor.com/liquidity-provider-staking-introduction-guide/}{Liquidity Provider Staking,  DeFactor Link}
\begin{enumerate}
\item Note that this is also a type of single sided staking discussed in
Point (\ref{enu:Single-sided-staking-allows}). The key point to remember
is that the LP Tokens can be considered as receipts for the crypto
assets deposits in an AMM LP Point (\ref{enu:AMM-Liquidity-Pairs}).
These LP Token receipts can be further staked to generate additional
yield.
\end{enumerate}
\item \label{enu:Lending:-Crypto-lending}Lending: Crypto lending is the
process of depositing cryptocurrency that is lent out to borrowers
in return for regular interest payments. Payments are typically made
in the form of the cryptocurrency that is deposited and can be compounded
on a daily, weekly, or monthly basis. \href{https://www.investopedia.com/crypto-lending-5443191}{Crypto Lending,  Investopedia Link};
\href{https://defiprime.com/decentralized-lending}{DeFi Lending,  DeFiPrime Link};
\href{https://crypto.com/price/categories/lending}{Top Lending Coins by Market Capitalization,  Crypto.com Link}.
\begin{enumerate}
\item Crypto lending is very common on decentralized finance projects and
also in centralized exchanges. Centralized cryptocurrency exchanges
are online platforms used to buy and sell cryptocurrencies. They are
the most common means that investors use to buy and sell cryptocurrency
holdings. \href{https://www.investopedia.com/tech/what-are-centralized-cryptocurrency-exchanges/}{Centralized Cryptocurrency Exchanges,  Investopedia Link}
\item Lending is a very active area of research both on blockchain and off
chain (traditional finance) as well (Cai 2018; Zeng et al., 2019;
Bartoletti, Chiang \& Lafuente 2021; Gonzalez 2020; Hassija et al.,
2020; Patel et al. , 2020). 
\item Lending is also a highly profitable business in the traditional financial
world (Kashyap 2022-I). Investment funds, especially hedge funds,
engage in borrowing securities to put on short positions depending
on their investment strategies.Long only investment funds typically
supply securities or lend their assets for a fee.
\item In finance, a long position in a financial instrument means the holder
of the position owns a positive amount of the instrument. \href{https://en.wikipedia.org/wiki/Long_(finance)}{Long Position in Finance,  Wikipedia Link}
\item In finance, being short in an asset means investing in such a way
that the investor will profit if the value of the asset falls. This
is the opposite of a more conventional \textquotedbl long\textquotedbl{}
position, where the investor will profit if the value of the asset
rises. \href{https://en.wikipedia.org/wiki/Short_(finance)}{Short Position in Finance,  Wikipedia Link}
\end{enumerate}
\end{enumerate}
\item \label{enu:Modern-portfolio-theory}Modern portfolio theory (MPT),
or mean-variance analysis, is a mathematical framework for assembling
a portfolio of assets such that the expected return is maximized for
a given level of risk. \href{https://en.wikipedia.org/wiki/Modern_portfolio_theory}{Modern Portfolio Theory,  Wikipedia Link}
\item \label{enu:The-semivariance-is}The semivariance is defined as the
expected squared deviation from the mean, calculated over those points
that are no greater than the mean. Its square root is the semi-deviation.
\href{https://en.wikipedia.org/wiki/Downside_risk}{Downside Risk (Semi-Variance),  Wikipedia Link}
\item \label{enu:Fuzzy-logic-is}Fuzzy logic is a form of many-valued logic
in which the truth value of variables may be any real number between
0 and 1. It is employed to handle the concept of partial truth, where
the truth value may range between completely true and completely false.
By contrast, in Boolean logic, the truth values of variables may only
be the integer values 0 or 1. \href{https://en.wikipedia.org/wiki/Fuzzy_logic}{Fuzzy Logic,  Wikipedia Link}
\item \label{enu:Possibility-theory-is}Possibility theory is a mathematical
theory for dealing with certain types of uncertainty and is an alternative
to probability theory. It uses measures of possibility and necessity
between 0 and 1, ranging from impossible to possible and unnecessary
to necessary, respectively. Professor Lotfi Zadeh first introduced
possibility theory in 1978 as an extension of his theory of fuzzy
sets and fuzzy logic. \href{https://en.wikipedia.org/wiki/Possibility_theory}{Possibility Theory,  Wikipedia Link}
\item \label{enu:Probability-theory-or}Probability theory or probability
calculus is the branch of mathematics concerned with probability.
The axioms of probability formalise it in terms of a probability space,
which assigns a measure taking values between 0 and 1, termed the
probability measure, to a set of outcomes called the sample space.
Any specified subset of the sample space is called an event.\href{https://en.wikipedia.org/wiki/Probability_theory}{Probability Theory,  Wikipedia Link}
\item \label{enu:Bayesian-statistics-is}Bayesian statistics is a theory
in the field of statistics based on the Bayesian interpretation of
probability where probability expresses a degree of belief in an event.
The degree of belief may be based on prior knowledge about the event,
such as the results of previous experiments, or on personal beliefs
about the event. \href{https://en.wikipedia.org/wiki/Bayesian_statistics}{Bayesian Statistics,  Wikipedia Link}
\begin{itemize}
\item Bayesian statistical methods use Bayes' theorem to compute and update
probabilities after obtaining new data. Bayes' theorem describes the
conditional probability of an event based on data as well as prior
information or beliefs about the event or conditions related to the
event.
\end{itemize}
\item \label{enu:Value-at-risk}Value at risk (VaR) is a measure of the
risk of loss of investment/Capital. It estimates how much a set of
investments might lose (with a given probability), given normal market
conditions, in a set time period such as a day. \href{https://en.wikipedia.org/wiki/Value_at_risk}{Value at Risk (VaR),  Wikipedia Link}
\item \label{enu:Behavioral-finance-is}Behavioral finance is the study
of the influence of psychology on the behavior of investors or financial
analysts. It assumes that investors are not always rational, have
limits to their self-control and are influenced by their own biases.\href{https://en.wikipedia.org/wiki/Behavioral_economics\#Behavioral_finance}{Behavioral Finance,  Wikipedia Link}
\begin{enumerate}
\item Behavioral economics is the study of the psychological, cognitive,
emotional, cultural and social factors involved in the decisions of
individuals or institutions, and how these decisions deviate from
those implied by classical economic theory. \href{https://en.wikipedia.org/wiki/Behavioral_economics}{Behavioral Economics,  Wikipedia Link}
\end{enumerate}
\item \label{enu:In-finance,-leverage}In finance, leverage (or gearing
in the United Kingdom and Australia) is any technique involving borrowing
funds to buy an investment, estimating that future profits will be
more than the cost of borrowing. This technique is named after a lever
in physics, which amplifies a small input force into a greater output
force, because successful leverage amplifies the smaller amounts of
money needed for borrowing into large amounts of profit. \href{https://en.wikipedia.org/wiki/Leverage_(finance)}{Leverage (Finance),  Wikipedia Link}
\item \label{enu:In-finance,-risk-factors}In finance, risk factors are
the building blocks of investing, that help explain the systematic
returns in equity market, and the possibility of losing money in investments
or business adventures. \href{https://en.wikipedia.org/wiki/Risk_factor_(finance)}{Risk Factor (Finance),  Wikipedia Link}
\item \label{enu:In-financial-Asset-Pricing}In financial economics, asset
pricing refers to a formal treatment and development of two interrelated
pricing principles - either general equilibrium asset pricing or rational
asset pricing. \href{https://en.wikipedia.org/wiki/Asset_pricing}{Asset Pricing,  Wikipedia Link}
\item \label{enu:Convex-optimization-is}Convex optimization is a subfield
of mathematical optimization that studies the problem of minimizing
convex functions over convex sets (or, equivalently, maximizing concave
functions over convex sets). \href{https://en.wikipedia.org/wiki/Convex_optimization}{Convex Optimization,  Wikipedia Link}
\item \label{enu:The-P-versus-NP}The P versus NP problem is a major unsolved
problem in theoretical computer science. In informal terms, it asks
whether every problem whose solution can be quickly verified can also
be quickly solved. \href{https://en.wikipedia.org/wiki/P_versus_NP_problem}{P versus NP problem,  Wikipedia Link}
\item \label{enu:In-computational-complexity-NP-Hardness}In computational
complexity theory, NP-hardness (non-deterministic polynomial-time
hardness) is the defining property of a class of problems that are
informally \textquotedbl at least as hard as the hardest problems
in NP\textquotedbl . \href{https://en.wikipedia.org/wiki/NP-hardness}{NP Hardness,  Wikipedia Link}
\item \label{enu:In-computational-complexity-NP}In computational complexity
theory, NP (nondeterministic polynomial time) is a complexity class
used to classify decision problems. NP is the set of decision problems
for which the problem instances, where the answer is \textquotedbl yes\textquotedbl ,
have proofs verifiable in polynomial time by a deterministic Turing
machine, or alternatively the set of problems that can be solved in
polynomial time by a nondeterministic Turing machine. \href{https://en.wikipedia.org/wiki/NP_(complexity)}{NP (Complexity),  Wikipedia Link}
\item \label{enu:Principal-component-analysis}Principal component analysis
(PCA) is a popular technique for analyzing large datasets containing
a high number of dimensions/features per observation, increasing the
interpretability of data while preserving the maximum amount of information,
and enabling the visualization of multidimensional data. Formally,
PCA is a statistical technique for reducing the dimensionality of
a dataset. \href{https://en.wikipedia.org/wiki/Principal_component_analysis}{Principal Component Analysis,  Wikipedia Link}
\item \label{enu:A-graphical-user-GUI}A graphical user interface, or GUI,
is a form of user interface that allows users to interact with electronic
devices through graphical icons and visual indicators such as secondary
notation. \href{https://en.wikipedia.org/wiki/Graphical_user_interface}{Graphical User Interface,  Wikipedia Link}
\item \label{enu:A-smart-contract}A smart contract is a computer program
or a transaction protocol that is intended to automatically execute,
control or document events and actions according to the terms of a
contract or an agreement. The objectives of smart contracts are the
reduction of need for trusted intermediators, arbitration costs, and
fraud losses, as well as the reduction of malicious and accidental
exceptions. Smart contracts are commonly associated with cryptocurrencies,
and the smart contracts introduced by Ethereum are generally considered
a fundamental building block for decentralized finance (DeFi) and
NFT applications. \href{https://en.wikipedia.org/wiki/Smart_contract}{Smart Contract,  Wikipedia Link}
\item \label{enu:A-non-fungible-token}A non-fungible token (NFT) is a unique
digital identifier that is recorded on a blockchain, and is used to
certify ownership and authenticity. It cannot be copied, substituted,
or subdivided. The ownership of an NFT is recorded in the blockchain
and can be transferred by the owner, allowing NFTs to be sold and
traded. \href{https://en.wikipedia.org/wiki/Non-fungible_token}{Non-Fungible Token (NFT),  Wikipedia Link}
\item \label{enu:In-mathematics-Metric-Space}In mathematics, a metric space
is a set together with a notion of distance between its elements,
usually called points. The distance is measured by a function called
a metric or distance function. \href{https://en.wikipedia.org/wiki/Metric_space}{Metric Space,  Wikipedia Link}
\item \label{enu:Stablecoin}A Stablecoin is a type of cryptocurrency where
the value of the digital asset is supposed to be pegged to a reference
asset, which is either fiat money, exchange-traded commodities (such
as precious metals or industrial metals), or another cryptocurrency.
\href{https://en.wikipedia.org/wiki/Stablecoin}{Stable Coin,  Wikipedia Link}
\item \label{EN:Net-Asset-Value}Net Asset Value is the net value of an
investment fund's assets less its liabilities, divided by the number
of shares outstanding. \href{https://www.investopedia.com/terms/n/nav.asp}{NAV,  Investopedia Link}
\item \label{enu:In-mathematics,-Parabola}In mathematics, a parabola is
a plane curve which is mirror-symmetrical and is approximately U-shaped.
It fits several superficially different mathematical descriptions,
which can all be proved to define exactly the same curves. \href{https://en.wikipedia.org/wiki/Parabola}{Parabola,  Wikipedia Link}
\begin{enumerate}
\item One description of a parabola involves a point (the focus) and a line
(the directrix). The focus does not lie on the directrix. The parabola
is the locus of points in that plane that are equidistant from the
directrix and the focus. 
\end{enumerate}
\begin{doublespace}
\item \label{enu:Sharpe-Ratio}In finance, the Sharpe ratio (also known
as the Sharpe index, the Sharpe measure, and the reward-to-variability
ratio) is a way to examine the performance of an investment by adjusting
for its risk. The ratio measures the excess return (or risk premium)
per unit of deviation in an investment asset or a trading strategy,
typically referred to as risk, named after William F. Sharpe. \href{https://en.wikipedia.org/wiki/Sharpe_ratio}{Sharpe Ratio, Wikipedia Link}
\end{doublespace}
\item \label{enu:In-probability-theory-Skew}In probability theory and statistics,
skewness is a measure of the asymmetry of the probability distribution
of a real-valued random variable about its mean. The skewness value
can be positive, zero, negative, or undefined. \href{https://en.wikipedia.org/wiki/Skewness}{Skewness,  Wikipedia Link}
\item \label{enu:In-probability-theory-Kurtosis}In probability theory and
statistics, kurtosis (from Greek: kyrtos or kurtos, meaning \textquotedbl curved,
arching\textquotedbl ) is a measure of the \textquotedbl tailedness\textquotedbl{}
of the probability distribution of a real-valued random variable.
Like skewness, kurtosis describes a particular aspect of a probability
distribution. There are different ways to quantify kurtosis for a
theoretical distribution, and there are corresponding ways of estimating
it using a sample from a population. \href{https://en.wikipedia.org/wiki/Kurtosis}{Kurtosis,  Wikipedia Link}
\item \label{enu:Gas-Fee}Gas is a unit of account within the EVM used in
the calculation of the transaction fee, which is the amount of ETH
a transaction's sender must pay to the network to have the transaction
included in the blockchain. \href{https://en.wikipedia.org/wiki/Ethereum\#Gas}{Ethereum (Gas Fees) Wikipedia Link}
\item \label{enu:A-decentralized-autonomous-DAO}A decentralized autonomous
organization (DAO), sometimes called a decentralized autonomous corporation
(DAC), is an organization managed in whole or in part by decentralized
computer program, with voting and finances handled through a blockchain.
In general terms, DAOs are member-owned communities without centralized
leadership. The precise legal status of this type of business organization
is unclear. \href{https://en.wikipedia.org/wiki/Decentralized_autonomous_organization}{Decentralized Autonomous Organization (DAO),  Wikipedia Link}
\item \label{enu:Ethereum,-which-was}Ethereum, which was conceived in 2013
and launched in 2015 (Wood 2014; Tapscott \& Tapscott 2016; Dannen
2017), provided a remarkable innovation in terms of making blockchain
based systems Turing complete (or theoretically being able to do what
any computer can do: Sipser 2006).
\begin{enumerate}
\item Ethereum is a decentralized, open-source blockchain with smart contract
functionality. Ether (Abbreviation: ETH) is the native cryptocurrency
of the platform. Among cryptocurrencies, ether is second only to bitcoin
in market capitalization. \href{https://en.wikipedia.org/wiki/Ethereum}{Ethereum,  Wikipedia Link}
\item \label{enu:Turing-Complete}In computability theory, a system of data-manipulation
rules (such as a computer's instruction set, a programming language,
or a cellular automaton) is said to be Turing-complete or computationally
universal if it can be used to simulate any Turing machine. \href{https://en.wikipedia.org/wiki/Turing_completeness}{Turing Completeness,  Wikipedia Link}

A Turing machine is a mathematical model of computation describing
an abstract machine that manipulates symbols on a strip of tape according
to a table of rules. Despite the model's simplicity, it is capable
of implementing any computer algorithm. \href{https://en.wikipedia.org/wiki/Turing_machine}{Turing Machine,  Wikipedia Link}
\end{enumerate}
\item \label{enu:Open-end-mutual-funds}Open-end mutual funds are purchased
from or sold to the issuer at the net asset value of each share as
of the close of the trading day in which the order was placed, as
long as the order was placed within a specified period before the
close of trading. They can be traded directly with the issuer. \href{https://en.wikipedia.org/wiki/Mutual_fund}{Mutual Fund,  Wikipedia Link}
\item \label{enu:Mutual-Funds-Fees}Mutual Funds typically pay their regular
and recurring, fund-wide operating expenses out of fund assets, rather
than by imposing separate fees and charges directly on investors.
\href{https://www.investor.gov/introduction-investing/investing-basics/glossary/mutual-fund-fees-and-expenses}{Mutual Fund Fees,   Wikipedia Link}
\item \label{enu:A-hedge-fund}A hedge fund is a pooled investment fund
that trades in relatively liquid assets and is able to make extensive
use of more complex trading, portfolio-construction, and risk management
techniques in an attempt to improve performance, such as short selling,
leverage, and derivatives. \href{https://en.wikipedia.org/wiki/Hedge_fund}{Hedge Fund,  Wikipedia Link}
\item \label{enu:An-exchange-traded-fund}An exchange-traded fund (ETF)
is a type of investment fund and exchange-traded product, i.e. they
are traded on stock exchanges. ETFs are similar in many ways to mutual
funds, except that ETFs are bought and sold from other owners throughout
the day on stock exchanges whereas mutual funds are bought and sold
from the issuer based on their price at day's end. \href{https://en.wikipedia.org/wiki/Exchange-traded_fund}{Exchange Traded Fund,  Wikipedia Link}
\end{enumerate}

\section{\label{sec:Appendix-of-Proofs}Appendix of Proofs of Propositions}

\subsection{\label{subsec:Parity-Line-Expected}Parity Line Expected Return}
\begin{proof}
Proof of Proposition (\ref{prop:Combined-Weights-Expected-Return})
is given here. Using the weights - indicated in the statement of the
proposition - for Alpha, Beta and Gamma we get the expected return
chosen by the investor as below,
\begin{equation}
\left[w_{\alpha}E\left(R_{\alpha}\right)\right]+\left[w_{\beta}E\left(R_{\beta}\right)\right]+\left[w_{\gamma}E\left(R_{\gamma}\right)\right]=\left[\left(w_{c}\right)\left(w_{c\alpha}\right)E\left(R_{\alpha}\right)\right]+\left[\left(w_{c}\right)\left(w_{c\beta}\right)E\left(R_{\beta}\right)\right]+\left[\left(w_{\gamma}\right)E\left(R_{\gamma}\right)\right]
\end{equation}
\begin{align}
 & =\left[\left(\frac{\left[E\left(R_{i}\right)-E\left(R_{\gamma}\right)\right]}{\left[E\left(R_{c}\right)-E\left(R_{\gamma}\right)\right]}\right)\left(w_{c\alpha}\right)E\left(R_{\alpha}\right)\right]\\
 & +\left[\left(\frac{\left[E\left(R_{i}\right)-E\left(R_{\gamma}\right)\right]}{\left[E\left(R_{c}\right)-E\left(R_{\gamma}\right)\right]}\right)\left(w_{c\beta}\right)E\left(R_{\beta}\right)\right]\\
 & +\left[\left(\frac{\left[E\left(R_{i}\right)-E\left(R_{c}\right)\right]}{\left[E\left(R_{\gamma}\right)-E\left(R_{c}\right)\right]}\right)E\left(R_{\gamma}\right)\right]
\end{align}
\begin{align}
 & =\left[\left(\frac{\left[E\left(R_{i}\right)\left(w_{c\alpha}\right)E\left(R_{\alpha}\right)-E\left(R_{\gamma}\right)\left(w_{c\alpha}\right)E\left(R_{\alpha}\right)\right]}{\left[E\left(R_{c}\right)-E\left(R_{\gamma}\right)\right]}\right)\right]\\
 & +\left[\left(\frac{\left[E\left(R_{i}\right)\left(w_{c\beta}\right)E\left(R_{\beta}\right)-E\left(R_{\gamma}\right)\left(w_{c\beta}\right)E\left(R_{\beta}\right)\right]}{\left[E\left(R_{c}\right)-E\left(R_{\gamma}\right)\right]}\right)\right]\\
 & -\left[\left(\frac{\left[E\left(R_{i}\right)-\left[w_{c\alpha}E\left(R_{\alpha}\right)\right]-\left[w_{c\beta}E\left(R_{\beta}\right)\right]\right]}{\left[E\left(R_{c}\right)-E\left(R_{\gamma}\right)\right]}\right)E\left(R_{\gamma}\right)\right]
\end{align}
\begin{align}
 & =\left[\left(\frac{\left[E\left(R_{i}\right)\left(w_{c\alpha}\right)E\left(R_{\alpha}\right)-E\left(R_{\gamma}\right)\left(w_{c\alpha}\right)E\left(R_{\alpha}\right)\right]}{\left[E\left(R_{c}\right)-E\left(R_{\gamma}\right)\right]}\right)\right]\\
 & +\left[\left(\frac{\left[E\left(R_{i}\right)\left(w_{c\beta}\right)E\left(R_{\beta}\right)-E\left(R_{\gamma}\right)\left(w_{c\beta}\right)E\left(R_{\beta}\right)\right]}{\left[E\left(R_{c}\right)-E\left(R_{\gamma}\right)\right]}\right)\right]\\
 & -\left[\left(\frac{\left[E\left(R_{i}\right)E\left(R_{\gamma}\right)-\left[E\left(R_{\gamma}\right)w_{c\alpha}E\left(R_{\alpha}\right)\right]-\left[E\left(R_{\gamma}\right)w_{c\beta}E\left(R_{\beta}\right)\right]\right]}{\left[E\left(R_{c}\right)-E\left(R_{\gamma}\right)\right]}\right)\right]
\end{align}
\begin{align}
 & =\left[\left(\frac{\left[E\left(R_{i}\right)\left(w_{c\alpha}\right)E\left(R_{\alpha}\right)\right]}{\left[E\left(R_{c}\right)-E\left(R_{\gamma}\right)\right]}\right)\right]\\
 & +\left[\left(\frac{\left[E\left(R_{i}\right)\left(w_{c\beta}\right)E\left(R_{\beta}\right)\right]}{\left[E\left(R_{c}\right)-E\left(R_{\gamma}\right)\right]}\right)\right]\\
 & -\left[\left(\frac{\left[E\left(R_{i}\right)E\left(R_{\gamma}\right)\right]}{\left[E\left(R_{c}\right)-E\left(R_{\gamma}\right)\right]}\right)\right]
\end{align}
\begin{align}
 & =\left[E\left(R_{i}\right)\left(\frac{\left[\left(w_{c\alpha}\right)E\left(R_{\alpha}\right)+\left(w_{c\beta}\right)E\left(R_{\beta}\right)-E\left(R_{\gamma}\right)\right]}{\left[E\left(R_{c}\right)-E\left(R_{\gamma}\right)\right]}\right)\right]
\end{align}
\begin{align}
 & =\left[E\left(R_{i}\right)\left(\frac{\left[E\left(R_{c}\right)-E\left(R_{\gamma}\right)\right]}{\left[E\left(R_{c}\right)-E\left(R_{\gamma}\right)\right]}\right)\right]
\end{align}
\begin{align}
 & =E\left(R_{i}\right)
\end{align}
\end{proof}

\section{\label{sec:Return-and-Risk}Return and Risk (Volatility) Calculations}

The weight calculation depends on calculating the return of asset
prices on a daily basis (could be at other frequencies as well, but
we start with daily returns). The return will be the continuously
compounded measure. The volatility will be the standard deviation
of the continuously compounded returns over a historical time period,
generally around 90 days. The volatility will need to be calculated
on a rolling 90 days basis. 
\begin{enumerate}
\item \label{enu:Continously-Compounded-Return,}Continuously Compounded
Return, $R_{t}$, at any time $t$ is given by the logarithm of the
ratio of the Price at time $t$, $P_{t}$, and the Price at time $t-1$,
$P_{t-1}$, as shown below,
\begin{equation}
R_{t}=\ln\left(\frac{P_{t}}{P_{t-1}}\right)
\end{equation}
\item \label{enu:Volatility-at-time}Volatility at time $t$, $\sigma_{t}$
is calculated as below using the average return, $\bar{R}_{t}$, at
time $t$ and then calculating the standard deviation. It can be done
directly if suitable libraries are available. $\sigma_{t}^{2}$ is
the variance at time $t$. Here, $T=90$ unless otherwise stated.
\begin{equation}
\bar{R}_{t}=\left(\frac{1}{T}\right)\sum_{i=t}^{t-T}R_{i}\label{eq:Return-Estimate}
\end{equation}
\begin{equation}
\sigma_{t}^{2}=\left(\frac{1}{T-1}\right)\sum_{i=t}^{t-T}\left(R_{i}-\bar{R}_{t}\right)^{2}
\end{equation}
\begin{equation}
\sigma_{t}=\sqrt{\left(\frac{1}{T-1}\right)\sum_{i=t}^{t-T}\left(R_{i}-\bar{R}_{t}\right)^{2}}\label{eq:Volatility-Risk-Estimate}
\end{equation}
\item Calculate the volatility for the longest historical time period possible
for each asset. Possibly for the last 360 days or longer. Some new
assets might have a relatively shorter price time series, and it is
okay to calculate volatility for those assets for the number of days
for which prices are available.
\end{enumerate}

\section{\label{sec:Appendix-of-Illustrations}Appendix of Additional Illustrations}

\begin{figure}[H]
\includegraphics[width=18cm]{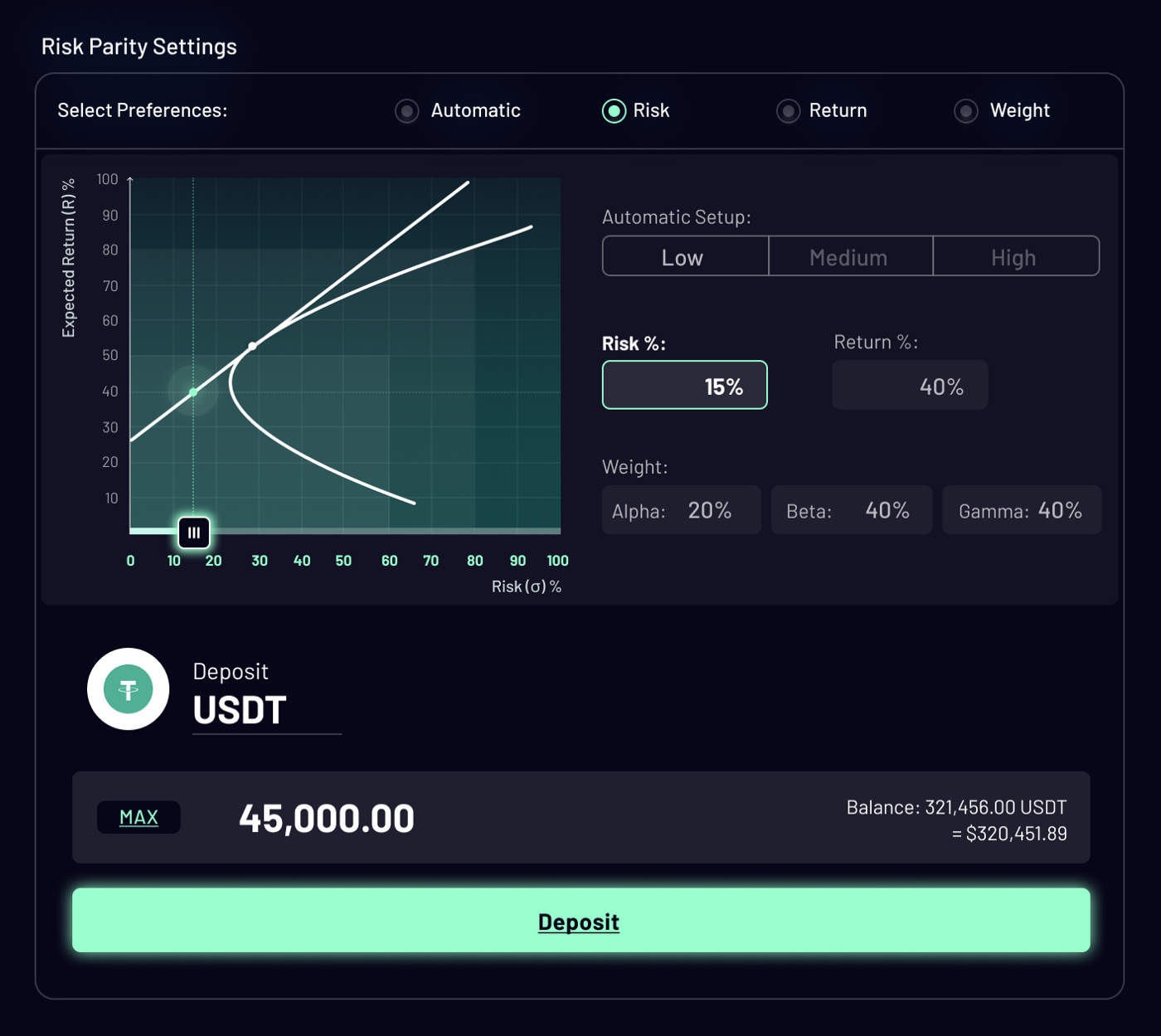}\caption{\label{fig:Deposit-Screen-Parity}Deposit Screen for Parity}
\end{figure}

\begin{figure}[H]
\includegraphics[width=18cm]{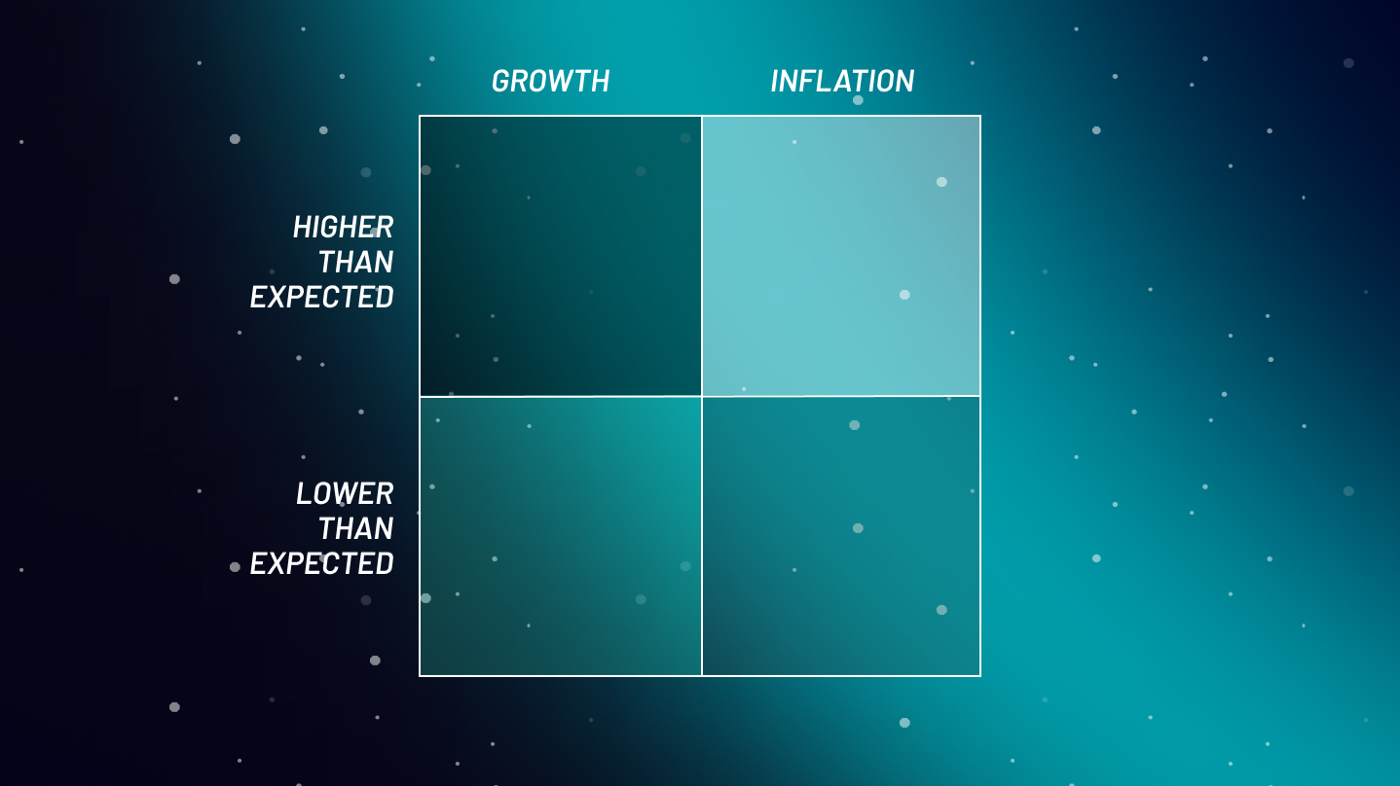}

\caption{\label{fig:Classification-of-Environmental}Classification of Environmental
Risks}
\end{figure}

\begin{figure}[H]
\includegraphics[width=18cm]{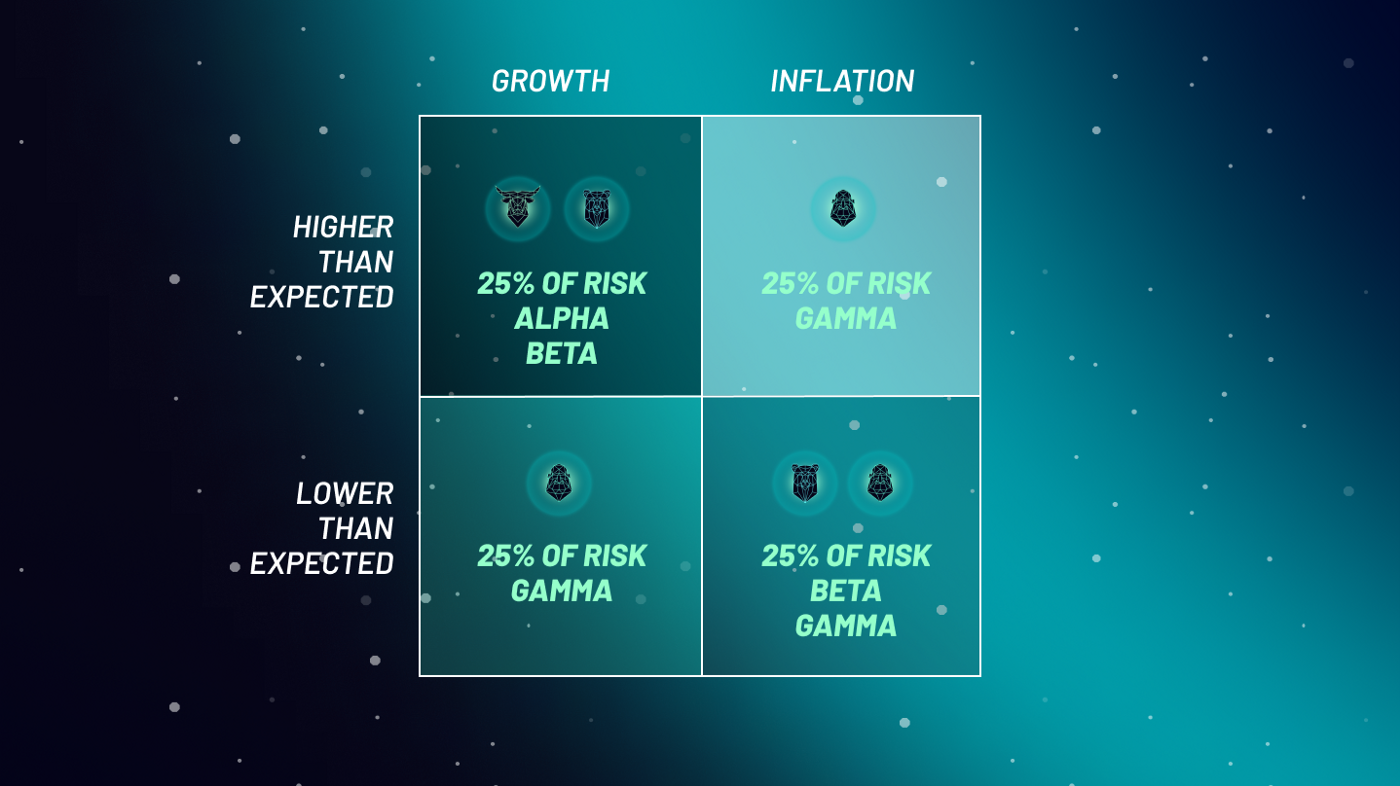}\caption{\label{fig:Portfolio-Risk-Concentrations}Portfolio Risk Concentrations}
\end{figure}
\begin{figure}[H]
\includegraphics[width=18cm]{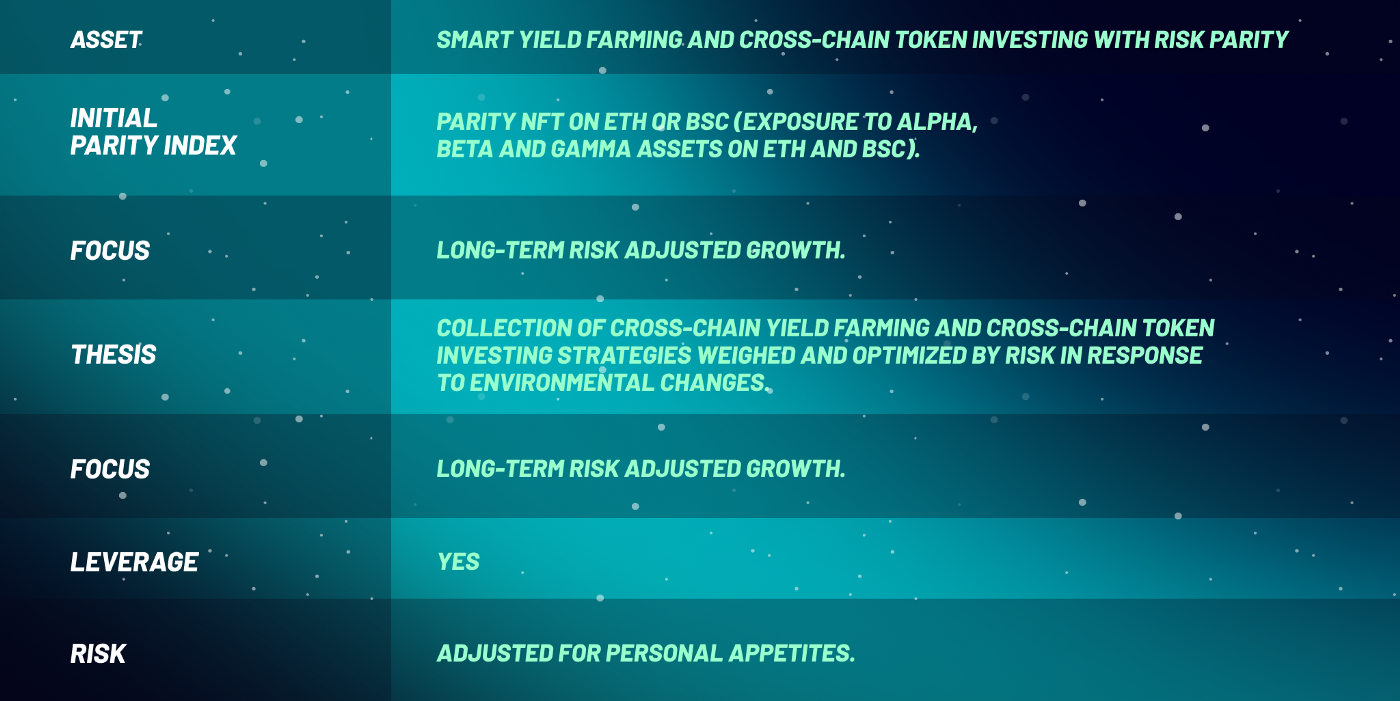}

\caption{\label{fig:Risk-Parity-Highlights}Risk Parity Highlights}
\end{figure}

\end{document}